%% file: main.tex
\documentclass[11pt,a4paper]{article}
\usepackage[utf8]{inputenc}
\usepackage[T1]{fontenc}
\usepackage{lineno}
\usepackage{amsmath}
\usepackage{amsthm}
\usepackage{amssymb}
\usepackage[margin = 1in]{geometry}
\usepackage{thmtools}
\usepackage{multicol}
\usepackage{algorithm}
\usepackage[noend]{algpseudocode}
\usepackage[shortlabels]{enumitem}
\usepackage{comment}
\usepackage{tabularx}
\usepackage{makecell}
\usepackage[]{hyperref}
\usepackage{graphicx}
\usepackage{subfig}
\usepackage{color}
\usepackage{cleveref}
\usepackage{svg}
\usepackage{textcomp}

\colorlet{DarkRed}{red!50!black}
\colorlet{DarkGreen}{green!50!black}
\colorlet{DarkBlue}{blue!50!black}

\hypersetup{
	linkcolor = DarkRed,
	citecolor = DarkGreen,
	urlcolor = DarkBlue,
	bookmarksnumbered = true,
	linktocpage = true, 
}

\usepackage{authblk}

\newtheorem{theorem}{Theorem}
\newtheorem{lemma}[theorem]{Lemma}

\newtheorem{definition}[theorem]{Definition}

\newtheorem{observation}[theorem]{Observation}

\newcommand{\eps}{\varepsilon}
\newcommand{\paren}[1]{\mathopen{}\left(#1\right)\mathclose{}}
\newcommand{\ceil}[1]{\mathopen{}\lceil#1\rceil\mathclose{}}
\newcommand{\floor}[1]{\mathopen{}\lfloor#1\rfloor\mathclose{}}
\newcommand{\poly}{\operatorname{\text{{\rm poly}}}}
\newcommand{\T}{\mathcal T}
\newcommand{\MC}{\mathcal M \mathcal C}
\renewcommand{\P}{\mathcal P}
\newcommand{\R}{\mathbb R}
\renewcommand{\a}{\overline{a}}
\newcommand{\E}{\mathbb E}
\newcommand{\Q}{\mathcal Q}
\DeclareMathOperator{\lev}{level}

\usepackage[
	style=alphabetic,
    maxalphanames=4,
    minalphanames=4,
    maxnames=20,
    minnames=20,
    backref=true,
	doi=true,
    backend=bibtex,
]{biblatex}
\usepackage[disable]{todonotes}

    \usepackage{etoolbox} %% <- for \cspreto, \csappto
    
    \newcommand*\linenomathpatch[1]{%
      \cspreto{#1}{\linenomath}%
      \cspreto{#1*}{\linenomath}%
      \csappto{end#1}{\endlinenomath}%
      \csappto{end#1*}{\endlinenomath}%
    }
    
    \linenomathpatch{equation}
    \linenomathpatch{gather}
    \linenomathpatch{multline}
    \linenomathpatch{align}
    \linenomathpatch{alignat}
    \linenomathpatch{flalign}

\title{Tree-Packing Revisited: Faster Fully Dynamic Min-Cut and Arboricity}
\author{Tijn de Vos\thanks{TU Graz. This work was partially conducted when the author was a PhD student at the University of Salzburg. This research was funded in whole or in part by the Austrian Science Fund (FWF) \url{https://doi.org/10.55776/P36280}, \url{https://doi.org/10.55776/I6915}, and \url{https://doi.org/10.55776/P32863}. For open access purposes, the author has applied a CC BY public copyright license to any author-accepted manuscript version arising from this submission.  This project has received funding from the European Research Council (ERC) under the European Union's Horizon 2020 research and innovation programme (grant agreement No~947702).}
\and \textcircled{r}\thanks{The author ordering was randomized using \url{https://www.aeaweb.org/journals/policies/random-author-order/} generator. It is requested that citations of this work list the authors separated by \texttt{\textbackslash textcircled\{r\}} instead of commas.} \and 
Aleksander B. G. Christiansen\thanks{Technical University of Denmark. This work was supported by VILLUM FONDEN grant 37507 ``Efficient Recomputations for Changeful Problems''.}}
\date{}

\begin{document}

\begin{titlepage}
\maketitle

\begin{abstract}

\input{abstract}

\end{abstract}
\thispagestyle{empty}
\newpage
\thispagestyle{empty}

\tableofcontents
 \newpage
  \listoftodos
\end{titlepage}

\newpage
\input{intro_new}

\newpage
\input{mincut}

\newpage
\input{arboricity}

\newpage
\input{lowerbound}

\newpage
\input{existence}

\section*{Acknowledgements}
We thank Pavel Arkhipov for pointing out two typos in earlier versions. 

\newpage
\printbibliography[heading=bibintoc]
% \bibliography{references}{}
% \bibliographystyle{plain}

\end{document}

%% file: abstract.tex
Tree-packings -- collections of spanning trees of a graph -- are a fundamental tool in the study of minimum cut and related graph parameters. They have played a central role in the design of algorithms across static, dynamic, and distributed settings. In this paper, we study both tree-packings themselves and their structural connections to min-cut and arboricity. Our results lead to faster dynamic algorithms for both problems.

For dynamic min-cut, [Thorup, Comb.\ 2007] used tree-packings to obtain his dynamic min-cut algorithm with $\tilde O(\lambda^{14.5}\sqrt{n})$ worst-case update time.
We reexamine this relationship, showing that we need to maintain fewer trees for such a result; we show that we only need to pack $\Theta(\lambda^3 \log m)$ greedy trees to guarantee either a 1-respecting cut or a trivial cut in some contracted graph.

Based on this structural result, we then provide a deterministic algorithm for fully dynamic exact min-cut that has $\tilde O(\lambda^{5.5}\sqrt{n})$ worst-case update time, for graphs with min-cut value at most $\lambda$. 
In particular, this also yields an algorithm for fully dynamic exact min-cut with $\tilde O(m^{1-1/12})$ amortized update time, improving upon $\tilde O(m^{1-1/31})$ [Goranci et al., SODA 2023].

We also give the first fully dynamic algorithm that maintains a $(1+\varepsilon)$-approximation of the fractional arboricity.  Our algorithm is deterministic and has $O(\alpha \log^6m/\varepsilon^4)$ amortized update time, for graphs with arboricity at most $\alpha$. 
We extend these results to a Monte Carlo algorithm with $O(\poly(\log m,\varepsilon^{-1}))$ amortized update time against an adaptive adversary. Our algorithms work on multi-graphs as well. 

Our structural results on tree-packing also include a lower bound for greedy tree-packing, which -- to the best of our knowledge -- is the first progress on this topic since~[Thorup, Comb.\ 2007].

%% file: intro_new.tex
\section{Introduction}
A \emph{tree-packing} $\T$ is a collection of spanning trees of a graph. 
Often one is interested in a tree-packing satisfying certain requirements, e.g., greedy, disjoint, which we detail where applicable. 
%sometimes the spanning trees are required to be disjoint, sometimes one lets the weight of an edge be the number of trees an edge belongs to, and require that the spanning trees form successive minimum spanning trees. In this case, the tree-packing is \emph{greedy}. 
%Sometimes these weights are scaled by the total number of trees, and one seeks a tree-packing that minimises the maximum weight or maximises the minimum weight. 
%
In particular, tree-packing first appeared in the seminal works of Nash-Williams~\cite{Nash61} and Tutte~\cite{Tutte61}. 
They studied sufficient and necessary conditions for when a graph can be decomposed into $k$ disjoint spanning subgraphs, i.e., when a graph admits a disjoint tree-packing of size $k$. 
They provided a sufficient and necessary condition by considering partition values: 
given a partition $\mathcal{P}$ of the vertex set of a graph $G$, the \emph{value} of $\mathcal{P}$ is then $\tfrac{|E(G/\mathcal{P})|}{|V(G/\mathcal{P})|-1}$ , where $G/\mathcal{P}$ is the graph obtained by contracting $\mathcal{P}$ in $G$. 
In particular, they showed that $G$ admits a decomposition into $k$ disjoint spanning trees if and only if the minimum partition value is at least $k$. 
Not long after, Nash-Williams noted~\cite{NashWilliams64} that his techniques extend to answer a question with a very similar flavor: what is the smallest number of trees needed to cover all edges of a graph?
This number is now known as the \emph{arboricity} of the graph, and Nash-Williams showed that the arboricity is the ceiling of the \emph{fractional arboricity} defined as $\alpha := \max_{S\subseteq V} \frac{|E(S)|}{|S|-1}$.

Subsequently, tree-packing has also been studied for its algorithmic applications, in particular for computing the value of a minimum cut and related problems~\cite{gabow1995matroid,karger2000minimum,ThorupK00,NaorR01,Thorup07,Thorup08,ChekuriQ019,DBLP:conf/stoc/DagaHNS19,ChekuriQX20,DoryEMN21}. 
The \emph{minimum cut value} of a graph is defined as the minimum number of edges whose deletion causes the graph to become disconnected. 
Such a collection of edges forms a \emph{minimum cut} or \emph{min-cut} for short\footnote{As per convention, we often simply write min-cut to refer to the size/value of a min-cut. Note that there can be multiple min-cuts, although there is a unique min-cut value.}. 
In many of these applications, \emph{ideal relative loads} play a central role. The ideal relative loads (formally defined in \Cref{sc:intro_TP}) correspond to some sort of `ideal tree-packing'. 
We are still far from fully understanding ideal relative loads and their limits.
In this paper, we hope to contribute to a better understanding. 
In particular, we study ideal relative loads by giving new upper and lower bounds on the size of tree-packings that approximate the ideal relative loads. 
Moreover, we show two structural results on ideal relative loads and min-cut/arboricity respectively.  
This leads to a faster fully dynamic exact min-cut algorithm and the first fully dynamic algorithm to maintain a $(1+\eps)$-approximation of the fractional arboricity. 
As such, we categorize our results in three parts: 1) tree-packings and ideal relative loads, 2) min-cut, and 3) arboricity. 

\subsection{Tree-Packings and Ideal Relative Loads}\label{sc:intro_TP}
% \paragraph{Definitions.}
Let $G=(V,E)$ be an unweighted undirected graph. A \emph{tree-packing} $\T$ of $G$ is a family of spanning trees, where edges can appear in multiple trees. The \emph{load} of an edge $e$, denoted by $L^{\T}(e)$, is defined by the number of trees that contain $e$. The \emph{relative load} is defined as $\ell^{\T}(e) = L^{\T}(e)/|\T|$. Whenever the tree-packing is clear from the context, we omit the superscript $\cdot^{\T}$. The \emph{packing value} of a tree-packing $\T$ is

\begin{equation*}
    \rm{pack\_val}(\T) := \frac{1}{\max_{e\in E} \ell^{\T}(e)}.
\end{equation*}
Dual to tree-packing we have a concept for partitions. For a partition $\P$ of the vertex set $V$, we define the \emph{partition value} as

\begin{equation*}
    \rm{part\_val}(\P) := \frac{|E(G/\P)|}{|\P|-1}.
\end{equation*}
We now have that (see e.g.~\cite{Nash61,Tutte61})

\begin{equation*}
    \Phi_G:= \max_{\T} \rm{pack\_val}(\T) = \min_{\P} \rm{part\_val}(\P).
\end{equation*}
We omit the subscript and simply write $\Phi$ when the graph is clear from the context.

Next, we introduce \emph{ideal relative loads}, following Thorup~\cite{Thorup07}. These loads, denoted by $\ell^*(e)$, are defined recursively. 
\begin{enumerate}
    \item Let $\P^*$ be a partition with $\rm{part\_val}(\P^*)=\Phi$.
    \item For all $e\in E(G/\P^*)$, set $\ell^*(e):= 1/\Phi$.
    \item For each $S\in \P^*$, recurse on the subgraph $G[S]$. 
\end{enumerate}

By definition, we have that $\Phi_G =\tfrac{1}{\max_{e\in E}\ell^*(e)}$. We show an analogous result for the fractional arboricity $\alpha(G)$ and the minimum ideal relative load:
%the following relation between arboricity and tree-packing. 
\begin{restatable}{theorem}{ArbStructural}\label{thm:arb_vs_packing}
    Let $G=(V,E)$ be an undirected, unweighted (multi-)graph, then 
    \begin{equation*}
        \alpha(G) = \frac{1}{\min_{e\in E} \ell^{*}(e)}.
    \end{equation*}
\end{restatable}
In \Cref{sc:intro_arb}, we demonstrate how this leads to novel dynamic algorithms to approximate the fractional arboricity. 

We obtain a second structural result, which regards tree-packing and the min-cut. We see a first connection between these two concepts as follows: let $\lambda$ denote the min-cut of $G$, then $\lambda/2 < \Phi \le \lambda$~\cite{karger2000minimum}. Combining this fact with $\Phi =\tfrac{1}{\max_{e\in E}\ell^*(e)}$, we see that the min-cut relates to the ideal relative loads. We investigate this relation further in \Cref{sc:intro_mincut}, where we state our somewhat more technical result.

Before we discuss these two structural results, we consider how to approximately compute ideal relative loads. 

\paragraph{Approximating Ideal Relative Loads.}
Computing the ideal relative loads exactly seems to be a challenging question. Giving a tree-packing such that its relative loads approximate the ideal relative loads seems easier. More precisely, the goal is to find a tree-packing $\mathcal T$ such that
\begin{equation}
    |\ell^\T(e)-\ell^*(e)| \leq \eps/\lambda,\label{eq:TP}
\end{equation}
for all $e\in E$. In applications, we need for algorithmic efficiency that $\mathcal T$ is \emph{small}. 
In general, such a tree-packing needs $\lambda/\eps$ trees. We prove that there always exists a tree-packing of this size. 

\begin{restatable}{theorem}{TPexistence}\label{thm:existence}
    Let $G$ be an unweighted, undirected (multi-)graph. There exists a tree-packing $\T$ with $|\T|=\Theta(\lambda/\eps)$ trees that satisfies
    \begin{equation*}
        |\ell^\T(e)-\ell^*(e)| \leq \eps/\lambda,
    \end{equation*}
    for all $e\in E$.
\end{restatable}
Although our algorithm is constructive, it is not efficient. Instead, we consider \emph{greedy tree-packings}.

\paragraph{Greedy Tree-Packing.}
Thorup~\cite{Thorup07} showed that a \emph{greedy tree-packing} approximates the ideal packing well. 
Here a greedy tree-packing is defined as follows: let the weight of an edge be the number of trees an edge belongs to, and require that the spanning trees in the packing form successive minimum weight spanning trees. 
Thorup showed that a greedy tree-packing $\T$ with $|\T|\geq 6\lambda \log m/\eps^2$ trees\footnote{Throughout this paper we write $\log m$ for all logarithmic factors. We note that for simple graphs this simplifies to $\log n$.} satisfies \Cref{eq:TP}.
% 
%     \begin{equation}
%         |\ell^\T(e)-\ell^*(e)| \leq \eps/\lambda, %\label{eq:thorup_TP}
%     \end{equation}
% for all $e\in E$. 
The fact that greedy tree-packings approximate ideal packings will be the basis for our dynamic min-cut and arboricity algorithms, since greedy tree-packings are easy to compute -- even in a dynamic setting.  

% \paragraph{Number of Greedy Trees.}
To make any algorithm based on \Cref{eq:TP} as efficient as possible, it is worth to study how many greedy trees are necessary to attain it. We establish a novel lower bound. 

\begin{restatable}{theorem}{TPlowerbound}\label{thm:TP_lower_bound}
    Let $G$ be an unweighted, undirected (multi-)graph. In general, a greedy tree-packing $\T$ needs $|\T|=\Omega(\lambda/\eps^{{3}/{2}})$ trees to satisfy
    \begin{equation*}
        |\ell^\T(e)-\ell^*(e)| \leq \eps/\lambda,
    \end{equation*}
    for all $e\in E$, whenever $\eps^{-1} = O(n^{1/3})$.
\end{restatable}

Although a greedy tree-packing is easy to compute and maintain, this lower bound combined with \Cref{thm:existence} indicates that investigating other tree-packings might be worthwhile.

\subsection{Min-Cut}\label{sc:intro_mincut}
In a seminal paper, Gabow~\cite{gabow1995matroid} showed how to compute both a packing of disjoint trees and a min-cut. Subsequently, Karger~\cite{karger2000minimum} used the fact that any large enough greedy tree-packing contains trees that $2$-respect some min-cut, i.e., it contains a spanning tree which crosses some min-cut at most twice. 
To arrive at his near-linear time min-cut algorithm for computing a min-cut, Karger used this fact by showing that one can efficiently compute the size of all cuts that $2$-respects a tree.
This observation has also been proven useful for computing $k$-cuts~\cite{Thorup08,ChekuriQ019}. 

This `semi-duality' between tree-packing and min-cut has additionally found applications in other models of computation such as distributed computing~\cite{DBLP:conf/stoc/DagaHNS19,DoryEMN21} and dynamic algorithms~\cite{ThorupK00,Thorup07}.
Dynamic algorithms maintain a solution to a problem -- for instance the size of a min-cut -- as the input graph undergoes deletions and insertions of edges. 
In this setting, it is not known how to maintain the smallest $2$-respecting cut in a dynamic forest. Instead, Thorup~\cite{Thorup07} showed that packing $|\T|=\Omega(\lambda^7 \log^{3} m)$ greedy trees is sufficient to guarantee that at least one tree in $\T$ $1$-respects a min-cut, provided that the size of the min-cut is at most $\lambda$. 
He then shows that one can dynamically maintain this packing in $\tilde O(|\T|^2 \sqrt{\lambda n})$ update time\footnote{We write $\tilde O(f)$ for $O(f\poly\log f)$.} and that one can maintain the smallest $1$-respecting cut efficiently in a dynamic forest. These things combine to give an exact dynamic min-cut algorithm running in $\tilde O(\lambda^{14.5}\sqrt{n})$ worst-case update time, whenever the min-cut has size at most $\lambda$. 

This dependency on $\lambda$ is persistent. Even the cases $\lambda = 1$ and $\lambda = 2$ are important and very well-studied in the dynamic setting~\cite{Frederickson85,GalilI91,WestbrookT92,EppsteinGIN97, Frederickson97, henzingerK1997fully,  HenzingerK99, HolmLT01,KapronKM13,NanongkaiSW17,HolmRT18,ChuzhoyGLNPS20}. Recently Jin, Sun, and Thorup~\cite{jin2024fully} gave an algorithm for the case that $\lambda\leq (\log m)^{o(1)}$. For general (polynomial) values of $\lambda$, Thorup~\cite{Thorup07} remains the state-of-the-art. 
One immediate way of improving the dependency on $\lambda$ in his update time is to show that packing even fewer trees still guarantees that at least one tree $1$-respects a min-cut. 
This prompted Thorup to ask if it is possible to show that an even smaller tree-packing always $1$-respects at least one min-cut.%, which is still unanswered.
% In recent work, Goranci et al.~\cite{GoranciHNSTW23} balanced the approach of Thorup with a different approach based on expander decompositions to get an exact dynamic algorithm for all values of $\lambda$. Again, an improvement in the $\lambda$ dependency immediately yields a faster algorithm. %In a similar vein to Thorup, they asked whether this is possible. 

\paragraph{Our Structural Result.}
In this paper, we show that a better $\lambda$ dependency is possible. 
We do not attain this by improving directly on Thorup's bound that  $\Omega(\lambda^7 \log m)$ greedy trees suffice. 
Instead, we find that in the dynamic setting, focusing exclusively on $1$-respecting cuts may be too limiting.
In particular, we prove that in the cases where one might need to pack many trees to ensure that at least one packed tree $1$-respects a min-cut, one can instead consider a corresponding approximate partition. We show that whenever we cannot guarantee a $1$-respecting min-cut, we can instead guarantee that the approximate partition contains a \emph{trivial} min-cut. Here, a trivial min-cut is a cut where one side consists of a single vertex.  
% We then make this approach algorithmic by showing how to maintain the trivial min-cuts of this approximate partition. 
With this approach, we need to pack only $\Omega(\lambda^3 \log m)$ trees, thus resulting in a much smaller final dependency on $\lambda$.

% As mentioned above, Thorup~\cite{Thorup07} shows that if we pack $\Omega(\lambda^7 \log m)$ greedy trees, we pack at least one tree that crosses some min-cut only once. 
% %This argument heavily relies on the relation between greedy trees and the ideal relative loads \Cref{eq:thorup_TP}. 
% Instead of directly improving upon this bound, i.e., showing that one can pack fewer trees, we investigate the relationship between ideal relative loads and min-cut more closely, obtaining the following result. 

\begin{restatable}{theorem}{thmCutExistence}\label{thm:CutExistence}
    Let $G$ be an unweighted, undirected (multi-)graph, and let $\T$ be a greedy tree-packing on $G$. Suppose $|\mathcal T|=\Omega(\lambda^3\log m)$, then at least one of the following holds 
    \begin{enumerate}
        \item Some $T\in \T$ $1$-respects a min-cut of $G$; or \label{item:1resp}
        \item Some trivial cut in $G/ \{e\in E: \ell^\T(e)<\tfrac{2}{\lambda}-\tfrac{1}{2\lambda^2}\}$ corresponds to a min-cut of $G$. \label{item:trivial_cut}
    \end{enumerate}    
\end{restatable}

\paragraph{Our Dynamic Results.}
% \subsection{Min-Cut}
% The minimum cut, or min-cut, of a graph is the smallest set of edges such that removing them makes the graph disconnected. We refer to the min-cut value as the number of edges in a min-cut. Note that there can be multiple min-cuts in a graph, although there is a unique min-cut value. 
We combine the structural result from \Cref{thm:CutExistence} with novel dynamic subroutines to obtain the following result. 
% This result is the main technical ingredient for \Cref{thm:mincut_par}, which also contains novel dynamic routines for maintaining both parts.
% We obtain the following deterministic result, which hence also naturally holds against an adaptive adversary.  
\begin{restatable}{theorem}{mincutPar}\label{thm:mincut_par}
     There exists a deterministic dynamic algorithm that, given an unweighted, undirected (multi-)graph $G=(V,E)$, maintains the exact min-cut value $\lambda$ if $\lambda\leq \lambda_{\max}$ in $\tilde O(\lambda_{\max}^{5.5}\sqrt{n})$ worst-case update time.\\
     It can return the edges of the cut in $O(\lambda\log m)$ time with  $\tilde O(\lambda_{\max}^5\sqrt{m})$ worst-case update time.     
\end{restatable}
This improves on the state-of-the-art by Thorup~\cite{Thorup07}, who achieves $\tilde O(\lambda_{\max}^{14.5}\sqrt{n})$ worst-case update time. Thorup also uses this result to obtain $(1+\eps)$-approximate min-cut against an oblivious adversary in $\tilde O(\sqrt{n})$ time. In this setting, our result improves the polylogarithmic factors. An application where our result has more impact is deterministic exact min-cut for unbounded min-cut value~$\lambda$. 

\begin{restatable}{theorem}{mincutCombi}
    There exists a deterministic dynamic algorithm that, given a simple, unweighted, undirected graph $G=(V,E)$, maintains the exact min-cut value $\lambda$ with amortized update time
    \[
        \tilde O(\min\{m^{1-1/12}, m^{11/13}n^{1/13},n^{3/2}\}).
    \]
    %\tijn{can also give $\tilde O(\lambda^{11/13}n^{12/13})$ update time}
    It can return the edges of the cut in $O(\lambda\log m)$ time with $\tilde O(m^{1-1/12})$ amortized update time. 
\end{restatable}
We obtain this result by using the algorithm of Goranci et al.~\cite{GoranciHNSTW23} for the high $\lambda$ regime. When they combine this with~\cite{Thorup07}, they achieve $\tilde O(m^{29/31}n^{1/31})=\tilde O(m^{1-1/31})$ amortized update time\footnote{See the updated arXiv version for the correct bounds. The SODA version of the paper states a different bound (namely $\tilde O(m^{1-1/16})$), which resulted from a typo when citing~\cite{Thorup07}.}. 
We note that Goranci et al.~\cite{GoranciHNSTW23} also provide a \emph{randomized} result with $\tilde O(n)$ worst-case update time against an adaptive adversary.

\subsection{Arboricity}\label{sc:intro_arb}
As mentioned, the \emph{fractional arboricity} $\alpha$ of a graph is defined as 
\begin{align*}
    \alpha := \max_{S\subseteq V} \frac{|E(S)|}{|S|-1}.
\end{align*}
The \emph{(integral) arboricity} is defined as $\lceil \alpha\rceil$, so it is strictly easier to compute. Equivalently, we can define the arboricity as the minimum number of trees needed to cover the graph~\cite{Nash61}. 
As stated in \Cref{thm:arb_vs_packing}, we prove that the arboricity can be expressed in terms of the ideal relative loads. 
We use this equality to maintain an approximation of the arboricity in the dynamic setting. 

\paragraph{Our Dynamic Results.}
We obtain the first dynamic $(1+\eps)$-approximation of the fractional arboricity. Our algorithm is deterministic, hence also works naturally against an adaptive adversary. 

\begin{restatable}{theorem}{arboricityDetNew}\label{thm:arb_det}
    There exists a deterministic dynamic algorithm that, given an unweighted, undirected (multi-)graph $G=(V,E)$, maintains a $(1+\eps)$-approximation of the fractional arboricity $\alpha$ when $\alpha\leq \alpha_{\max}$ in $O(\alpha_{\max}\log^6 m / \eps^4)$ amortized update time or a Las Vegas algorithm with $O(\alpha_{\max}^2 m^{o(1)} / \eps^4)$ worst-case update time.
\end{restatable}
%\tijn{we could add that we can output the subgraph maximizing this}
This improves the approximation factor in the state-of-the-art: Chekuri et al.~\cite{chekuri2024adaptive} give a $(2+\eps)$-approximation of the fractional arboricity. It has $O(\log \alpha_{\max} \log^2 m  /\eps^4)$ amortized update time or $O(\log \alpha_{\max}\log^3 m  /\eps^6)$ worst-case update time. In fact, for simple graphs, the value of the densest subgraph is a $(1+\eps)$-approximation for large values of $\alpha$. Combining this with our result for smaller values of $\alpha$, we obtain the following result. 

\begin{restatable}{theorem}{arboricitySimple}\label{thm:arb_simple}
    There exists a deterministic dynamic algorithm that, given a simple, unweighted, undirected graph $G=(V,E)$, maintains a $(1+\eps)$-approximation of the fractional arboricity $\alpha$ in $O(\log^6 m / \eps^5)$ amortized update time or a Las Vegas algorithm with $O(m^{o(1)}/\eps^6)$ worst-case update time. 
\end{restatable}

For multi-graphs, the value of the densest subgraph remains a $2$-approximation of the fractional arboricity, even when $\alpha$ becomes large. To obtain an efficient algorithm for $(1+\eps)$-approximation, we use a sampling technique to reduce the case of large $\alpha$ to the case of small $\alpha$. Although this is rather straightforward against an oblivious adversary, we need to construct a more sophisticated scheme to deal with an adaptive adversary.  

\begin{restatable}{theorem}{arboricity}\label{thm:arb}
    There exists a dynamic algorithm that, given an unweighted, undirected (multi-)graph $G=(V,E)$, maintains a $(1+\eps)$-approximation of the fractional arboricity $\alpha$ against an adaptive adversary in $O( \log^{12} m / \eps^{15})$ amortized update time or $O(m^{o(1)} / \eps^{19})$ worst-case update time.  
\end{restatable}

Against an oblivious adversary we obtain the improved amortized update time of  $O(\log^{7} m /\eps^{6})$, or worst-case update time $O(m^{o(1)}/\eps^{8})$, see \Cref{thm:arb_obl}.

We remark that the update time of many dynamic algorithms is parameterized by the arboricity~\cite{BrodalF99,BernsteinS15,BernsteinS16,chekuri2024adaptive,KopelowitzKPS13,NeimanS16,PelegS16}. This shows that the arboricity is an important graph parameter. These algorithms are faster when the arboricity is small, and this is exactly the regime where dynamic approximations of the arboricity are lacking. We hope that our results contribute to a better understanding of this.

\subsection{Further Related Work.} 

\paragraph{Tree-Packing.}
Tree-packing can be formulated as a (packing) LP. One way of solving LPs is via the \emph{Multiplicative Weights Update (MWU) Method}, see e.g., \cite{AroraHK12}. In particular, this method needs at most $O(\rho\log m/\eps^2)$ iterations for a packing/covering LP with width~$\rho$~\cite{PlotkinST95}. For tree-packing, we have width $\rho=\Theta(\lambda)$. 

The MWU method and \emph{greedy} tree-packing are strongly related. In particular, Harb, Quanrud, and Chekuri~\cite{HarbQC22} note that if one uses (a particular version of) the MWU method then the resulting algorithm would also be greedy tree-packing. Chekuri, Quanrud, and Xu~\cite{ChekuriQX20} conjecture that greedy tree-packing is in fact equivalent to approximating the LP via the standard MWU method. 

% General packing LPs have an iteration count lower bound of $\Omega(\rho\log m/\eps^2)$ iterations~\cite{KleinY15}. It is unclear whether this also holds for the special case of tree-packing.  

We note that in both greedy tree-packing and in the MWU method, there is still freedom in the update step: in greedy tree-packing there are often multiple minimum spanning trees (e.g., the first tree can be \emph{any} spanning tree) and in the MWU method we can select any of the violated constraints. 
%So, although we do not expect \emph{every} greedy tree-packing to perform better, we might still be able to show that a specific greedy tree-packing beats the general lower bound. 
For example, we could select a minimum spanning tree at random or select a violated constraint at random. Such adaptation of the standard approach might lead to better greedy tree-packing or MWU algorithms.  

Tree-packing is also used for minimum $k$-cut \cite{NaorR01, ChekuriQX20}. It is an interesting open question whether our techniques extend to this setting.

\paragraph{Min-Cut.}
Jin and Sun~\cite{JinS21} gave an algorithm for $s-t$ min-cut, for min-cut up to size $(\log m)^{o(1)}$ with $n^{o(1)}$ worst-case update time. Recently, they generalized these techniques to get $n^{o(1)}$ worst-case update time for the global min-cut~\cite{jin2024fully} up to size $(\log m)^{o(1)}$. 
There is a further line of work on smaller values of the min-cut: in particular for graph connectivity (whether $\lambda\geq 1$)~\cite{Frederickson85,HenzingerK99, HolmLT01,KapronKM13,NanongkaiSW17,ChuzhoyGLNPS20} and 2-edge connectivity (whether $\lambda \geq 2$)~\cite{WestbrookT92, GalilI91,EppsteinGIN97, Frederickson97, henzingerK1997fully, HolmLT01,HolmRT18}. 

For planar graphs, Lacki and Sankowski~\cite{LackiS11} provided a deterministic algorithm with $\tilde O(n^{5/6})$ worst-case update and query time. 

There are faster algorithms for dynamic min-cut, at the cost of an approximation factor: Thorup and Karger~\cite{ThorupK00} gave a $(2+\eps)$-approximation in $\tilde O(\poly(\log m,\eps^{-1}))$ amortized update time, and Thorup~\cite{Thorup07} gave a $(1+\eps)$-approximation in $\tilde O(\sqrt{n})$ worst-case update time. 

In the partially dynamic setting, Thorup~\cite{Thorup07} gives an algorithm with $\tilde O(n^{3/2}+m)$ total time for the purely decremental and incremental settings. Further, there are incremental algorithms with $O(\lambda \log m)$ or $O(\poly\log m )$ amortized update time, by Henzinger~\cite{Henzinger97} and Goranci, Henzinger, and Thorup~\cite{GoranciHT18} respectively.

\paragraph{Arboricity.}
Computing the (fractional) arboricity of a graph is a relatively difficult problem; we do not know a static linear-time algorithm for the exact version. 
Gabow showed how to compute the arboricity in $\tilde O(m^{3/2})$ time~\cite{gabow1995matroid,Gabow98}. For the approximate version, much faster algorithms are known. 
In particular, Eppstein gave a $2$-approximation in $O(m+n)$ time~\cite{Eppstein94}. 

Plotkin, Shmoys, and Tardos gave a FPTAS for solving fractional packing and fractional covering problems, and their algorithm applied to fractional arboricity takes $\tilde O(m\alpha/\eps^2)$ time~\cite{PlotkinST95} and provides a $(1+\eps)$-approximation, where $\alpha$ denotes the value of the fractional arboricity. 
Toko, Worou, and Galtier gave a $(1+\eps)$-approximation for the fractional arboricity in $ O(m\log^3 m/\eps^2)$ time~\cite{WorouG16}. 
Blumenstock and Fischer~\cite{BlumenstockF20} gave a $(1+\eps)$-approximation of the arboricity in $O(m\log m\log\alpha/\eps)$ time. More recently, Quanrud~\cite{quanrud2024faster} gave a $(1+\eps)$-approximation of the arboricity w.h.p.\ in $O(m\log m\alpha(n) + n\log m(\log m + (\log\log m+\log(1/\eps))\alpha(n))/\eps^3)$ time, where $\alpha(n)$ is the inverse Ackermann\footnote{We apologize for the abuse of notation; in the remainder of the paper the inverse Ackermann does not appear and $\alpha$ always refers to the fractional arboricity itself.}.

The exact fractional arboricity can be computed in $\tilde O(mn)$ time~\cite{Gabow98}.

% \tijn{Joachim et al \cite{BlikstadMNT23} actually don't improve state of the art}

For dynamic arboricity, there is a deterministic exact algorithm with worst-case update time $\tilde O(m)$~\cite{BanerjeeRS20}.
The work of Brodal and Fagerberg~\cite{BrodalF99} implies a $4$-approximation of the arboricity in $O(\poly\log m)$ amortized update time against an oblivious adversary.

Both for better approximations and for the fractional arboricity, we turn to the closely related \emph{densest subgraph problem}, where we want to determine the density $\rho$:
\[
    \rho := \max_{S\subseteq V} \frac{|E(S)|}{|S|}.
\]

There have been multiple dynamic approximate densest subgraph algorithms over the last years~\cite{EpastoLS15,BhattacharyaHNT15,SawlaniW20}.  
The state-of-the-art~\cite{chekuri2024adaptive} gives a $(1+\eps)$-approximation, which in turn implies a $(2+\eps)$-approximate fractional arboricity algorithm (for simple graphs). It has $O(\log^2 m \log \rho /\eps^4)$ amortized update time or $O(\log^3 m \log \rho /\eps^6)$ worst-case update time. 

The algorithms via densest subgraph, but also \cite{BrodalF99}, use out-orientations to compute the approximate fractional arboricity. Although this works well for crude approximations, it does not seem to lead to high precision approximations.

\subsection{Technical Overview}
In this section, we give an overview of our techniques, starting with the results whose proofs are most extensive. However, we first define the model our dynamic results are part of.

\paragraph{Dynamic Graph Algorithms.}
 The goal is to maintain a data structure for a changing graph, that maintains the solution to some problem, e.g., the value of the min-cut or the arboricity. The input graph undergoes a series of updates, which are either edge insertions or edge deletions. If the sequence of updates is fixed from the start, we say we have an \emph{oblivious} adversary. If the sequence of updates can be based upon the algorithm or the state of the data structure, we say we have an \emph{adaptive} adversary. We say the algorithm has \emph{amortized update time} equal to~$t$, if it spends $\sigma t$~time for a series of $\sigma$ updates. We say the algorithm has \emph{worst-case update time} equal to~$t$ if it spends at most $t$~time after every update.

In some cases, the data structure does not store the solution to the problem explicitly, but it can be retrieved by a query, e.g., in case of the \emph{edges} of a min-cut. In this case we state the required time to answer such a query. 

\subsubsection{Min-Cut}
To obtain our dynamic min-cut result, \Cref{thm:mincut_par}, we have three technical contributions. The first is to prove \Cref{thm:CutExistence}. Our algorithm then consists of maintaining such a tree-packing, its 1-respecting cuts, and the corresponding trivial cuts. Our second contribution is to show how to maintain such trivial cuts, and the third is a new technique to decrease the recourse of tree-packing. 

\paragraph{Min-Cut and Tree-Packing.}
First, let us sketch the proof of \Cref{thm:CutExistence}. We remark that $\lambda$ and $\Phi$ only differ by a constant factor, which we will use throughout. 
\begin{restatable}{lemma}{lemphilambda}\label{lm:phi_lambda}
    $\lambda/2 < \Phi \leq \lambda$.
\end{restatable}

The main observation that allows one to show that some tree $1$-respects a min-cut is that if some min-cut $C$ has 
\[
\sum_{e \in C} \ell^{\T}(e) < 2,
\]
then it must be $1$-respected by some tree. 
Indeed, then the average number of times a tree in $\T$ crosses $C$ is 
\[
\frac{1}{|\T|}\sum_{e \in C} L^{\T}(e) = \sum_{e \in C} \ell^{\T}(e) < 2,
\]
so some tree must $1$-respect $C$. 
Hence, if $\sum_{e \in C} \ell^{*}(e)$ is small enough, i.e., far enough below $2$, then the tree-packing does not need to be very large to also concentrate the sum $\sum_{e \in C} \ell^{\T}(e)$ below $2$.

Hence, we can restrict ourselves to the case where every min-cut $C$ has $\sum_{e \in C} \ell^{*}(e) \approx 2$.
In this case, one can show that in fact every edge $e$ participating in a min-cut has $\ell^{*}(e) \geq \bar{a} \approx \tfrac{2}{\lambda}$ for some appropriately chosen $\bar{a}$. 
This in turn means that for a different value $\hat{a}$ only slightly smaller than $\bar{a}$, the graph $G/\{e \in E: \ell^{*}(e) < \hat{a}\}$ must contain trivial min-cuts.
To show this we assume this is not the case for contradiction, and then count the edges of the graph in two different ways: once using the $\ell^{*}$ values and once by summing the degrees. 
The first way of counting implies that $G/\{e \in E: \ell^{*}(e) < \hat{a}\}$ contains roughly $\hat{a}^{-1} |V(G/\{e \in E: \ell^{*}(e) < \hat{a}\})| \approx \tfrac{\lambda}{2} |V(G/\{e \in E: \ell^{*}(e) < \hat{a}\})|$ edges, and the second way of counting shows via the Hand-Shaking Lemma that $G/\{e \in E: \ell^{*}(e) < \hat{a}\}$ contains at least $\tfrac{\lambda+1}{2} |V(G/\{e \in E: \ell^{*}(e) < \hat{a}\})|$ -- a contradiction. 

Now that we know a feasible $\hat a$ exists, we need to show that we do not need a very large tree packing~$\T$ to obtain the result. Here, the next step is to show that, for an arbitrary trivial min-cut $X$ of $G/\{e \in E: \ell^{*}(e) < \hat{a}\}$, $\Phi_{G[X]}$ is sufficiently large.
We claim that this finishes the proof, because this implies that we do not need to pack too many trees in order to concentrate $\ell^{\T}$ values of the edges in a min-cut far above the $\ell^{\T}$ values of edges in $G[X]$. 
In particular, one can ensure that all edges $e$ in a min-cut have $\ell^{\T}(e) > \hat{a}$ and all edges in $G[X]$ has $\ell^{\T}(e) < \hat{a}$. 
This then implies that $X$ is a trivial min-cut of $G/\{e \in E: \ell^{\T}(e) < \hat{a}\}$. 

To show this, we lower bound the value of $\Phi_{G[X]}$ via the objective value of an optimization problem, which encodes the size of every trivial cut in the partition $\mathcal{P}$ realizing $\Phi_{G[X]}$. 
The key property needed is that no min-cut edges are in $E(G[X])$ and so very few trivial cuts of $G[X]/\mathcal{P}$ have size smaller than $\lambda + 1$. 
This allows us to bound the objective value of the optimization problem significantly above $\tfrac{\lambda}{2}$, which in turn implies that all $\ell^{*}$ values of edges in $G[X]$ are significantly below $\tfrac{2}{\lambda}$, thus yielding the required property. For more details, see \Cref{sc:cut_existence}.

\paragraph{Maintaining Trivial Cuts.}
Second, we describe how we maintain the trivial cuts efficiently. 
The key challenge in maintaining the size of the trivial cuts, is that we do not know how to maintain an explicit representation of $G/\{e \in E: \ell^{\T}(e) < \hat{a}\}$. 
The difficult part is that dynamically maintaining contractions and un-contractions explicitly can be very time consuming, as we might need to look at many edges in order to assign them the correct endpoints after the operation is performed. 
To circumvent this issue, we maintain a weaker, implicit representation of $G/\{e \in E: \ell^{\T}(e) < \hat{a}\}$. 
Instead of maintaining the graph explicitly, we only maintain the vertices of $G/\{e \in E: \ell^{\T}(e) < \hat{a}\}$ as well as their degrees. 
This weaker representation turns out to be much easier to maintain.

We maintain a dynamic connectivity data structure on the graph induced by all edges $e$ with $\ell^{\T}(e) < \hat{a}$. 
The connected components of this graph correspond exactly to the vertices of $G/\{e \in E: \ell^{\T}(e) < \hat{a}\}$. 
In order to calculate the degree of these vertices, we maintain the number of edges $e$ with $\ell^{\T}(e) \geq \hat{a}$ incident to each vertex $v \in V(G)$. 
Since the connectivity data structure supplies us with a connectivity witness in the form of a spanning forest, we can maintain a spanning tree of each connected component as a top tree~\cite{AlstrupHLT05}.
By weighting each vertex $v \in V(G)$ by the number of incident edges $e$ with $\ell^{\T}(e) \geq \hat{a}$, we can use the top trees to dynamically maintain the sum of vertex weights of each tree. 
Since loops are counted twice, the sum of vertex weights in a spanning tree corresponds exactly to the degree of the corresponding vertex in $G/\{e \in E: \ell^{\T}(e) < \hat{a}\}$. 
By using a min-heap, we can then maintain the minimum degree of $G/\{e \in E: \ell^{\T}(e) < \hat{a}\}$, which can only over-estimate the size of the corresponding trivial cut. 
To report the min-cut, we can then retrieve the cut by searching each vertex with weight strictly greater than $0$ in only $\log m$ time, and then report each edge with $\ell^{\T}(e) \geq \hat{a}$ incident to the vertex in $\log m$ time per edge. For more details, see \Cref{sc:trival_cuts}.

\paragraph{Recourse in Tree-Packing.}
Third, we turn to the tree-packing. We note that the most expensive step in our algorithm is to maintain $1$-respecting cuts on $|\T|=\Theta(\lambda_{\max}^3 \log m)$ minimum spanning trees. Every time such a spanning tree changes, we need to update its $1$-respecting cuts in $\tilde O(\sqrt{n})$ time. Therefore, it is important to bound the \emph{recourse} in these spanning trees: the number of changes to the tree-packing due to an edge update. Naively, we have $O(|\T|^2)=O(\lambda_{\max}^6\log^2m)$ recourse (see e.g.,~\cite{ThorupK00}), to see this consider an edge deletion. This edge is contained in at most $|\T|$ trees. In each of these trees it needs to be replaced, which leads to a weight increase for the replacement edge. This change needs to be propagated in all following trees, so can lead to a series of $|\T|$ changes. We show that we can get away with a recourse of $O(\lambda_{\max}^5\log^2 m)$. See \Cref{lm:TP_recourse_impr} for a formal statement. This improvement is independent of the improvement from \Cref{thm:CutExistence} in the number of trees we need to pack, and might have applications in other instances where a greedy tree-packing is used. 

The first observation is that if the graph has min-cut $\lambda$, then an edge will be in roughly $|\T|/\lambda$ trees, not in all $|\T|$. If $\lambda=\lambda_{\max}$ throughout the update sequence, then this gives the result. However, $\lambda$ can become arbitrary small. To mitigate this fact, we keep $\log \lambda_{\max}$ copies of the tree-packing, $\T_1, \T_2, \dots$, where $\T_i$ corresponds to the case where the min-cut value is in $[2^i,2^{i+1})$. If $\lambda \geq 2^i$, then we have $\T_i$ as a normal tree-packing on $G=(V,E)$. However, if $\lambda< 2^i$, then $\T_i$ is a tree-packing on $(V,E\cup E_{\rm{virtual}})$, where $E_{\rm{virtual}}$ is some set of virtual edges ensuring that the min-cut stays large. These edges get added as $\lambda$ shrinks and deleted as $\lambda$ grows. 
 
When using this global view of the connectivity, we obtain an algorithm with $O(\lambda_{\max}^5\log^2 m)$ amortized recourse. However, by inspecting the $\ell^\T$ values, we have a much more refined view of the connectivity. With a careful analysis, this means we can obtain the same bound on the worst-case recourse. This comes down to deleting a virtual edge when it is no longer needed, and not waiting until the global min-cut has increased. We refer to \Cref{sc:TP_recourse} for the details. 

\subsubsection{Arboricity}
We revisit the relationship between ideal tree-packing and the partition values. We recall that~\cite{Nash61,Tutte61}
\[
    \Phi_G:= \max_{\T} \rm{pack\_val}(\T) = \min_{\P} \rm{part\_val}(\P), 
\]
which Thorup~\cite{Thorup07} used to recursively define the ideal relative loads $\ell^*$. By a close inspection of these definitions, we prove \Cref{thm:arb_vs_packing}:
\[
   \alpha(G) = \frac{1}{\min_{e\in E} \ell^{*}(e)}.
\]
This means that we can simply estimate $\alpha(G)$ by $(\min \ell^\T(e))^{-1}$, with $\ell^T$ some good approximation of $\ell^*$. We recall that by packing $|\T|=\Theta(\lambda \log m/\eps^2)$ greedy spanning trees, we have that
\begin{equation*}
    |\ell^\T(e)-\ell^*(e)| \leq \eps/\lambda, 
\end{equation*}
for all $e\in E$~\cite{Thorup07}.
By maintaining this many greedy spanning trees through dynamic minimum spanning tree algorithms, we obtain update time $\sim \alpha_{\max}^4$, see \Cref{sc:arb_wUB_warm_up}. Similar to the tree-packing for min-cut, we can alter the graph. This means we have an artificially high min-cut of $\lambda=\Theta(\alpha)$, this both decreases the number of trees we need to pack and decreases the recourse  (and hence update time) to $\sim \alpha_{\max}$, see \Cref{sc:arb_recourse}. It is not a direct application of the same result, since there are even disconnected graphs with high arboricity. We make an adaptation such that we can leverage high min-cut anyway. 

% We conjecture that packing $|\T|=\Theta(\lambda \log m/\eps^2)$ greedy spanning trees gives
% \[
%     |\ell^\T(e)-\ell^*(e)| \leq \eps \ell^*(e). 
% \]
% If this is true, then the dependency on $\alpha_{\max}$ in the update time becomes linear. However, to get rid of $\alpha_{\max}$ altogether, we need some additional techniques. 

\paragraph{Simple Graphs.}
First, we consider simple graphs, to obtain a result independent of $\alpha_{\max}$. Suppose $S\subseteq V$ satisfies $\frac{|E(S)|}{|S|-1}=\alpha$. Since in simple graphs we have $|E(S)|\leq |S|(|S|-1)$, we see that 
\[ 
    |S| = \frac{|S|(|S|-1)}{|S|-1} \geq \frac{|E(S)|}{|S|-1}=\alpha. 
\]
This means that for large values of $\alpha=\Omega(1/\eps)$, we have that $1/|S|\approx 1/(|S|-1)$, hence we have $\rho =\max_{S\subseteq V} \tfrac{|E(S)|}{|S|} \approx \max_{S\subseteq V} \tfrac{|E(S)|}{|S|-1}=\alpha$. We combine our algorithm for bounded $\alpha_{\max}$, \Cref{thm:arb_det}, with an efficient algorithm to compute a $(1+\eps)$-approximation of the density~$\rho$~\cite{SawlaniW20,chekuri2024adaptive} to obtain \Cref{thm:arb_simple}. See \Cref{sc:arb_simple} for more details.

\paragraph{Multi-Graphs.}
For multi-graphs, the same argument does not hold: even in subgraphs of size $|S|=2$ we can have many parallel edges, so the density is only a $2$-approximation of the arboricity, even for large values of $\alpha$. Instead, we use a sampling approach to reduce the case of large $\alpha$ to small $\alpha$. The idea is simple (see also e.g., \cite{McGregorTVV15}): if we sample every edge with probability $\Theta(\tfrac{\log m}{\alpha \eps^2})$, we should obtain a graph with arboricity $\Theta(\tfrac{\log m}{\eps^2})$. By maintaining $\log m$ copies with guesses for $\alpha=2,4,8,\dots$, we can find the arboricity. We describe in \Cref{sc:arb_multi_obl} how this gives an algorithm against an oblivious adversary. 

However, an adaptive adversary poses a problem: such an adversary can for example delete sampled edges to skew the outcome. Of course it is too costly to resample all edges after every update. A first idea is to only resample the edges incident to an updated edge. Although this would give us the needed probabilistic guarantee, it has high update time: a vertex can have degree in the sampled graph as big as $n/\alpha$. To combat this, we work with ownership of edges: each edge belongs to one of its endpoints. When an edge is updated, we resample all edges of its owner. When we assign arbitrary owners, this does not give us any guarantees on the degree yet. We remark that ownership can be seen as an orientation of the edge, where the edge is oriented away from the owner. Now we can guarantee that each vertex owns at most $(1+\eps)\alpha$ edges by using an out-orientation algorithm~\cite{chekuri2024adaptive}. 
%In fact, we show that each time we sample the edges of a vertex, it gives the right result with high probability. Hence a polynomial number of attacks cannot skew the outcome. 
For more details, see \Cref{sc:arb_multi_adap}.

\subsubsection{Lower Bound for Greedy Tree-Packing}
To show that a greedy tree-packing needs $\Omega(\lambda/\eps^{3/2})$ trees to satisfy $|\ell^\T(e)-\ell^*(e)|\leq \eps/\lambda$ (\Cref{thm:TP_lower_bound}), we give a family of graphs, together with a tree-packing, such that if we pack $o(\lambda/\eps^{3/2})$ trees, then $|\ell^\T(e)-\ell^*(e)|> \eps/\lambda$ for some edge $e\in E$. First, we restrict to the case $\lambda=\Phi=2$. We give a graph that is the union of two spanning trees, see \Cref{fig:modified_pair_intro}. The vertical string in the graph consists of $k$ vertices, the circular part of the remaining $n-k$ vertices. We show that we can over-pack edges at the top of the vertical string (edge $a$ in \Cref{fig:modified_pair_intro}), at the cost of under-packing edges at the bottom of the vertical string (edge $b$ in \Cref{fig:modified_pair_intro}). The load of two neighboring edges can differ by at most $1$, since the packing is greedy. We prove that by packing $\Theta(k^3)$ trees, we can get the optimal difference of $\Theta(k)$ between the highest and lowest edge. Since both of them are supposed to have a value of $\ell^*(a)=\ell^*(b)=1/\Phi$, we obtain an error $|\ell^T(a)-\ell^*(a)|> \Theta(k/k^3)=\Theta(1/k^2)$. By setting $k=\Omega(\eps^{-1/2})$, we obtain the result.  

    \begin{figure}%
        \centering
        \subfloat{{\includegraphics[width=4cm]{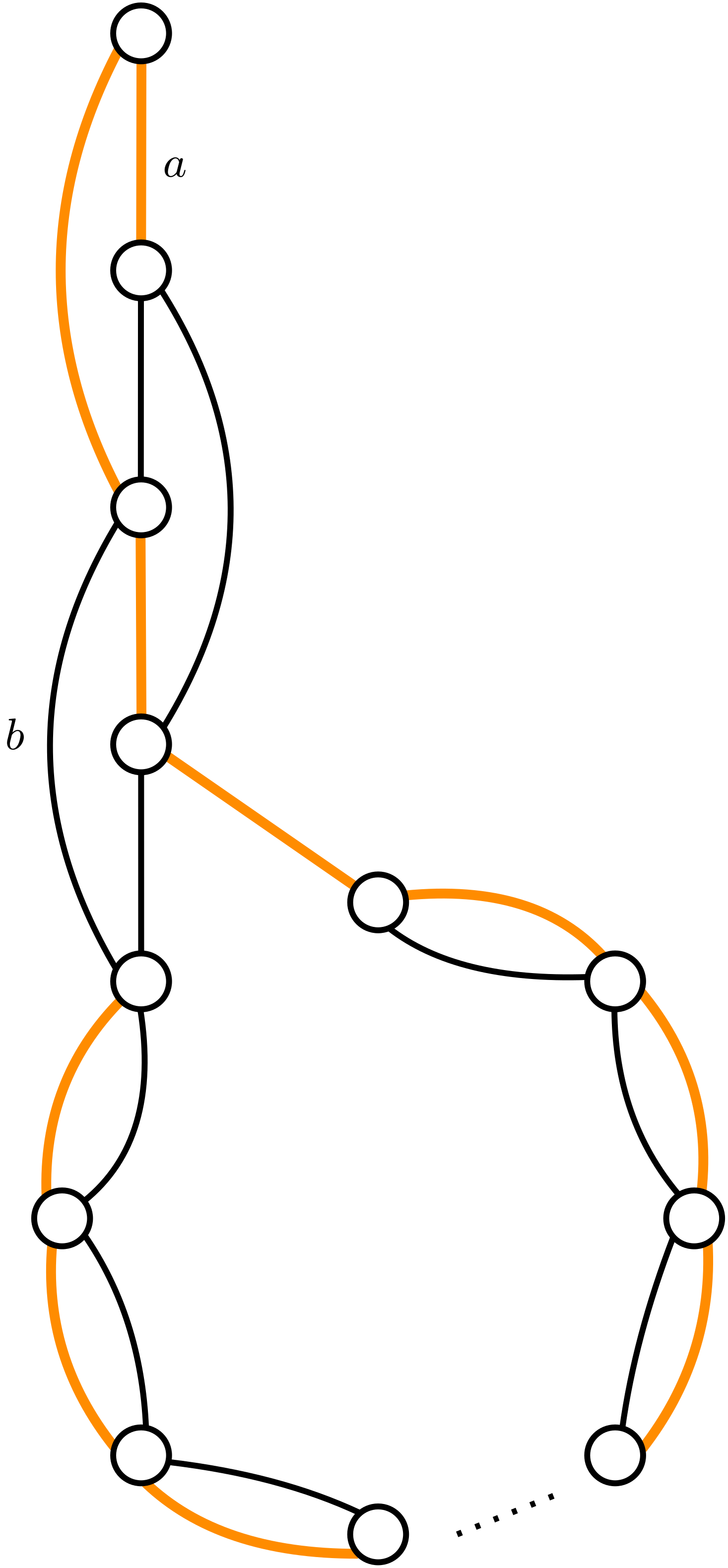} }}%
        \hspace{30mm}%
        \subfloat{{\includegraphics[width=4cm]{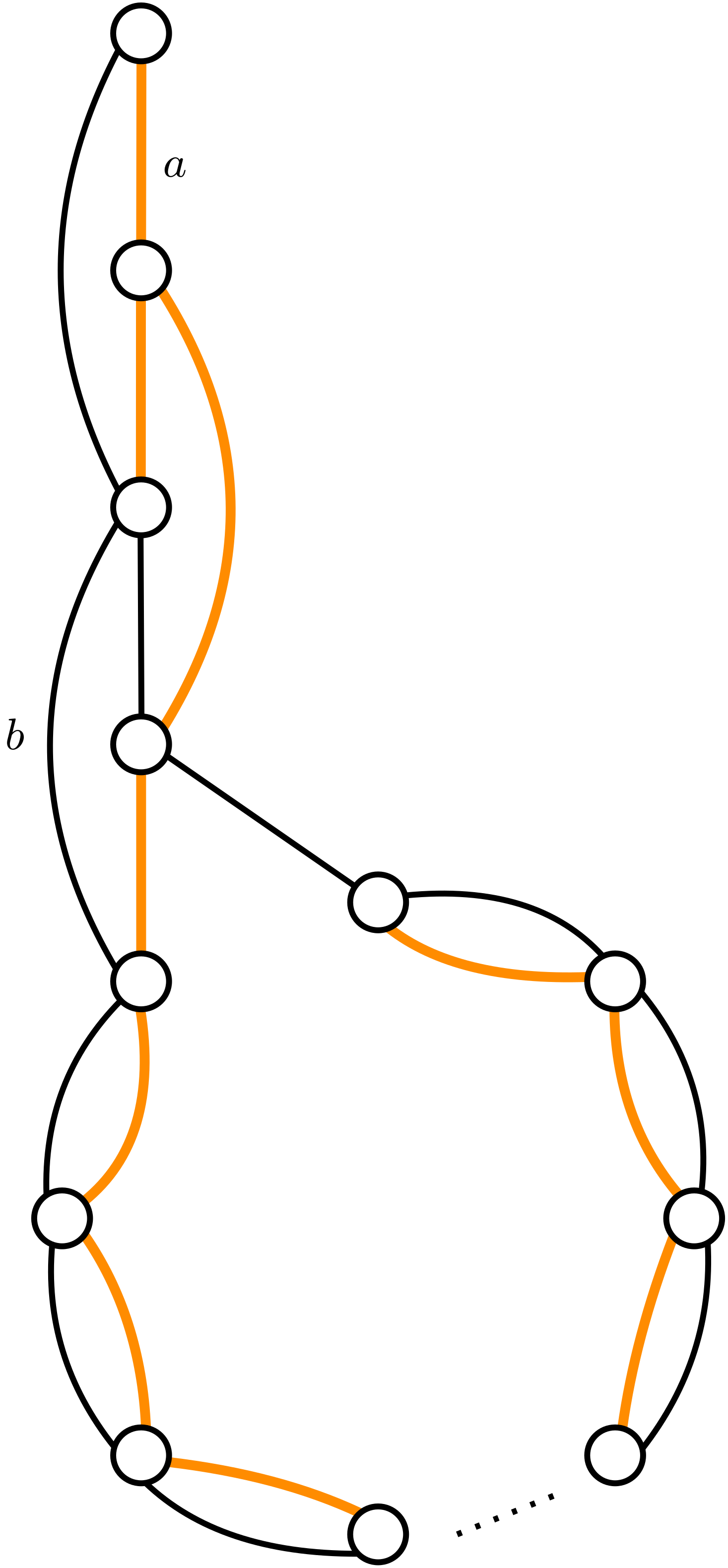} }}%
        \caption{The graph with two trees of the tree-packing colored. Observe that all edges are packed once in these two trees except for $a$, which is packed twice, and $b$ which is never packed. }
        \label{fig:modified_pair_intro}%
    \end{figure}

We generalize this result for any even $\lambda$, by copying every edge in the above construction $\lambda/2$ times and packing trees on each copy in parallel. For more details, we refer to \Cref{sc:LB}.

\subsubsection{Existence of Small Tree-Packing}
Our last technical contribution is \Cref{thm:existence}, showing that there exists a small tree-packing. We first consider the case that $\ell^*(e)=1/\Phi$ for every edge $e\in E$. The proof is rather simple, and based on Kaiser's~\cite{Kaiser12} elegant proof of the tree-packing theorem. This is the theorem initially proved by Tutte~\cite{Tutte61} and Nash-Williams~\cite{Nash61} that shows that $\Phi$ (in our notation) is well defined. Another phrasing is as follows: a graph $G$ contains $k$ pairwise disjoint spanning trees if and only if for every partition $\P$ of $V(G)$, the graph $G/\P$ has at least $k(|\P|-1)$ edges. We generalize this to $k$ trees + 1 forest and show that extending this forest to an arbitrary spanning tree gives the required packing. 

Next, we use the ideal load decomposition to generalize this to any graph. The trees in our packing are simply unions of the trees on each component. For more details, see \Cref{sc:existence}.

%% file: mincut.tex
\section{Min-Cut}
In this section, we reexamine the relation between min-cut and tree-packing. 
The main technical contribution is the following theorem.

\thmCutExistence*

Thorup~\cite{Thorup07} showed that if $|\T|=\Omega(\lambda^7\log^3m)$, then some $T\in \T$ $1$-respects a min-cut. With this new result, we significantly decrease the number of trees we need to pack, leading to a significant speed-up.

This section is organized as follows. In \Cref{sc:cut_existence}, we give a proof of \Cref{thm:CutExistence}. In \Cref{sc:trival_cuts}, we show how to estimate the trivial cuts of part \ref{item:trivial_cut} in \Cref{thm:CutExistence} efficiently. In \Cref{sc:TP_recourse}, we show how to decrease the recourse in the tree-packing, which leads to a faster running time. In \Cref{sc:mincut_par}, we provide our dynamic min-cut algorithm for bounded $\lambda$. And finally, in \Cref{sc:mincut_gen}, we provide our dynamic min-cut algorithm for general $\lambda$. 

\paragraph{Preliminaries.}
Before we move on to the proof of \Cref{thm:CutExistence}, we first make the theorem statement more precise, and we make the notation more concise. 

For some set of edges $E'\subseteq E$, we denote $G/E'$ for the graph where we contract all edges in $E'$. Suppose $uv\in E\setminus E'$ and $u$ and $v$ are contracted in $G/E'$, then we keep $uv$ as a self-loop in $G/E'$. Such edges count \emph{twice} towards the degree of the corresponding vertex in $G/E'$. 

Let $E_{\circ a}^*:= \{e\in E : \ell^*(e) \circ a\}$  and $E_{\circ a}^\T:= \{e\in E : \ell^\T(e) \circ a\}$ for $\circ \in \{\geq, >, \leq <, =\}$. Further let $G_a := G/E^\T_{<a}$. Then in our new notation, part \ref{item:trivial_cut} of \Cref{thm:CutExistence} corresponds to a trivial cut in $G_a=G/E^\T_{<a}$ for $a=\tfrac{2}{\lambda}-\tfrac{1}{2\lambda^2}$.

Let $\MC\subseteq E$ denote the set of all edges that are contained in at least one min-cut. And let $\overline{a}$ denote the largest value in $\R$ such that $E^*_{\geq \overline{a}}\supseteq \MC$.

We recall the definition of $\Phi$. 
\begin{lemma}[\cite{Nash61,Tutte61}] \label{lma:NashTutte}
We have
    $$\Phi_G:= \max_{\T} \rm{pack\_val}(\T) = \min_{\P} \rm{part\_val}(\P).$$
\end{lemma}

We will repeatedly use the following simple lemma concerning $\Phi$ and the min-cut $\lambda$ as mentioned in the introduction. We include a proof for completeness. 
\lemphilambda*
% \begin{lemma}[\cite{karger2000minimum}]\label{lm:phi_lambda}
%     $\lambda/2 < \Phi \leq \lambda$.
% \end{lemma}
\begin{proof}
    We recall that $\Phi = \min_{\P} \tfrac{|E(G/\P)|}{|\P|-1}$, this immediately gives that 
    \[\Phi \leq \min_{\P=\{A,B\}} \frac{|E(A,B)|}{2-1}= \lambda.\]
    Now consider an arbitrary partition $\P$. We note that $2\tfrac{|E(G/\P)|}{|\P|}$ is the average value of a trivial cut in $G/ \P$. Since this corresponds to a cut in $G$, we get  $2\tfrac{|E(G/\P)|}{|\P|}\geq \lambda$. Now we see
    \[
        \Phi \geq \frac{|E(G/\P)|}{|\P|-1} > \frac{|E(G/\P)|}{|\P|} \geq \lambda/2. %\qedhere
    \]
\end{proof}

We have the following lemma showing that a sufficiently large greedy tree-packing approximates the ideal packing, introduced in \Cref{sc:intro_TP}, well. 
\begin{lemma}[\cite{Thorup07}]\label{lm:tree_packing_thorup}
    A greedy tree-packing $\T$ with $|\T|\geq 6\lambda \log m/\eta^2$ trees, for $\eta<2$, satisfies
    \begin{align*}
        |\ell^\T(e)-\ell^*(e)| \leq \eta/\lambda,
    \end{align*}
    for all $e\in E$.
\end{lemma}

\subsection{A Proof of Theorem~\ref{thm:CutExistence}} \label{sc:cut_existence}
Let us first briefly discuss the intuition behind the proof. We start by using reasoning similar to arguments appearing in Thorup~\cite{Thorup07} to show that if some edge $e' \in \MC$ has $\ell^{*}(e')$ small enough, then any min-cut containing $e'$ will be crossed once by some tree even if we only pack a relatively small amount of trees. This follows from the observation that if 
\[
\sum_{e \in C} \ell^{\T}(e) < 2
\]
for some cut $C$, then some tree crosses $C$ at most once. Indeed, the average number of times $C$ is crossed by a tree is given by
\[
\frac{1}{|\T|}\sum_{e \in C} L^{\T}(e) = \sum_{e \in C} \ell^{\T}(e)
\]
and at least one tree will cross $C$ at most the average number of times. Hence, if $\ell^{*}(e')$ is small enough, i.e., far enough below $\frac{2}{\lambda}$, then it is sufficient to pack $O(\lambda^3 \log m)$ trees in order to concentrate $\sum_{e \in C} \ell^{\T}(e)$ below $2$. It is important to note here that in order to get this dependency on $\lambda$ one has to use the fact that many of the edges in $C$ will have $\ell^{*}$ values sufficiently far below $\frac{2}{\lambda}$. 

From here on, our argument goes in a completely different direction compared to those of Thorup. The starting point for the second part of \Cref{thm:CutExistence} is that we may assume $\overline{a}$ is sufficiently close to $\tfrac{2}{\lambda}$. 
The end goal is to show that for an appropriate value $\hat{a}$ only slightly smaller than $\overline{a}$, and for any large enough greedy tree-packing $\T$, the graph $G/E^{\T}_{<\hat{a}}$ must contain at least one trivial min-cut. We show this in two main steps. First, we show that $G/E^{*}_{<\bar{a}}$ contains at least one trivial min-cut. 
To show this we may assume this is not the case for contradiction, and then count the edges of the graph in two different ways: once using the $\ell^{*}$ values and once by summing the degrees. This then yields a contradiction since $\hat{a}^{-1}$ is sufficiently smaller than $\tfrac{\lambda + 1}{2}$.

Next, we show that any edge $e$ contracted in $G/E^{\T}_{<\hat{a}}$ has sufficiently small $\ell^{*}(e)$. To do so, we exploit the fact that if $\ell^{*}(e) < \overline{a}$, then $e$ is not in any min-cut by construction. We focus on an arbitrary trivial min-cut $X$ of $G/E^{*}_{<\hat{a}}$. The remaining part of the proof then boils down to showing that $\Phi_{G[X]}$ is sufficiently large.
To do so, we formulate an optimization problem which lower bounds the value $\Phi_{G[X]}$ can take, and then show that the solution to the optimization problem is large enough. To do so, we heavily exploit that no edge in $G[X]$ belongs to a min-cut, and so only very few trivial cuts of any partition of $G[X]$ can deviate from having degree at least $\lambda + 1$. 

As outlined above, we begin by showing that we must have some $1$-respecting cut if $ \overline{a}$ is sufficiently far from $\tfrac{2}{\lambda}$. 
%The following lemma can be established using arguments similar to those of Thorup~\cite{Thorup07}: 
\begin{lemma} \label{lma:1resp}
    If $\a\leq \tfrac{2}{\lambda} - \tfrac{1}{c \lambda^2} + \tfrac{1}{\gamma \lambda^2}$ with $\gamma \geq 4c$ then \ref{item:1resp} holds. 
\end{lemma}
\begin{proof}
      By assumption there exists some $e' \in \MC$ such that $\ell^{*}(e') = \a$. Now pick any min-cut $C = (A,B)$ containing $e'$ and fix it for the remainder of the proof. 
      We will first show that $|E^{*}_{= \a} \cap C| \geq \tfrac{\lambda}{2}$. Indeed, consider the graph $G/E^{*}_{<\a}$. Since $C$ separates the endpoints of $e'$ (or $C-e'$ is a smaller cut), it must be that $C$ separates the endpoints of $e'$ in $G[E^{*}_{=\a}]$. 
      Now we note that $\Phi_{X} = 1/\a \geq \tfrac{\lambda}{2}$, by assumption on $\a$ for any connected component $X$ of $(G/E^{*}_{<\a})[E^{*}_{=\a}]$. Next, we apply \Cref{lm:phi_lambda} to the connected component $X_{e'}$ of $(G/E^{*}_{<\a})[E^{*}_{=\a}]$ containing $e'$, which gives that the min-cut in this graph has size at least $\tfrac{\lambda}{2}$. In particular, the smallest cut in $X_{e'}$ separating the endpoints of $e'$ has size at least $\tfrac{\lambda}{2}$. Hence, it must be that $|E^{*}_{= \a} \cap C| \geq \tfrac{\lambda}{2}$. 

      Following the earlier discussion, we now want to argue that 
      \[
            \sum_{e \in C} \ell^{\T}(e) < 2.
      \]
      This follows from the fact that we pack $|\T| = \Omega( \tfrac{\lambda\log m}{\eps^2})$ trees, and so we achieve concentration when $\eps < \tfrac{3}{8c\lambda}$ and thus we can write:
      \[
          \sum_{e \in C} \ell^{\T}(e) \leq \sum_{e \in C} \ell^{*}(e) + \lambda \cdot{} \frac{\eps}{\lambda} 
                                        \leq \lambda \cdot{} \frac{2}{\lambda} - \frac{\lambda}{2}(\frac{1}{c \lambda^2} - \frac{1}{\gamma \lambda^2}) + \eps 
                                        \leq 2 - \frac{3}{8c \lambda} + \eps 
                                        < 2,
      \]
    which per the previous discussion implies that \ref{item:1resp} holds.
\end{proof}

Following the proof outline in the beginning of this section, we next show that if $\a$ is greater than or equal to the cut-off value of $\tfrac{2}{\lambda} - \tfrac{1}{c \lambda^2} + \tfrac{1}{\gamma \lambda^2}$, then it is sufficiently close to $\tfrac{2}{\lambda}$ for \ref{item:trivial_cut} to hold. 
We begin by showing that $G/E^{*}_{<\a}$ must contain at least one trivial min-cut that corresponds to a min-cut of $G$.
\begin{lemma} \label{lma:trivexist}
    Suppose that  $\gamma \geq 4c > 4$ and that $\a$ is as outlined above, then $G/E^{*}_{<\tfrac{2}{\lambda} - \tfrac{1}{c \lambda^2} + \tfrac{1}{\gamma \lambda^2}}$ contains at least one trivial min-cut that corresponds to a min-cut of $G$.
\end{lemma}
\begin{proof}
    As hinted at earlier, we will assume that this is not the case for contradiction. We will then count the edges of $G' = G/E^{*}_{<\tfrac{2}{\lambda} - \tfrac{1}{c \lambda^2} + \tfrac{1}{\gamma \lambda^2}}$ twice to reach the contradiction. 
    To do so, let $\P_0$ be the partition for $G$ with $\rm{part\_val}(\P_0)=\Phi$. Next, we recurse on the subgraphs $G[S]$ for each $S\in \P_0$. 
    To keep notation simple, we do not introduce double indices, but just continue numbering: first, we write $\P_0= S_1, \dots, S_{|\P_0|}$.  Then, for the partition $\P_1$ of $S_1$, we write $\P_1=\{S_{|\P_0|+1}, \dots,S_{|\P_0|+|\P_1|}\} $. Next, for the partition $\P_2$ of $S_2$, we write $\P_2=\{S_{|\P_0|+|\P_1|}, \dots,S_{|\P_0|+|\P_1|+|\P_2|}\} $. We continue until we reach the partition of $S_{|\P_0|}$, and then go down a level to $S_{|\P_0|+1}$. Formally, we partition the next set $S_i$, with minimal $i$. 
    
    The above process gives us a total decomposition with $\P_i$ the partition on $G[S_i]$, where each $S_i\in \P_j$ for some $j< i$, for $i\in \mathcal{I}$ for some index set $\mathcal I$. We are only interested in parts of the decomposition, and so we will only recurse on $S_i\in \P_j$ if $\Phi_{G[S_i]} \leq (\tfrac{2}{\lambda} - \tfrac{1}{c \lambda^2} + \tfrac{1}{\gamma \lambda^2})^{-1}$.
    
    Now we have the containment:
    \begin{align*}
        E^{*}_{\geq \frac{2}{\lambda} - \frac{1}{c \lambda^2} + \frac{1}{\gamma \lambda^2}} \subset \bigcup_{i\in \mathcal I} E(G[S_i]/\P_i).
    \end{align*}
    Indeed, any edge in $E^{*}_{\geq \frac{2}{\lambda} - \frac{1}{c \lambda^2} + \frac{1}{\gamma \lambda^2}}$ will be included since by \Cref{lma:NashTutte}, we have that $\Phi_{G[S]} = (\max_{e \in G[S_i]} \ell^{*}(e))^{-1}$, and so the recursion will not terminate until every edge in $E^{*}_{\geq \frac{2}{\lambda} - \frac{1}{c \lambda^2} + \frac{1}{\gamma \lambda^2}}$ belongs to some subgraph. 
    In fact, the containment is an equality.
    Next, we have 
    \begin{align*}
        |E^{*}_{\geq \frac{2}{\lambda} - \frac{1}{c \lambda^2} + \frac{1}{\gamma \lambda^2}}| &\leq \sum_{i\in \mathcal I} |E(G[S_i]/\P_i)| \\
        &\leq \sum_{i\in \mathcal I} \Phi_i (|V_{i}| -1) \\
        &\leq (\frac{2}{\lambda} - \frac{1}{c \lambda^2} + \frac{1}{\gamma \lambda^2})^{-1}\sum_{i\in \mathcal I} (|V_{i}| -1), 
    \end{align*}
    where $\Phi_i$ is the packing value of $\P_i$, and $V_{i}=V(G[S_i]/\P_i)$.
    Next, we show that 
    \begin{align*}
        \sum_{i\in \mathcal I} (|V_{i}| -1) \leq |V(G')|-1.
    \end{align*}
    We divide $\mathcal I$ in different depths of recursion: let $\mathcal I_j$ be such that for $i\in \mathcal I_j$ we have $S_i$ in recursion depth $j$ (i.e., it is in $j$ partitions before). For ease of notation, we let a singleton cluster $\{v\}$ again partition into $\{v\}$. This has no effect on the sum over $|V_{i}| -1$, but ensures that all vertices make it to the lowest depth. Let $r$ be the total recursion depth. 
    With this notation, we can write for $0 \leq j < r$
    \[
    \sum \limits_{i \in I_j} |V_{i}| = |I_{j+1}| 
    \]
    since we must recurse exactly once on each vertex in $V_{i}$ in the next level. 
    Let the partition $\hat{\P}$ of $G$ be obtained by letting $S \in \hat{\P}$ if and only if $S \in \P_i$ for some $i \in I_{r}$. 
    Then, similarly to above, the introduced notation implies that 
    \[
    \sum \limits_{i \in I_r} |V_{i}| = |V(G/\hat{\P})|.
    \]
    We have $ E^{*}_{\geq \frac{2}{\lambda} - \frac{1}{c \lambda^2} + \frac{1}{\gamma \lambda^2}} \subset E(G/\hat{\P})$ by before. 
    Thus, we note that any edge contracted in $G/\hat{\P}$ is also contracted in $G'$, and therefore we have $|V(G/\hat{\P})| \leq |V(G')|$. Indeed, contracting an edge can never increase the size of the vertex set. 
    Finally, we can write
    \begin{align*}
         \sum_{i\in \mathcal I} (|V_{i}| -1) &= \sum_{j=0}^r \sum_{i\in \mathcal I_j} (|V_{i}| -1) \\
         &= \sum_{j=0}^r \left( \sum_{i\in \mathcal I_j} |V_{i}|\right) - |I_j|\\
         &=  |V(G/\hat{\P})| - |I_r| +\sum_{j=0}^{r-1} |\mathcal I_{j+1}| -|\mathcal I_{j}| \\
         &= |V(G/\hat{\P})| - |I_0|\\
         &\leq |V(G')|-1.
    \end{align*}
    Now we upper bound as follows:
    \begin{align*}
         |E^{*}_{\geq \frac{2}{\lambda} - \frac{1}{c \lambda^2} + \frac{1}{\gamma \lambda^2}}| 
         &\leq \left(\frac{2}{\lambda} - \frac{1}{c \lambda^2} + \frac{1}{\gamma \lambda^2}\right)^{-1}\sum_{i\in \mathcal I} (|V_{i}| -1) \\
         &\leq \left(\frac{2}{\lambda} - \frac{1}{c \lambda^2}\right)^{-1}(|V(G')|-1) \\
         &\leq \frac{\frac{\lambda}{2}}{1-\frac{1}{2c\lambda}}(|V(G')|-1) \\
         &\leq \frac{\lambda}{2}\left(1+\frac{1}{c\lambda}\right)(|V(G')|-1) \\
         &\leq \left(\frac{\lambda}{2} + \frac{1}{2c}\right)(|V(G')|-1),
    \end{align*}
    where we used the well-known fact that $\tfrac{1}{1-x} \leq 1+2x$ if $x < 1$.
    
    To reach a contradiction, we let $\mathcal{P}$ be the final partition obtained in the above procedure and apply the classical hand-shaking lemma to obtain: 
    \[
        |E^{*}_{\geq \frac{2}{\lambda} - \frac{1}{c \lambda^2} + \frac{1}{\gamma \lambda^2}}| =\frac{1}{2} \sum  \limits_{X \in \mathcal{P}} d_{G'}(X) \geq \frac{\lambda + 1}{2}|V(G')|,
    \]
    which is a contradiction for $c > 1$. Here we used the assumption that $G'$ contains no trivial cut of size $\lambda$. Note that a cut of $G'$ corresponds to a cut in $G$ containing the exact same edges.
\end{proof}

Finally, we show that if one was to recurse the idealized load decomposition on a partition in $G/E^{*}_{<\tfrac{2}{\lambda} - \tfrac{1}{c \lambda^2} + \tfrac{1}{\gamma \lambda^2}}$, then the recursive call on a trivial cut will produce a new set $S$ for which $\Phi_{G[S]}$ is sufficiently far away from $\a$ to establish a separation. 
\begin{lemma} \label{lma:GSphi}
    Suppose $\a \geq \tfrac{2}{\lambda} - \tfrac{1}{c \lambda^2} + \tfrac{1}{\gamma \lambda^2}$, and let $\mathcal{P}$ be the partition induced in $G/E^{*}_{<\frac{2}{\lambda} - \frac{1}{c \lambda^2} + \frac{1}{\gamma \lambda^2}}$, and let $S \in \mathcal{P}$ be any trivial min-cut in this partition. Then we have
    \[
    \Phi_{G[S]} \geq \frac{\lambda}{2} + \frac{1}{2}.
    \]
\end{lemma}
\begin{proof}
    We will exploit the fact that no edge belonging to $E(G[S])$ is contained in any min-cut. We consider the minimum partition of $G[S]$, call it $\mathcal{P}_{S}$, and let $\tilde{G} = G[S]/\mathcal{P}_{S}$. Then it must be the case that any vertex of degree $d \leq \lambda$ in $\tilde{G}$ is incident to at least $\lambda + 1 - d$ edges in $G/E^{*}_{<\frac{2}{\lambda} - \frac{1}{c \lambda^2} + \frac{1}{\gamma \lambda^2}}$. Indeed, otherwise an incident edge in $ G[S]$ belongs to some min-cut in $G$, which contradicts the choice of $\a$.
    It now follows that any choice of $\tilde{G}$ with $k$ vertices must induce a feasible solution to the following optimization problem with objective value $\Phi_{\tilde{G}}$: 
    \begin{subequations}
        \begin{alignat*}{2}
        &\!\min_{\mathbf{\alpha}, \mathbf{\beta}}        &\qquad& \frac{1}{2(k-1)}\sum \limits_{i = 1}^{k} (\lambda + 1 - \alpha_i + \beta_i)\\
        &\text{subject to} &      & \forall i: \alpha_i, \beta_i\geq 0,\\
        &                  &      & \sum \limits_{i = 1}^{k} \alpha_{i} \leq \lambda.
        \end{alignat*}
    \end{subequations}
    Indeed, for an arbitrary numbering of the vertices of $\tilde{G}$, we pick the unique feasible choice of $\mathbf{\alpha}$ and $\mathbf{\beta}$ that satisfies $\deg(v_i) = \lambda + 1 - \alpha_i + \beta_i$ with $\deg_{G/E^{*}_{<\frac{2}{\lambda} - \frac{1}{c \lambda^2} + \frac{1}{\gamma \lambda^2}}}(v_i) + \deg(v_i) = \lambda + 1 +\beta_i$. It then follows by the hand-shaking lemma and the definition of the partition value that:
    \[
    \Phi_{\tilde{G}} = \frac{E(\tilde{G})}{|V(\tilde{G})|-1} = \frac{\sum \limits_{v \in V(\tilde{G})} \deg(v_i)}{2(k-1)} = \frac{1}{2(k-1)}\sum \limits_{i = 1}^{k} (\lambda + 1 - \alpha_i + \beta_i).
    \]
    Furthermore, by the above, we know that $\sum \limits_{v \in V(\tilde{G})} \max\{0, \lambda + 1 - d(v)+\beta(v)\} \leq \lambda$ since any vertex $v$ with degree $d \leq \lambda $ in $\tilde{G}$ must be incident to at least $\lambda + 1 - d + \beta(v)$ edges in $G/E^{*}_{<\frac{2}{\lambda} - \frac{1}{c \lambda^2} + \frac{1}{\gamma \lambda^2}}$. By assumption, we know that $S$ has at most $\lambda$ such edges. 

    In particular, we find that the solution to the above optimization problem lower bounds the value of $\Phi_{\tilde{G}}$. 
    Since we are optimizing over a closed and bounded subset of $\mathbb{R}^{2k}$ it follows by the Extreme Value Theorem that an optimal solution exists. To find it, simply note that if any $\beta_i >0$, we can reduce the objective value by setting $\beta_{i} = 0$, and if $\sum \limits_{i = 1}^{k} \alpha_{i} < \lambda$, we can reduce the objective value by increasing any $\alpha_{i}$. Hence, we can assume that any solution to the optimization problem has $\mathbf{\beta} = \mathbf{0}$ and $\sum \limits_{i = 1}^{k} \alpha_{i} = \lambda$. Now simple calculations yield the result: 
    \[
        \Phi_{G[S]} \geq \frac{1}{2(k-1)}\sum \limits_{i = 1}^{k} (\lambda + 1 - \alpha_i) = \frac{k(\lambda + 1) - \lambda}{2(k-1)} = \frac{(k-1)\lambda + k}{2(k-1)} = \frac{\lambda}{2} + \frac{k}{2(k-1)} > \frac{\lambda}{2} + \frac{1}{2},
    \]
    for any integer $k \geq 2$. 
\end{proof}
If $\a\leq \tfrac{2}{\lambda} - \tfrac{1}{c \lambda^2} + \tfrac{1}{\gamma \lambda^2}$ then \ref{item:1resp} holds by Lemma~\ref{lma:1resp}. Otherwise, if $\a\geq \tfrac{2}{\lambda} - \tfrac{1}{c \lambda^2} + \tfrac{1}{\gamma \lambda^2}$ we will show that \ref{item:trivial_cut} holds. To this end, let $S \in V\left(G/E^{*}_{<\frac{2}{\lambda} - \frac{1}{c \lambda^2} + \frac{1}{\gamma \lambda^2}}\right)$ be any trivial min-cut. Note that under the current assumptions and mild assumptions on $c$ and $\gamma$ (which we will verify later), Lemma~\ref{lma:trivexist} guarantees the existence of such a trivial cut. If we can show that for any tree-packing with $|\T| = \Omega(\lambda^3 \log m)$, we have $\ell^{\T}(e) < \tfrac{2}{\lambda} - \tfrac{1}{c \lambda^2}$ for all $e \in E(G[S])$ and $\ell^{\T}(e) > \tfrac{2}{\lambda} - \tfrac{1}{c \lambda^2}$ for all $e \in E\left(G/E^{*}_{<\frac{2}{\lambda} - \frac{1}{c \lambda^2} + \frac{1}{\gamma \lambda^2}}\right)$, it then follows that $S$ is exactly a trivial cut of $G_{\hat{a}}$ thus establishing \ref{item:trivial_cut}.

Note that $S$ is a connected component of $\left(V,E^{*}_{<\frac{2}{\lambda} - \frac{1}{c \lambda^2} + \frac{1}{\gamma \lambda^2}}\right)$, and therefore computing the $\ell^{*}$ values on $G[S]$ corresponds to computing the $\ell^{*}$ values on $G$ from the last depth where $S$ was not contracted. 
Indeed, here the recursive definition will recurse on exactly $G[S]$. 
Therefore we have
\[
    \max_{e \in G[S]} \ell^{*}(e) =  1/\Phi_{G[S]}.
\]
Now by this inequality, \Cref{lma:NashTutte}, and \Cref{lma:GSphi} we have
\[
\max_{e \in G[S]} \ell^{*}(e) \leq \Phi_{G[S]}^{-1}\le \paren{\frac{\lambda}{2} + \frac{1}{2}}^{-1} \leq \frac{\frac{2}{\lambda}}{1+\frac{1}{\lambda}} \leq \frac{2}{\lambda}(1-\frac{1}{2\lambda}) = \frac{2}{\lambda} - \frac{1}{\lambda^{2}},
\]

since for $x \leq 1$ we have $\tfrac{1}{1+x} \leq 1-\tfrac{x}{2}$. Now, any greedy tree-packing $\T$ containing at least $|\T| \geq 6\tfrac{\log m }{\eps^2} \lambda$ trees will by \Cref{lm:tree_packing_thorup} have that:
\[
    \max_{e \in G[S]} \ell^{\T}(e) \leq \frac{\eps}{\lambda} + \max_{e \in G[S]} \ell^{*}(e) \leq \frac{2}{\lambda} - \frac{1}{\lambda^{2}} + \frac{\eps}{\lambda}.
\]
Hence, if $\tfrac{1}{\lambda^{2}} - \tfrac{\eps}{\lambda} > \tfrac{1}{c\lambda^{2}}$, then we have what we wanted to show. Therefore, we only require $\tfrac{\eps}{\lambda} < \tfrac{c-1}{c\lambda^{2}}$ for this part to hold. 

For any $e \in E\left(G/E^{*}_{<\frac{2}{\lambda} - \frac{1}{c \lambda^2}+ \frac{1}{\gamma \lambda^2}}\right)$, we similarly have  for any greedy tree-packing $\T$ containing at least $|\T| \geq 6\tfrac{\log m}{\eps^2} \lambda$ trees that by \Cref{lm:tree_packing_thorup}
\[
    \min_{e \in E\left(G/E^{*}_{<\frac{2}{\lambda} - \frac{1}{c \lambda^2}+ \frac{1}{\gamma \lambda^2}}\right)} \ell^{\T}(e) \geq -\frac{\eps}{\lambda} + \min_{e \in E\left(G/E^{*}_{<\frac{2}{\lambda} - \frac{1}{c \lambda^2}+ \frac{1}{\gamma \lambda^2}}\right)} \ell^{*}(e) \geq \frac{2}{\lambda} - \frac{1}{c \lambda^2}+ \frac{1}{\gamma \lambda^2} - \frac{\eps}{\lambda}.
\]
Hence for this part we simply require $\tfrac{\eps}{\lambda} < \tfrac{1}{\gamma \lambda^2}$. 

So for all of the arguments above to be valid, we require that $\gamma \geq 4c$, $c > 1$, $\tfrac{\eps}{\lambda} < \tfrac{c-1}{c\lambda^{2}}$, $\tfrac{\eps}{\lambda} < \tfrac{1}{\gamma \lambda^2}$, and that $\eps < \tfrac{3}{8c \lambda}$. 
Now the first condition together with a choice of $c \geq \tfrac{3}{2}$ immediately imply that $\tfrac{c-1}{c\lambda^{2}} \geq \tfrac{1}{\gamma \lambda^2}$, so setting $c = 2$, $\gamma = 8$, and $\eps = \tfrac{1}{16\lambda}$ works. As argued earlier, we now have that all edges in $G[S]$ are contracted in $G_{\tfrac{2}{\lambda}-\tfrac{1}{2\lambda^2}}$ and that all edges in $E\left(G/E^{*}_{<\frac{2}{\lambda} - \frac{1}{c \lambda^2}+ \frac{1}{\gamma \lambda^2}}\right)$ are not contracted in $G_{\tfrac{2}{\lambda}-\tfrac{1}{2\lambda^2}}$ so in particular $S$ represents the required trivial min-cut in case 
\ref{item:trivial_cut}, when $\a\geq \tfrac{2}{\lambda} - \tfrac{1}{c \lambda^2} + \tfrac{1}{\gamma \lambda^2}$. This concludes the proof of \Cref{thm:CutExistence} for any greedy tree-packing $\T$ with at least $|\T| \geq 6\cdot{}16^2\cdot{}\lambda^{3} \log m = 1536 \cdot{}\lambda^{3} \log m$ trees. 

\subsection{Estimating Trivial Cuts in \texorpdfstring{$G_a$}{Ga}}\label{sc:trival_cuts}
In this section, we design a data structure that takes a parameter $a$ as input and is able to report an estimate of the size of the smallest trivial cut in $G_a$. 
Depending on the current loads, $G_a$ might only contain one vertex, in which case the data structure returns $\infty$. 
%Later, we will combine many such data structures together with the algorithm of Thorup~\cite{Thorup07} this with \Cref{thm:CutExistence} to give a min-cut. 
Formally, we show the following lemma.
\begin{lemma}\label{lm:trivial_cuts}
    Let $G$ be a dynamic unweighted, undirected (multi-)graph, and assume we have access to a black-box dynamic algorithm that maintains a tree-packing $\T$ on $G$ with loads $\ell^{\T}(\cdot{})$. 
    Suppose an update to $G$ results in $P(n)$ loads crossing $a$ during the update of $\T$. 
    Then there is a deterministic data structure which reports a value $\mu$ such that:
    \begin{itemize}
        \item If $|V(G_a)| = 1$, then $\mu = \infty$.
        \item Else $\mu = \min_{X \in V(G_a)} d_{G_a}(X)$.
    \end{itemize}
    The algorithm has $O((1+P(n))\log m)$ amortized update time, or $O((1+P(n))\sqrt{n})$ worst-case update time. 
    Both data structures can list the edges incident to some vertex $X \in V(G_a)$ with $d_{G_a}(X) = \mu$ in $O(\log m)$ worst-case time per edge. 
    Here loops are listed twice. 
\end{lemma} 
We briefly note that if $G$ has min-cut $\lambda$ and $\bar{a} \geq \tfrac{2}{\lambda} - \tfrac{3}{8\lambda^2}$ at some time $t$, then the algorithm will return $\mu_t = \lambda$ at time $t$, provided that $a = \tfrac{2}{\lambda}-\tfrac{1}{2\lambda^2}$. 
Indeed, when $|V(G_a)| \geq 2$, the set of edges incident to a single vertex of $G_a$ are super-sets of cuts in $G$, and so the degrees of vertices in $G_a$ upper bounds the size of cuts in $G$. 
Hence $\mu_t \geq \lambda$. 
However, under these assumptions if follows by \Cref{thm:CutExistence} that some trivial cut of $G_a$, say around $X \in G_a$, will be a min-cut of $G$. Since the proof of \Cref{thm:CutExistence} shows that all edges $e \in E(G[X])$ have $\ell^{\T}(e)< a$, it follows that $X$ has no loops in $G_a$, and we conclude that $\mu_t \leq \lambda$. 
Note also that for any choice of $a$ such that $|V(G_a)| \geq 2$, we have $\mu_{t} \geq \lambda_{t}$ since the edges incident to any vertex of $G_a$ form a super-set of a cut, as argued above. 

We briefly discuss the intuition behind the proof. 
Since it can be expensive to support contractions and un-contractions, we will not maintain an explicit representation of $G_a$. 
Instead, we will maintain the graph $\Gamma = G[E^{\T}_{<a}]$ explicitly. 
For each connected component of $\Gamma$, we maintain a spanning tree. 
For any connected component $C$ in $\Gamma$, we let the $\emph{external degree}$ of $C$ be the number of endpoints of edges in $E^{\T}_{\geq a}$ belonging to $C$. 
We observe that the external degree of a connected component $C$ corresponds exactly to the degree of the vertex $X$ in $G_{a}$ represented by $C$. 
Thus we can maintain the degrees of vertices in $G_{a}$, by maintaining the external degrees of every component in $\Gamma$.
This can be achieved by storing the spanning tree of each component as a top tree. By storing some additional information, the top trees can compute the external degrees exactly. 
Finally, we note that we only have to perform updates to $\Gamma$ or the top trees, whenever some edge is inserted or deleted into $G$, or whenever the load of some edge crosses $a$. 

\paragraph{Description of data structure.} 
The data structure updates as follows: 
\begin{itemize}
    \item We maintain $\Gamma$ as well as a connectivity data structure on $\Gamma$:
    \begin{itemize}
        \item After an update to $G$, we add a new edge $e$ to $\Gamma$ if $\ell^{\T}(e) < a$ after the tree-packing is updated, and we delete an old edge $e$ from $\Gamma$ if $\ell^{\T}(e) < a$ before the tree-packing is updated. 
        \item We remove an edge $e$ from $\Gamma$ if its load $\ell^{\T}(e) < a$ before the tree-packing is updated and its load $\ell^{\T}(e) \geq a$ after the tree-packing is updated.
        \item We add an edge $e$ to $\Gamma$ if its load $\ell^{\T}(e) \geq a$ before the tree-packing is updated and its load $\ell^{\T}(e) < a$ after the tree-packing is updated.
    \end{itemize}
    \item For each vertex $v \in V(G)$, we maintain the degree of $v$ in $G[E^{\T}_{\geq a}]$ as well as the edges from $G[E^{\T}_{\geq a}]$ incident to $v$. 
    \item For each connected component $C$ in $\Gamma$, we maintain the sum 
    \[
    S(C):= \sum \limits_{v \in C} d_{G[E^{\T}_{\geq a}]}(v),
    \]
    as well as enough information to list all vertices $v \in C$ with $d_{G[E^{\T}_{\geq a}]}(v) > 0$.
    \item We maintain a min-heap containing $S(C)$ for every connected component $C$ of $\Gamma$. 
\end{itemize}
We can report the value $\mu$ by setting $\mu$ equal to the minimum element of the min-heap. In the case where the min-heap only contains one element, we set $\mu = \infty$.
We can report all edges by first listing the vertices $v \in C$ with $d_{G[E^{\T}_{\geq a}]}(v) > 0$, and then listing the edges incident to $v$ in $G[E^{\T}_{\geq a}]$.

\paragraph{Implementation.} 
Next we describe how to efficiently implement the above steps. 
We need access to different data structures, which we list below. 
To maintain the connected components of $\Gamma$, we will use the following data structures:
\begin{lemma}[\cite{WorstcaseConn}] \label{lm:wcConn}
     There exists a deterministic fully-dynamic algorithm that maintains a spanning forest of a dynamic graph in $O(\sqrt{n})$ worst-case update time and $O(1)$ worst-case recourse.
\end{lemma}
\begin{lemma}[Theorem 3 in~\cite{HolmLT01}] \label{lm:amConn}
    There exists a deterministic fully-dynamic algorithm that maintains a spanning forest of a dynamic graph in $O(\log^2 m)$ amortized update time and $O(1)$ worst-case recourse.
\end{lemma}
In addition to these data structures, we maintain each spanning tree as a top tree~\cite{AlstrupHLT05}. 
We use the following interface:
\begin{lemma}[Theorem 1~\cite{AlstrupHLT05}] \label{lm:top}
    For a dynamic forest, one can maintain a top tree of height $O(\log m)$ supporting the operations \emph{\texttt{link}}$(u,v)$, \emph{\texttt{cut}}$(u,v)$, and \emph{\texttt{expose}}$()$ in $O(\log m)$ worst-case update time using only $O(1)$ calls to \emph{\texttt{create}} and \emph{\texttt{destroy}} and $O(\log m)$ calls to \emph{\texttt{join}} and \emph{\texttt{split}}.
\end{lemma}
The following lemma is then a routine application of top trees.
\begin{lemma} \label{lm:extDegrees}
    Given a dynamic forest $F$ where each vertex $v$ has integer weight $w(v)$, there is a data structure supporting the following operations in $O(\log m)$ worst-case time per operation:
    \begin{itemize}
        \item \emph{\texttt{link}}$(u,v)$: add an edge between $u$ and $v$ in $F$.
        \item \emph{\texttt{cut}}$(u,v)$: delete an edge between $u$ and $v$ in $F$.
        \item \emph{\texttt{IncrementWeight}}$(v)$: increment the weight of $v$ by $1$.
        \item \emph{\texttt{DecrementWeight}}$(v)$: decrement the weight of $v$ by $1$.
        \item \emph{\texttt{TotalWeight}}$(T)$: return the sum of vertex weights in $T$.
    \end{itemize}
    Furthermore, for each tree $T$, one can list the vertices of $T$ with weight~$>0$ in $O(\log m)$ time per vertex. 
\end{lemma}
\begin{proof}
    The first two operations are already supported directly by the top tree. 
    To support the other three, we store as additional information for a cluster $A$ the sum of weights of non-boundary vertices, $\operatorname{WeightSum}(A)$. 
    This information can be maintained under a $C =\texttt{join}(A,B)$ operation by updating 
    \[
    \operatorname{WeightSum}(C) = \operatorname{WeightSum}(A) + \operatorname{WeightSum}(B) + \sum \limits_{v \in (\partial A \cup \partial B)\setminus \partial C} w(v),
    \]
    with only $O(1)$ overhead. 
    In order to implement \texttt{IncrementWeight}$(v)$ and \texttt{DecrementWeight}$(v)$, one can call \texttt{Expose}$(v)$, thus turning $v$ into a boundary vertex. Then the weight of $v$ can be updated without invalidating any information in the top tree.
    In order to answer a \texttt{TotalWeight}$(T)$ query, we let $R$ be the root-cluster of the top tree representing $T$, and return 
    \[
    \operatorname{WeightSum}(R) + \sum \limits_{v \in \partial R} w(v).
    \]
    Finally, we maintain as additional information for each cluster the number of non-boundary vertices of weight~$>0$. 
    This can be done analogously to above. 
    We can then find all vertices of weight~$>0$ as follows: if the current cluster is a leaf cluster, return all non-boundary vertices of weight~$>0$. Otherwise, report all non-boundary vertices that are boundary vertices for both children, and recurse on all children containing at least one vertex of weight~$>0$. 
    In the special case where the current cluster is the root cluster, we also return all boundary vertices of weight~$>0$. 
    Since the top tree has height $O(\log m)$, and each recursion takes constant time, we can report each such vertex in $O(\log m)$ time per vertex.
\end{proof}
We now implement the data structure as follows. We use the connectivity data structure from \Cref{lm:amConn} to maintain $\Gamma$ (or \Cref{lm:wcConn} for worst-case guarantees). 
Both algorithms maintain a spanning forest of $\Gamma$ which we additionally store using the data structure from \Cref{lm:extDegrees}. 
We maintain the invariant that $w(v) = d_{G[E^{\T}_{\geq a}]}(v)$.
Then, we immediately support the required operations. 
To maintain the invariant, we note that whenever an edge $uv$ leaves $\Gamma$, we increment the weights of $u$ and $v$ by $1$, and whenever an edge, previously in $G[E^{\T}_{\geq a}]$, enters $\Gamma$, we decrement the weights of $u$ and $v$ by $1$. 
New or old edges resulting from an update to $G$ are handled similarly. 
Finally, we can implement the dynamic min-heap in $O(\log m)$ update and query time using any standard balanced binary search tree. 

\paragraph{Correctness.} Correctness follows readily from \Cref{lm:amConn}, \Cref{lm:wcConn}, and \Cref{lm:extDegrees}. 
We need only verify that if $\Gamma$ contains at least two vertices, then 
$\mu = \min_{X \in V(G_a)} d_{G_a}(X)$, but as noted earlier this follows from the fact that 
\[
\sum \limits_{v \in C} d_{G[E^{\T}_{\geq a}]}(v) =  d_{G_a}(X),
\]
if $X$ is the vertex in $G_a$ represented by the connected component $C$ in $\Gamma$.

\paragraph{Analysis.} Each operation on $G$ is supported in $O(T(n) + \log m)$ time, where $T(n)$ is the time needed for the connectivity data structure. 
Each time the load of some edge crosses the threshold $a$, we have to perform $O(1)$ deletions and insertions to $\Gamma$, and update $O(1)$ weights of vertices. 
The first type of operation is supported in $O(T(n) + \log m)$ time by \Cref{lm:amConn}, \Cref{lm:wcConn}, and \Cref{lm:extDegrees}. 
Indeed, the only additional operations we need account for are the updates to the min-heaps. Each insertion or deletion to $\Gamma$ forces at most $O(1)$ changes to the set of connected components, so this can be supported in $O(\log m)$ time. 
The second type of operation is supported in $O(\log m)$ time directly by \Cref{lm:extDegrees}.

\subsection{Recourse in Tree-Packing}\label{sc:TP_recourse}
The goal of this section is to show that we can bound the recourse to $\sim \lambda_{\max}^5$ when maintaining $\sim \lambda_{\max}^3$ greedy spanning trees.
A standard argument gives that we can maintain $|\T|$ greedy spanning trees with $O(|\T|^2)$ recourse (see e.g.,~\cite{ThorupK00}), but with a more careful analysis, we can shave a factor~$\lambda_{\max}$. We start with the following lemma, which showcases the main idea. 

\begin{lemma}\label{lm:TP_recourse}
    Let $|\T|\geq \lambda \log m$. We can maintain $|\T|$ greedy spanning trees $\T$ with a recourse of $O(|\T|^2/\lambda)$.
\end{lemma}
\begin{proof}
    First, let us consider a deletion of an edge $e$. We observe that it can appear in at most a limited number of trees since by \Cref{lm:tree_packing_thorup} (with $\eta=1)$
    \[
        \ell^\T(e)\leq \ell^*(e)+ 1/\lambda \leq \max_{e'} \ell^*(e')+1/\lambda \leq 1/\Phi+1/\lambda \leq 3/\lambda,
    \]
    where the third inequality follows by \Cref{lma:NashTutte}, and the last inequality holds by \Cref{lm:phi_lambda}. So we have $L^\T(e)=|\T|\ell^\T(e)=O(|\T|/\lambda)$. In each tree where $e$ appears, we need to do an update that can lead to a chain of $|\T|$ updates. Hence we get total recourse $O(|\T|^2/\lambda)$. 
\end{proof}

This lemma by itself does not help us yet: we are maintaining $\sim\lambda_{\max}^3$ trees, and $\lambda_{\max}^6/\lambda$ can be as big as $\lambda_{\max}^6$, when $\lambda$ becomes small. 
The trick is to maintain $O(\log \lambda_{\max})$ different tree-packings $\T_i$, each of different size.
Now, we only need the tree-packing $\T_i$ to be correct when $\lambda_t\in [2^i,2^{i+1})$, and so some of the packings can be much smaller. For $i$ such that $\lambda_t \ge 2^{i+1}$, the tree-packing can just be seen as a truncated packing from a larger $i$ and will be correct as well. For larger $i$, i.e., when $\lambda < 2^i$, we can add fake input edges to keep the min-cut larger and hence the update time smaller. First, we observe that in \Cref{lm:TP_recourse} we do not actually use the min-cut, but we have a more local argument: we need $L^{\T_i}(e)$ to be small enough. We will exploit this by locally adding edges, such that all $L^{\T_i}(e)$ stay small. 

\begin{lemma}\label{lm:TP_recourse_impr}
    We can maintain a tree-packing of size $\Theta(\lambda_t^3 \log m)$ in the following manner. 
    We maintain $\log m$ tree-packings $\T_i$ of various graphs of size $|\T_i|=\Theta(2^{3i} \log m)$, with $O(2^{5i}\log^2 m)$ worst-case recourse in tree-packing $\T_i$, using any minimum spanning tree algorithm with $O(1)$ worst-case recourse to maintain the individual trees. We then have that $\T_i$ for $i$ such that $\lambda_t\in [2^i,2^{i+1})$ is a tree-packing on $G$ of the required size. 
\end{lemma}
\begin{proof}
    We will maintain tree-packings $\T_1, \T_2, \dots, \T_{\log \lambda_{\max}}$, where $|\T_i|=O(2^{3i} \log m)$. We only need the packing $\T_i$ with $\lambda_t\in [2^i,2^{i+1})$. The other $\T_i$ will not necessarily correspond to a tree-packing of $G$, but we maintain all of them simultaneously with the stated recourse. 

    First, note that for $\T_i$ with $2^i\leq \lambda_t$ the results directly holds by \Cref{lm:TP_recourse}. In fact, we do not need to compute these separately, but we can just use the truncated tree-packing below the cut-off value. For the rest of the proof, we focus on $\T_i$ with $2^i> \lambda_t$. 
    
    Next, we consider an initialization step for the tree-packings $\T_i$ with $2^i>\lambda_0$. For each of these, we do not start with only $G$, but we add a path graph where each edge appears with multiplicity $2^{i-4}$. We call these added edges \emph{virtual}. Next, we compute a tree-packing for the resulting graph $G_i$. Clearly this graph now has min-cut at least $2^{i-4}$. More precisely, it also means that $L^{\T_i}(e)\leq |\T_i|/{2^{i-4}}=16|\T_i|/2^i$, hence the recourse of updates to this graph is now bounded by $O(2^{5i}\log^2 m)$. Next, we delete all virtual edges with $L^{\T_i}(e)\leq 8|\T_i|/2^i$. These edges were not necessary, in the sense that $G$ already guarantees the right connectivity. 

    \begin{itemize}
        \item[Claim 1.] Deleting a virtual edge $e$ cannot lead to $L^{\T_i}(e')>16|\T_i|/2^i$ for any $e'$. \\
        \textit{Proof. } The last time $e$ was not picked, there must be some path separating the end points of $e$ where each edge has load at most $8|\T_i|/2^i$. 
        In the worst case we add the entire load of $8|\T_i|/2^i$ to this path. Since all edges had load at most $8|\T_i|/2^i$, so they now have load at most $16|\T_i|/2^i$. 
        Note that the load of $e$ will be distributed along this path before being added to an edge $e'$ with $L^{\T_i}(e')=16|\T_i|/2^i$, as any tree previously containing $e$ missed at least one edge from this path.
    \end{itemize}

    This initialization time is  $O(n\cdot 2^{i-4} \cdot 2^{5i}\log^2 m)$. This can be divided in a worst-case manner over insertions in the graph, carrying out at most $O(2^{5i}\log^2 m)$ operations per insertion. The reason is that to reach $\lambda_t\in [2^i,2^{i+1})$, we need at least $n\cdot 2^i$ edges in $G$. 

    Now after every update, whenever a virtual edge satisfies $L^{\T_i}(e)\leq 8|\T_i|/2^i$, we delete it. Note that by Claim 1, this does not lead to any edge $e'$ with $L^{\T_i}(e')>16|\T_i|/2^i$. So we do not get a chain of insertions and deletions.

    However, one edge insertion can lead to multiple deletions. In fact, it can lead to $2^i$ deletions, which would take more than the stated worst-case update time. Instead of actually deleting it right away, we add each virtual edge that needs to be deleted to a deletion queue and delete one edge from the queue after any update. Note that these delayed deletions only have a positive effect on the min-cut (and hence update time). Moreover, if a virtual edge in the deletion queue later increases its load above the threshold, then it is removed from the deletion queue.

    % \begin{itemize}
    %     \item[Claim 2.] An edge insertion can lead to deleting at most $2^i$ virtual edges. \\
    %     \textit{Proof.}  
    %     Let $e=uv$ denote the inserted edge. The number of disjoint $u-v$ paths is the value of the $s-t$ max flow, which equals the $s-t$ min-cut. Since the edge $uv$ can only lead to deletions of virtual edges $e$ when $L^{\T_i}(e)< L^{\T_i}(e')=3|\T_i|/2^i$, we have the the $s-t$ min-cut is at most $2^i$. \tijn{not sure this is true} Since each virtual edge that gets deleted corresponds to a $s-t$ path, we get the result.   
    % \end{itemize}

    Further, if an edge deletion $e$ leads to $L^{\T_i}(e')>16|\T_i|/2^i$ for any $e'$, then we keep $e$ as a virtual edge. This guarantees that at any point in time we have $L^{\T_i}(e)\leq 16|\T_i|/2^i$ for all edges $e$, hence our recourse stays bounded. 
    In particular, the min-cut of $G_i$ is always at least $\tfrac{2^{i}}{16}$. 
    Indeed, for any $e$ in $G_i$, we have $\ell^{*}(e) \leq \max_{f} \ell^{\T_i}(f) \leq 16/2^{i}$. 
    This implies that $\lambda(G_i) \geq \Phi \geq \tfrac{1}{\max_{f} \ell^{*}(f)} \geq \tfrac{2^{i}}{16}$.
    Finally, we show correctness.

    \begin{itemize}
        \item[Claim 2.] If $\lambda_t\in [2^i,2^{i+1})$, then all virtual edges have been placed in the deletion queue. \\
    \textit{Proof. } Since $\lambda_t\in [2^i,2^{i+1})$, we have 
    \begin{align*}
        L^{\T_i}(e)&\leq |\T_i|\ell^*(e)+ |\T_i|/\lambda_t \leq |\T_i|\max_{e'} \ell^*(e')+|\T_i|/\lambda_t\\
        &\leq |\T_i|/\Phi+|\T_i|/\lambda_t \leq 3|\T_i|/\lambda_t\leq 3|\T_i|/2^i\leq 8|\T_i|/2^i.
    \end{align*} 
    \end{itemize}
    Next, we need to show that the deletion queue will be empty by this time. This shows that when $\lambda_t\in [2^i,2^{i+1})$, the tree-packing is correct. In Claim 3, we show that if there are $k$ virtual edges $e$ with $L^{\T_i}(e)\geq 8|\T_i|/2^i$, then we need to insert at least $k$ edges to obtain $\lambda\in [2^i,2^{i+1})$. Hence whenever an edge gets moved to the deletion queue, we know that enough insertions will be performed later, at which time we can carry out the deletion. 

    \begin{itemize}
        \item[Claim 3.] Suppose there are $k$ virtual edges $e$ with $L^{\T_i}(e)\geq 8|\T_i|/2^i$, then we need to insert at least $k$ edges to obtain $\lambda\in [2^i,2^{i+1})$.\\
        \textit{Proof.} Consider the idealized load packing of this graph (including the virtual edges). We write $L^*(e):=|\T_i|\ell^*(e)$. 
        We observe that for each of our $k$ edges
        \[
            L^*(e) \geq L^{\T_i}(e) -|\T_i|/2^{i} \geq 8|\T_i|/2^{i} -|\T_i|/2^{i} \geq 4|\T_i|/2^{i},
        \]
        where the first inequality uses \Cref{lm:tree_packing_thorup} with $\eta=1/16$ so that the maximum error in the estimate is bounded by $1/2^{i}$.
        When $\lambda \geq 2^i$, we have that
        \[
            2^i \leq \lambda \leq 2\Phi = 2 \min_\P \rm{part\_val}(\P), 
        \]
        where the second inequality holds by \Cref{lm:phi_lambda} and equality holds by the definition of $\Phi$. Recall that $\P^*$ denote the packing satisfying this minimum and defining the ideal loads, so we have $L^*(e)\leq |\T_i|/\Phi\leq 2|\T_i|/\lambda \leq 2|\T_i|/2^{i}$ for any $e$. 

        Consider the partition $\P$ defined by the vertices of $G/\{e\in E: L^{*}(e) < 4|\T_i|/2^i\}$. 
        Whenever $\lambda \geq 2^{i}$ any edge $e$ across this partition must have $L^*(e)\leq 2|\T_i|/2^{i}$, meaning that we must have at least doubled the no.\ of edges across this partition. 
        Since this partition used to contain the $k$ virtual edges, we need to insert at least $k$~edges to do so. 
    \end{itemize}

    At any time the deletion queue is empty, we can directly apply Claim 3 for the current virtual edges. If it is nonempty, we delete at least one edge from the queue, which decreases the number of virtual edges. We note that the number of virtual edges can increase by 1 when a deleted edge becomes virtual. In that case, that there will be an additional insertion across the partition $\P^*$ on this level.     
\end{proof}

\subsection{The Algorithm for Bounded \texorpdfstring{$\lambda$}{lambda}}\label{sc:mincut_par}
In this section we give the algorithm that finds the min-cut if this is below some threshold value~$\lambda_{\max}$, which appears as a parameter in the update time. Intuitively, the structure is as follows. If we know $\lambda$, then:

\begin{enumerate}[i)]
    \item We maintain a greedy tree-packing $\T$ of size $\Theta(\lambda^3 \log m)$.\label{step:TP}
    \item We maintain the minimum size of all $1$-respecting cuts of each tree in $\T$. \label{step:1resp}
    \item Maintain the minimal trivial cut of $G_a$, for $a=\tfrac{2}{\lambda}-\tfrac{1}{2\lambda^2}$. \label{step:Ga}
\end{enumerate}
Then by \Cref{thm:CutExistence}, one of the two gives the right answer. We note that we can do Step \ref{step:1resp} efficiently by the following lemma. 

\begin{lemma}[Proposition 24 in \cite{Thorup07}]\label{lm:1_resp}
    There exists a deterministic dynamic algorithm that, given an unweighted, undirected (multi-)graph $G$ with a dynamic spanning tree $T$, maintains a min-cut that $1$-respects the tree in $\tilde O(\sqrt{m})$ worst-case update time. 

    It can return the edges of the cut in $O(\log m)$ time per edge. 
\end{lemma}

Both the size of the tree-packing in Step~\ref{step:TP}, and the graph $G_a$ in Step~\ref{step:Ga}, need $\lambda$, the value of the min-cut.
However, we do not know $\lambda$ -- as our goal is to compute it -- and it changes over time. In the proof we will show how to maintain these structures for all values of $\lambda$ simultaneously without incurring too much overhead.

\mincutPar*
\begin{proof}
    We first note that w.l.o.g.\ we can assume that we have $O(\lambda_{\max} n)$ edges, using the connectivity sparsifier from Nagamochi and Ibaraki~\cite{NagamochiI92}. We use the dynamic version of the sparsifier from Eppstein et al.~\cite{EppsteinGIN97}, which guarantees that each update to original graph leads to at most $2$ updates to the sparsified graph\footnote{We abuse notation slightly and denote by $G$ the graph the rest of the algorithm is run on. I.e., $G$ is the original graph when no sparsifier is applied, and the sparsified graph when a sparsifier is applied.} in $O(\lambda_{\max}\sqrt n)$ update time.

    \textbf{Algorithm.}
    \begin{enumerate}
        \item Maintain $\log \lambda_{\max}$ tree-packings $\T_1, \T_2, \dots$, with $|\T_i|=\Theta(2^{3i}\log m)$. 
        \item For each packing $\T_i$, maintain the minimum size of all cuts $1$-respecting at least one tree in the packing. \label{step:1-resp_cut}
        \item For each $\mu \in \{1, 2, \dots, \lambda_{\max}\}$,  maintain the minimum trivial cut of $G_{a_\mu}$ by \Cref{lm:trivial_cuts}, for $a_\mu:=\tfrac{2}{\mu}-\tfrac{1}{2\mu^2}$. Here we use the tree-packing $\T_i$ such that $\mu \in [2^i,2^{i+1})$ for $G_{a_\mu}:= G/\{e\in E : \ell^{\T_i}(e) < a_\mu\}$.   \label{step:trivial_cut}
        \item Maintain the minimum of all cuts from \ref{step:1-resp_cut} and \ref{step:trivial_cut}. \label{step:min}
    \end{enumerate}
     We output the result of Step~\ref{step:min}.

    \textbf{Correctness.}
    First, we note that the minimum of Step~\ref{step:min} cannot be below the min-cut. Each of the values from Step~\ref{step:1-resp_cut} corresponds to a cut in the graph, hence can only over-estimate the min-cut. Each value from Step~\ref{step:trivial_cut} is a cut in a contracted version of $G$, hence corresponds to a cut\footnote{Technically, there are also self-loops in $G_a$, hence the value of a cut in $G_a$ can be bigger than the value of the corresponding cut in $G$. However, this can only lead to further over-estimation.} in $G$ or equals $\infty$, and thus can only over-estimate the min-cut of $G$.
    
    Now let $i$ such that $\lambda_t\in [2^i,2^{i+1})$. By \Cref{thm:CutExistence}, we know that if 1-respecting cuts of the tree-packing $\T_i$ do not give a minimum-cut, then for $\mu=\lambda_t$, we have some trivial cut in $G_{a_\mu}$ which is a minimum-cut. Hence either Step~\ref{step:1-resp_cut} or Step~\ref{step:trivial_cut} outputs the value of the min-cut. 

    \textbf{Update time.} We analyze the update time of each of the four parts of the algorithm. 
    \begin{enumerate}
        \item We can maintain tree-packings with $R(n)=O(\lambda_{\max}^5\log^2 m)$ worst-case recourse per tree-packing, by \Cref{lm:TP_recourse_impr}.
        For the update time, we need to maintain minimum spanning trees, where we have $R(n)$ updates to these trees. We can maintain a minimum spanning tree with $\tilde O(\sqrt n)$ worst-case update time~\cite{Frederickson85,EppsteinGIN97}, so this takes $\tilde O(\lambda_{\max}^5\sqrt{n})$ worst-case update time in total.
        \item Maintaining the minimal 1-respecting cut take $\tilde O(\sqrt{m})=\tilde O(\sqrt{\lambda_{\max}n})$ update time by \Cref{lm:1_resp}. Since we have $R(n)$ updates for these trees, this takes $\tilde O(\lambda_{\max}^{11/2}\sqrt n)$ worst-case update time. 
        \item Each of the tree-packings $\T_i$ decides which edges loads cross the value $a_{\mu}$ for the trivial cuts corresponding to $\mu \in [2^i,2^{i+1})$. Let $P_{\mu}(n)$ denote the number of edges crossing $a_\mu$ in \Cref{lm:trivial_cuts}. We remark that $\sum_{\mu \in [2^{i},2^{i+1})} P_{\mu}(n)=R(n)$, since any update to the tree-packing can change $\ell^\T_i(e)$ for one edge. This change is as follows $\ell^{\T_i}(e) = \tfrac{L^{\T_i}(e)}{|\T_i|}\rightarrow \tfrac{L^{\T_i}(e)\pm 1}{|\T_i|}$, so a change of size $\tfrac{1}{|\T_i|}=O(\tfrac{1}{\mu^3})$ for each $\mu \in [2^{i},2^{i+1})$. Since this change is so small, it can cross at most one value $a_\mu = \tfrac{2}{\mu}-\tfrac{1}{2\mu^2}$ for a given $i$.

        Since \Cref{lm:trivial_cuts} takes $O((1+P_{\mu}(n))\sqrt n)$ worst-case update time, we get $\tilde O(\lambda_{\max}^5\sqrt n)$ worst-case update time, when summing over all $i$. 
         
        \item We can implement this step with a min-heap, which has $O(\lambda_{\max}\log m)$ worst-case update time.    
    \end{enumerate}
    We conclude that we have $\tilde O(\lambda_{\max}^{11/2}\sqrt n)$ worst-case update time.

     \textbf{Returning the cut edges}.
     If the cut is from Step~\ref{step:1-resp_cut}, we can return the edges of the cut with $O(\log m)$ time per edge by \Cref{lm:1_resp}. The caveat is that these edges are only correct if we did not apply the connectivity sparsifier. This changes the factor $\sqrt{\lambda_{\max}n}$ to $\sqrt{m}$ in the update time. 
     If the cut is from Step~\ref{step:trivial_cut}, we can return the edges with $O(\log m)$ time per edge by \Cref{lm:trivial_cuts}. 
\end{proof}

\subsection{General \texorpdfstring{$\lambda$}{lambda}}\label{sc:mincut_gen}
We use the following result of Goranci, Henzinger, Nanongkai, Saranurak, Thorup, and Wulff{-}Nilsen~\cite{GoranciHNSTW23}.
\begin{lemma}[Corollary 4.1 \cite{GoranciHNSTW23}]\label{thm:mincut_high}
    There exists a deterministic fully dynamic algorithm that, given a simple, unweighted, undirected graph $G = (V, E)$ with $m$ edges and a parameter $\phi\in (0, 1)$,  maintains a min-cut estimate $\mu(G)$ in $\tilde O(1/\phi^3 +\phi m)$ amortized time per edge insertion or deletion. If $\phi \geq 240/\delta$, where $\delta$ is the minimum degree then the min-cut estimate is correct, i.e., $\mu(G) = \lambda(G)$.

    It can return the edges of the cut in $O(\lambda\log m)$ time with the same update time. 
\end{lemma}

In~\cite{GoranciHNSTW23}, they balance this with~\cite{Thorup07} to obtain an update time of $\tilde O(\tau^{29/2}\sqrt{n}+m/\tau)=\tilde O(m^{29/31}n^{1/31})=\tilde O(m^{1-1/31})$ for $\tau=m^{2/31}n^{-1/31}$. When we balance it with our \Cref{thm:mincut_par}, we obtain the following result. We note that since \Cref{thm:mincut_high} only holds for simple graph, our combined result is also restricted to that case. 

\mincutCombi*
\begin{proof}
    Let $\tau$ be a parameter to be determined later. We run the algorithm of \Cref{thm:mincut_par} with $\lambda_{\max}=\tau+1$, denoted by Algorithm~$\mathcal A$, and the algorithm of \Cref{thm:mincut_high} with $\phi=240/\tau$, denoted by Algorithm~$\mathcal B$. 
    We note that Algorithm~$\mathcal A$ is correct when $\lambda \leq \tau +1$ and Algorithm~$\mathcal B$ is correct when $\lambda \geq \tau$ (using that $\lambda \leq \delta)$. 
    We decide which output to take as follows. 
    \begin{itemize}
        \item If the current value of $\lambda$ is at most $\tau$, we use Algorithm~$\mathcal A$ for the next update.
        \item If the current value of $\lambda$ is at least $\tau+1$, we use Algorithm~$\mathcal B$ for the next update.
        \item Using an efficient static min-cut computation, e.g.,~\cite{KawarabayashiT19}, we can compute the initial value and decide whether we start with Algorithm $\mathcal A$ or $\mathcal B$.        
    \end{itemize}
    Now correctness directly follows from the guarantees on Algorithm $\mathcal A$ and $\mathcal B$.

    \textbf{Update time.} We have amortized update time $\tilde O(\tau^{11/2}\sqrt{n}+\tau^3+m/\tau)=O(\tau^{11/2}\sqrt n +m/\tau)$. Balancing this gives $\tau=m^{2/13}n^{-1/13}$, hence we have update time $\tilde O(m^{11/13}n^{1/13})=\tilde O(m^{1-1/13})$.

    To optimize for large values of $m$, we can use the connectivity sparsifier of Nagamochi and Ibaraki~\cite{NagamochiI92} again, which brings $m$ down to $\min\{m,\lambda n\}$. If we use this in the regime $\lambda \leq \tau'$, we obtain running time 
    \[ \tilde O((\tau'n)^{11/13}n^{1/13})+ \tilde O(\tau'^3 + m/\tau').\]
    For different choices of $\tau'$ we can get the following running times  (up to polylogarithmic factors):
    \begin{itemize}
        \item $m^{11/24}n^{1/2}+m^{13/8}n^{-3/2}$;
        \item $n^{9/7}+mn^{-3/7}$;
        \item $m^{11/52}n^{12/13}+m^{3/4}$.
    \end{itemize}
    We remark that the last one is always smaller than $n^{11/26+12/13}+n^{3/2}=n^{35/26}+n^{3/2}=O(n^{1.5})$, since $G$ is simple.

    \textbf{Returning the cut edges}.
    We take the version of \Cref{thm:mincut_par} that can return the cut edges, which has $\tilde O(\lambda_{\max}^5\sqrt m)$ amortized update time, and \emph{do not} apply the connectivity sparsifier. This gives an amortized update time of $\tilde O(m^{1-1/12})$. 
\end{proof}

%% file: arboricity.tex
\section{Arboricity}\label{sc:arb}
In this section, we first show a structural result: a relation between the fractional arboricity and the ideal relative loads of a tree-packing, see \Cref{sc:arb_struc}. We then give a deterministic dynamic algorithm that is efficient for small values of $\alpha$, see \Cref{sc:arb_wUB}. For simple graphs, we can combine this with the state of the art for densest subgraph approximation, see \Cref{sc:arb_simple}. For multi-graphs, we need to downsample the high $\alpha$ regime to low $\alpha$ regime. While this is relatively straightforward against an oblivious adversary (\Cref{sc:arb_multi_obl}), it is much more involved against an adaptive adversary (\Cref{sc:arb_multi_adap}).

\subsection{Structural Result}\label{sc:arb_struc}
The idea is to show that $(\min \ell^{*}(e))^{-1} = \alpha(G)$. Then we can estimate $\alpha(G)$ by simply taking $\alpha_{\rm{est}} = (\min \ell(e))^{-1}$, with $\ell$ some good approximation of $\ell^*$. For the integral arboricity, $\lceil \alpha \rceil$, a similar result follows already from \cite{Tutte61,Nash61,NashWilliams64}. We achieve this, more nuanced result by using the language of ideal load decompositions from~\cite{Thorup07}.
In particular, we use the following observation. We provide a proof for completeness. 
\begin{observation}[\cite{Thorup07}]\label{lm:Phi_inc}
    For each $S\in \P^*$, we have $\Phi_{G[S]} \geq \Phi$.
\end{observation}
\begin{proof}
    We will prove this by contradiction. Suppose $\P$ is a partition of $S$ such that $\rm{part\_val}(\P) < \Phi$. Let $\P'= (\P^*\setminus \{S\})\cup \P$ be a partition of $V$, then we see that 
    \[
        \rm{part\_val}(\P') =  \frac{|E(G/\P')|}{|\P'|-1} =  \frac{|E(G/\P^*)|+|E(G[S]/\P)|}{|\P^*|-1+|\P|-1} < \frac{\Phi(|\P^*|-1)+\Phi(|\P|-1)}{|\P^*|-1+|\P|-1}= \Phi,
    \]
    using that $ \frac{|E(G/\P^*)|}{|\P^*|-1}= \Phi$ and $\frac{|E(G[S]/\P)|}{|\P|-1}=\rm{part\_val}(\P) < \Phi$.
\end{proof}

Now we can show the main result.

\ArbStructural*
\begin{proof}
We start by showing $\alpha(G) \geq (\min_{e\in E} \ell^{*}(e))^{-1}$. We note that $(\min_{e\in E} \ell^{*}(e))^{-1}= \max_i \Phi_i$, ranging over all partitions appearing in the ideal partitioning. Denote the graph where this is maximized by $G^*$, and its vertex set by $X$.
By \Cref{lm:Phi_inc}, we have that it does not further partition, hence $\P^*=X$, so $G^*/\P^*=G^*$ and $G^*=G[X]$. Now we can conclude
\begin{align*}
    (\min_{e\in E} \ell^{*}(e))^{-1} &= \max_i \Phi_i = \frac{|E(G^*/\P^*)|}{|\P^*|-1} = \frac{|E(G[X])|}{|X|-1}\\
    &\leq \max_{Y \subseteq V} \frac{|E(G[Y])|}{|Y|-1} = \alpha(G).
\end{align*}
Next, we show that $\alpha(G) \leq (\min_{e\in E} \ell^{*}(e))^{-1}$.
Let $Y\subseteq V$ be any subset. We will show that  
\begin{equation*}
    \frac{|E(Y)|}{|Y|-1} \leq (\min_{e\in E} \ell^{*}(e))^{-1}.
\end{equation*}
Hence $\alpha(G) = \max_{Y\subseteq V} \frac{E(Y)}{|Y|-1} \leq (\min_{e\in E} \ell^{*}(e))^{-1}.$

The arguments we use here are very similar to the proof of \Cref{lma:trivexist}, where we present them with more detail. We inspect the definition of ideal relative loads: let $\P_0$ the partition for $G$ with $\rm{part\_val}(\P_0)=\Phi$. Now we recurse on the subgraphs $G[S]$ for each $S\in \P_0$. We write the total decomposition as $\P_i$ the partition on $G[S_i]$, where each $S_i\in \P_j$ for some $j< i$, for $i\in \mathcal{I}$ for some index set $\mathcal I$.  Now we have the disjoint union:
\begin{equation*}
    E(G) = \bigcup_{i\in \mathcal I} E(G[S_i]/\P_i).
\end{equation*}
So in particular we have 
\begin{align*}
    |E(Y)| = \sum_{i\in \mathcal I} |E(Y) \cap E(G[S_i]/\P_i)| \leq \sum_{i\in \mathcal I} \Phi_i (|V_{Y,i}| -1),
\end{align*}
where $\Phi_i$ is the partition value of $\P_i$, and $V_{Y,i}=Y\cap V(G[S_i]/\P_i)$. 
We use here that $|E(Y) \cap E(G[S_i]/\P_i)|\leq \Phi_i (|V_{Y,i}| -1)$, so let us prove that. Suppose $|E(Y) \cap E(G[S_i]/\P_i)|> \Phi_i (|V_{Y,i}| -1)$. Then we create the alternative partition $\P_i'$ where we contract all edges of $E(Y) \cap E(G[S_i]/\P_i)$. This gives partition value
\begin{align*}
    \Phi_i' &\leq \frac{|E(G[S_i]/\P_i)|-|E(Y) \cap E(G[S_i]/\P_i)|}{|V_i|-1 -(|V_{Y,i}| -1)} < \frac{|E(G[S_i]/\P_i)|-\Phi_i (|V_{Y,i}| -1)|}{|V_i|-1 -(|V_{Y,i}| -1)}\\
    &= \frac{\Phi_i(|V_i|-1)-\Phi_i (|V_{Y,i}| -1)|}{|V_i|-1 -(|V_{Y,i}| -1)} = \Phi_i.
\end{align*}
But that contradicts the choice of $\P_i$ being the minimal partition. 

Next, we prove that 
\begin{equation*}
    \sum_{i\in \mathcal I} (|V_{Y,i}| -1) = |Y|-1.
\end{equation*}
We divide $\mathcal I$ in different depths of recursion: let $\mathcal I_j$ be such that for $i\in \mathcal I_j$ we have $S_i\cap Y$ in recursion depth $j$ (i.e., it is in $j$ partitions before)\footnote{For ease of notation, we let a singleton cluster $\{v\}$ again partition into $\{v\}$. This has no effect on the sum over $|V_{Y,i}| -1$, but ensures that all vertices make it to the lowest depth.}. Let $r$ be the total recursion depth. Then we can write
\begin{align*}
     \sum_{i\in \mathcal I} (|V_{Y,i})| -1) &= \sum_{j=0}^r \sum_{i\in \mathcal I_j} (|V_{Y,i}| -1) \\
     &= \sum_{j=0}^r \left( \sum_{i\in \mathcal I_j} |V_{Y,i}|\right) - |I_j|\\
     &=  \sum_{j=1}^r |\mathcal I_{j}| -|\mathcal I_{j-1}| \\
     &= |\mathcal I_r| -|\mathcal I_{0}|\\
     &= |Y|-1.
\end{align*}
Now we use this as follows:
\begin{align*}
    \frac{|E(Y)|}{|Y|-1} &\leq \frac{\sum_{i\in \mathcal I} \Phi_i (|V_{Y,i}| -1)}{|Y|-1} \\
    &\leq \left(\max_{i\in \mathcal I}\Phi_i\right) \frac{\sum_{i\in \mathcal I}|V_{Y,i}| -1}{|Y|-1} \\
    &= \max_{i\in \mathcal I}\Phi_i\\
    &= (\min_{e\in E} \ell^{*}(e))^{-1},
\end{align*}
where the last equality holds by definition of $\ell^*$.

\end{proof}

Next, we use the fact that a greedy tree-packing approximates an ideal packing well to give an approximation of the fractional arboricity.

\begin{lemma}\label{lm:arb_apx}
     A greedy tree-packing $\T$ with $|\T|\geq \tfrac{24\alpha^2\log m }{\lambda\eps^2}$ trees satisfies, for $\eps \in (0,1)$, 
     \begin{equation*}
         \left|\frac{1}{\min_{e\in E}\ell^\T(e)}- \alpha\right| \leq \eps\alpha.
     \end{equation*}
\end{lemma}
\begin{proof}
    The tree-packing $\T$ contains 
    \[
        \frac{24\alpha^2\log m }{\lambda \eps^2}\geq  \frac{6(1+\eps)^2\alpha^2\log m }{\lambda \eps^2}=\frac{6\log m \lambda}{\left(\frac{\lambda}{\alpha(1+\eps)}\eps\right)^2}
    \]
    trees. We note\footnote{Here we use that $\lambda \leq 2\alpha$. One way to see that is $\lambda \leq 2\Phi = 2\tfrac{1}{\max_{e\in E}\ell^*(e)}\leq 2\tfrac{1}{\min_{e\in E}\ell^*(e)}=2\alpha$.} that $\tfrac{\lambda}{\alpha(1+\eps)}\eps \leq \tfrac{2}{1+\eps}\eps <2 \eps<2$, so we can apply \Cref{lm:tree_packing_thorup} with $\eta= \tfrac{\lambda}{\alpha(1+\eps)}\eps$ to obtain 
    \[
        \left|\min_{e\in E}\ell^\T(e)-\min_{e\in E}\ell^*(e)\right|\leq \frac{\eps}{\alpha(1+\eps)}.
    \]
    Now we see that
    \begin{align*}
        \frac{1}{\min_{e\in E}\ell^\T(e)}- \alpha &\leq \frac{1}{\min_{e\in E}\ell^*(e)-\frac{\eps}{\alpha(1+\eps)}}- \alpha\\ &=
        \frac{1}{\frac{1}{\alpha}-\frac{\eps}{\alpha(1+\eps)}}- \alpha\\
        &= \left(\frac{1}{\frac{1+\eps}{1+\eps}-\frac{\eps}{(1+\eps)}}-1\right) \alpha\\
        &= \eps\alpha. 
    \end{align*}
    The proof that $\alpha -\tfrac{1}{\min_{e\in E}\ell^\T(e)}\leq \eps\alpha$ is analogous.
\end{proof}

\subsection{Dynamic Algorithm for Bounded \texorpdfstring{$\alpha$}{alpha}}\label{sc:arb_wUB}
In this section, we show how to dynamically maintain the estimate $\tfrac{1}{\min_{e\in E}\ell^\T(e)}$, giving us the arboricity estimate. We start with a warm-up giving the result almost directly by plugging in a deterministic minimum spanning tree algorithm. 

\subsubsection{Warm-Up}\label{sc:arb_wUB_warm_up}

\begin{lemma}\label{lm:arb_warm_up}
    There exists a deterministic dynamic algorithm that, given an unweighted, undirected (multi-)graph $G=(V,E)$, maintains a $(1+\eps)$-approximation of the fractional arboricity $\alpha$ when $\alpha\leq \alpha_{\max}$ in $O(\alpha_{\max}^4\log^6 m / \eps^4)$ amortized update time or a Las Vegas algorithm with $O(\alpha_{\max}^4 m^{o(1)} / \eps^4)$ worst-case update time.
\end{lemma}
\begin{proof}
    By \Cref{lm:arb_apx}, all we need to do is maintain $\Theta (\tfrac{\alpha_{\max}^2\log m }{\lambda \eps^2})=O(\tfrac{\alpha_{\max}^2\log m }{\eps^2})$ greedy spanning trees, and then maintain the minimum over $\ell^\T(e)$.
    The latter can simply be done by maintaining a min-heap, which has an update time of $O(\log m)$. 

    For the former, we note that any edge insertion or deletion can lead to overall $O\left( \left( \tfrac{\alpha_{\max}^2\log m }{\eps^2}\right)^2 \right)$ updates to the spanning trees~\cite{ThorupK00}. 
    We can use a deterministic dynamic minimum spanning tree algorithm with 
    \begin{itemize}
        \item $O(\log^4 m)$ amortized update time deterministically \cite{HolmLT01}; or
        \item $m^{o(1)}$ worst-case update time (Las Vegas algorithm) \cite{NanongkaiSW17}. 
    \end{itemize}
    Multiplying these update times with the number of spanning tree updates per insertion/deletion gives the result. Note that both subsume the $O(\log m)$ update time for maintaining the min-heap. 
\end{proof}

\subsubsection{Recourse in Tree-Packing}\label{sc:arb_recourse}
Next, we show how we can bound the recourse in the tree-packing to shave a factor $\alpha_{\max}$. This is similar to \Cref{lm:TP_recourse_impr}, but has additional complications: to bound the recourse we need to guarantee a high min-cut, and we need to show we can keep an artificially high min-cut in a graph with certain arboricity. Note that there are even disconnected graphs with linear arboricity, so this is a non-trivial adaptation. 
% We recommend the reader to recap \Cref{lm:TP_recourse_impr} before reading this proof. 

\arboricityDetNew*
\begin{proof}
    By \Cref{lm:arb_apx}, all we need to do is maintain $\Theta (\tfrac{\alpha^2\log m }{\lambda\eps^2})$ greedy spanning trees, and then maintain the minimum over $\ell^\T(e)$. As opposed to \Cref{lm:arb_warm_up}, we do not do this once with $\alpha_{\max}$, but keep $\log \alpha_{\max}$ copies corresponding to different values of $\alpha$. The goal is to establish $O(|\T|^2/\alpha)$ recourse instead of the trivial $O(|\T|^2)$ recourse. We do this by utilizing that the recourse is $O(|\T|^2/\lambda)$, and artificially increase the min-cut to $\Theta(\alpha)$. This simultaneously means we only need to maintain $|\T|=\Theta (\tfrac{\alpha\log m }{\eps^2})$ greedy spanning trees.
    
    To be precise, we maintain tree-packings $\T_1, \T_2, \dots, \T_{\log \lambda_{\max}}$, where $|\T_i|=O(2^{i} \log m/\eps^2)$. We only need the packing $\T_i$ when $\alpha\in [2^i,2^{i+1})$. If $\alpha \notin [2^i,2^{i+1})$, we still need the update times to hold, but we do not need the output to be correct.  
    Each tree-packing $\T_i$ is the tree-packing of some graph $G_i$, where we show that $\alpha=\alpha(G_i)$ if $\alpha\in [2^i,2^{i+1})$. As opposed to \Cref{lm:TP_recourse_impr}, we do not guarantee the stronger statement that $G_i=G$ in this case. 

    First, we assume that $G_i$ initially has min-cut $\lambda \geq 2^{i-1}$. 
    This means that initially $\ell^{\T_i}(e) \leq \ell^*(e)+\eps/2^{i-1}\leq 1/\Phi_{G_i}+1/2^{i-1}\le 2/2^{i-1}+1/2^{i-1}=O(1/2^i)$, using \Cref{lm:tree_packing_thorup} with $\eta=1$. Throughout the updates, we keep an edge $e$ as a virtual edge if deleting $e$ would cause any edge $e'$ to reach $\ell^{\T_i}(e') \geq 16/2^i$. By the same arguments as in \Cref{lm:TP_recourse_impr}, this gives amortized recourse $O(|\T_i|^2/2^i)$ for $\T_i$.
    We delete a virtual edge $e$ if $\ell^{\T_i}(e)< 8/2^i$ after some update. This cannot lead to the load of any other edge rising above $16/2^i$, see Claim 1 in \Cref{lm:TP_recourse_impr}.
    Since any virtual edge was inserted as a real edge at some point, we can amortize the cost of this deletion. 
    Next, we show that the virtual edges do not interfere with the arboricity if $2^i \leq \alpha$, in which case $8/2^i \ge 8/\alpha$. We do this with two claims. 

    \begin{itemize}
        \item[Claim 1.] If $\ell^{\T_i}(e)\geq 8/\alpha$ then $\ell^{*}(e) \geq 4/\alpha$.\\
        \textit{Proof.} By \Cref{lm:tree_packing_thorup} we have $\ell^*(e)\geq \ell^{\T_i}(e)-\eps/\alpha \ge 4/\alpha$. 
    \end{itemize}

    \begin{itemize}
        \item[Claim 2.] Let $\ell^{*}(e) > \tfrac{1}{\alpha}$ and let $S^*\subseteq V$ be such that $\alpha=\tfrac{|E(S^*)|}{|S^*|-1}$. Then $e \notin S^{*}$.\\
        \textit{Proof.} Consider the partition induced by the last level of the ideal load decomposition $\mathcal{P}^{*}$, i.e., the classes of $\mathcal{P}^{*}$ are the connected components induced by edges $f$ with $\ell^{*}(f) = \tfrac{1}{\alpha}$. 
        Note that by arguments similar to~\Cref{lma:trivexist} and~\Cref{thm:arb_vs_packing}, we have $1)$ that $e \in E(G/\mathcal{P}^{*})$ and $2)$ that $\frac{|E(G/\mathcal{P}^{*})|}{|\mathcal{P}^{*}|-1} < \alpha$. 
        Now suppose the classes of $\mathcal{P}^{*}$ are $P_1, \dots, P_t$. W.l.o.g., possibly by renumbering, we can assume that $|S^{*}\cap P_{j}| \neq \emptyset$ for all $1\leq j \leq i$ and $|S^{*}\cap P_{j}| = \emptyset$ for all $j > i$. If $i = 1$, we are done, so assume $i > 1$.
        Observe first that the graph induced by $P_1, \dots, P_i$ in $G/\mathcal{P}^{*}$ contains strictly less than $\alpha(i-1)$ edges. 
        Indeed, suppose for contradiction this is not so. If $i = t$ then the supposition contradicts $2)$, and if $i < t$, then one would need to contract $P_1, \dots, P_{i}$ in every level except for the last level of the ideal load decomposition, thus contradicting the choice of $\P^{*}$. 

        Finally, we claim that for some $s \in [i]$, we must have $|E(G[S^{*}\cap P_{s}])| > \alpha(|S^{*}\cap P_{s}|-1)$ contradicting that $G$ has arboricity $\alpha$. 
        Indeed, observe that 
        \begin{align*}
            \alpha = \frac{|E(S^*)|}{|S^*|-1} &\le \frac{|E(P_1, P_2, \dots, P_{i})|+\sum_{j = 1}^{i} |E(S^{*}\cap P_{j})|}{|S^*|-1} \\
            &= \frac{|E(P_1, P_2, \dots, P_{i})|+\sum_{j = 1}^{i} |E(S^{*}\cap P_{j})|}{(i-1)+ \sum_{j = 1}^{i} (|S^{*}\cap P_{j}|-1)}\\
            &< \frac{\alpha(i-1)+\sum_{j = 1}^{i} |E(S^{*}\cap P_{j})|}{(i-1) + \sum_{j = 1}^{i} (|S^{*}\cap P_{j}|-1)},
        \end{align*}
        where we used that $|S^{*}| = \sum_{j = 1}^{i} |S^{*}\cap P_{j}| = i + \sum_{j = 1}^{i} (|S^{*}\cap P_{j}|-1)$. 
        Finally, the claim follows by observing that the final inequality is false if $ |E(S^{*}\cap P_{j})| \leq \alpha(|S^{*}\cap P_{j}|-1)$ for all $1\leq j \leq i$. 
    \end{itemize}

    Combining the two claims above, we see that any edge with $\ell^{\T_i}(e)\ge 8/2^i \ge 8/\alpha$ is not part of a subgraph for some choice of $S^*$ achieving the arboricity. So $\alpha(G_i) \leq \alpha(G)$. Since $G\subseteq G_i$, we also have $\alpha(G) \leq \alpha(G_i)$, thus we conclude $\alpha(G)=\alpha(G_i)$. This means that it suffices to approximate the arboricity in $G_i$.
 
    Now we discuss an initial initialization step to guarantee that we have $\lambda \geq 2^{i-1}$. We do this by \emph{vertex} insertions. We initialize all our data structures on a graph with $n$ vertices, but no edges. We add an edge only if both endpoints have degree at least $2^{i-1}$. Note that this corresponds to inserting the vertex with its edges when its degree reaches this boundary. This means we insert at most two vertices at any time. 
    When we insert such a vertex, and it has less than $2^{i-1}$ edges towards the other vertices already present, we add the remainder in virtual edges. This guarantees min-cut $\lambda \geq 2^{i-1}$ at any time. 
    Observe that an edge (or its virtual representative) is added and deleted at most twice. 

    \begin{itemize}
    \item[Claim 3.] If $\deg(v)\leq 2^{i-1}$, then $v\notin S^*$ for any $S^*\subseteq V$ such that $2^i\leq \alpha=\tfrac{|E(S^*)|}{|S^*|-1}$.\\
        \textit{Proof. } Let $S^*\subseteq V$ be such that $2^i\leq \alpha=\tfrac{|E(S^*)|}{|S^*|-1}$, and suppose $v\in S^*$
        \begin{align*}
            \frac{|E(S^*\setminus \{v\})|}{|S^*\setminus \{v\}|-1}= \frac{|E(S^*)|-|E(S,\{v\})|}{|S^*|-2} \geq \frac{2^i(|S^*|-1)-2^{i-1}}{|S^*|-2} > 2^i,
        \end{align*}
        which is a contradiction. 
    \end{itemize}

    Next, we consider the update time: each such vertex insertion takes $O(|\T_i|^2/2^i)$ amortized update time. We show that we can perform such a vertex insertion in $O(|\T_i|^2)$ time. Since it required $2^{i-1}$ insertions to reach this degree, this gives the amortized bound. 
    We consider each of the $|\T_i|$ trees in the packing and perform all updates simultaneously to it. This consists of the $2^{i-1}$ insertions, plus $O(2^{i-1} \cdot{} |\T_i|/2^{i} )=O(|\T_i|)$ recourse from the earlier trees (since each edge ends up in at most $O(|\T_i|/2^i)$ trees). In total this gives $O(|\T_i|(|\T_i|+2^{i-1}))=O(|\T_i|^2)$ time.

    As before, we use a deterministic dynamic minimum spanning tree algorithm with $O(\log^4 m)$ amortized update time deterministically \cite{HolmLT01} for maintaining the trees in each packing. 
    This gives update time $O(2^{i} \log^6 m/\eps^4)$ for each $\T_i$, hence the update times sum up to $O(\alpha_{\max}\log^6 m / \eps^4)$ amortized update time. 

    The worst-case result is obtained in the same way, with the exception that the insertions cannot be amortized over the edge insertions. Hence the worst-case update time is proportional to $|\T_i|^2$. Using that we maintain the minimum spanning trees in $n^{o(1)}$ time~\cite{NanongkaiSW17}, we obtain the result. 
\end{proof}

\subsection{Dynamic Result for Simple Graphs}\label{sc:arb_simple}
We will use the following approximation for the densest subgraph problem, where we define $\tilde \rho$ to be a $(1-\eps)$-approximation of $\rho$ when $(1-\eps)\rho\leq \tilde \rho \leq \rho$.\footnote{This is equivalent to a $(1+\eps)$- or a $(1\pm \eps)$-approximation by re-scaling. We use the $(1-\eps)$-version for ease of notation in the proof.}

\begin{lemma}[\cite{SawlaniW20,chekuri2024adaptive}]\label{lm:DSP}
    There exists a deterministic dynamic algorithm that, given an unweighted, undirected (multi-)graph $G=(V,E)$, maintains a $(1-\eps)$-approximation of the density $\rho$ in $O(\log^3 m / \eps^4)$ amortized update time or $O(\log^4 m / \eps^6)$ worst-case update time. 
\end{lemma}

We use this result in the high arboricity regime ($\alpha \geq 1/\eps$), where the density is a good approximation of the fractional arboricity for simple graphs. 

\arboricitySimple*
\begin{proof}
    The algorithm is as follows:
    \begin{itemize}
        \item Maintain a fractional arboricity estimate $\tilde \alpha$ with $\alpha_{\max}=\Theta(1/\eps)$ using \Cref{thm:arb_det};
        \item Maintain a $(1-\eps)$-estimate $\tilde\rho$ of the densest subgraph using \Cref{lm:DSP};
        \item If $\alpha\leq \Theta(1/\eps)$, output $\tilde{\alpha}$, else output $\tilde \rho$. 
    \end{itemize}

    By \Cref{thm:arb_det}, we know that $\tilde \alpha$ is a $(1+\eps)$-approximation of the fractional arboricity if $\alpha \leq \alpha_{\max}$. We next show that $\tilde \rho$ is a $(1+\eps)$-approximation if $\alpha \geq 1/\eps$. Because these ranges overlap, and an update can change the fractional arboricity by at most $1$, we can easily see when we should switch from one estimate to the other. 

    Note that we always have
    \[
        \tilde\rho \leq \rho = \max_{S\subseteq V} \frac{|E(S)|}{|S|} = \frac{|E(S^*)|}{|S^*|} \leq \frac{|E(S^*)|}{|S^*|-1} \leq \max_{S\subseteq V}\frac{|E(S)|}{|S|-1}= \alpha, 
    \]
    for some $S^*\subseteq V$. Now let $S^*$ be such that 
    \[
        \max_{S\subseteq V}\frac{|E(S)|}{|S|-1}=\frac{|E(S^*)|}{|S^*|-1}.
    \]
    Using that $G$ is simple, we now have for $\alpha \geq 1/\eps$
    \[
        |S^*| \geq \frac{|S^*| (|S^*|-1)}{|S^*|-1} \geq \frac{|E(S^*)|}{|S^*|-1} = \alpha \geq 1/\eps.
    \]
    We use this to see that 
    \begin{align*}
        \alpha - \rho &= \frac{|E(S^*)|}{|S^*|-1} - \max_{S\subseteq V} \frac{|E(S)|}{|S|} \leq \frac{|E(S^*)|}{|S^*|-1} -\frac{|E(S^*)|}{|S^*|}\\
        &= |E(S^*)|\frac{|S^*|-(|S^*|-1)}{|S^*|(|S^*|-1)}= \frac{\alpha}{|S^*|} \leq \eps\alpha. 
    \end{align*}
    Rearranging gives us 
    \[
        \tilde\rho \geq (1-\eps)\rho \geq (1-\eps)^2\alpha. 
    \]
    Setting $\eps\leftarrow \eps/3$ gives that $\tilde \rho$ is a $(1-\eps)$-approximation of the fractional arboricity~$\alpha$. 

    Concerning the update time, we need $O(\log^6 m / \eps^5)$ amortized update time for small $\alpha$, and $O(\log^3 m / \eps^4)$ amortized update time for large $\alpha$.

    For the small $\alpha$ regime we have $O(m^{o(1)}/\eps^6)$ worst-case update time, and $O(\log^4 m/\eps^6)$ worst-case update time for the high $\alpha$ regime. 
\end{proof}

\subsection{Downsampling for Multi-Graphs}\label{sc:arb_multi_obl}
As shown in \Cref{sc:arb_wUB}, for multi-graphs, the number of spanning trees we need to pack scales with~$\alpha_{\max}$. In this section, we show how to use a standard sampling technique (see e.g., \cite{McGregorTVV15}) to get rid of this dependency. 
% We note that although this requires randomness, it does work against an adaptive adversary. 

\begin{restatable}{theorem}{ThmArbMulti} \label{thm:arb_obl} 
    There exists a dynamic algorithm that, given an unweighted, undirected multi-graph $G=(V,E)$, maintains a $(1+\eps)$-approximation of the fractional arboricity $\alpha$ when in $O(\log^{7} m /\eps^{6})$ amortized update time or with $O(m^{o(1)}/\eps^{8})$ worst-case update time. The algorithm is correct with high probability against an oblivious adversary.   
\end{restatable}
\begin{proof}
    The idea is to maintain $\log m$ graphs, denoted by $H_i$, which are initialized by sampling each edge with probability $p_i=\frac{24c\log m}{2^i \eps^2}$ (for $i$ s.t. $p_i<1$). Now if $\alpha = \Theta(2^i)$, then $H_i$ has fractional arboricity $\Theta(\tfrac{\log m}{\eps^2})$. To compute this, on each graph $H_i$, we run the algorithm of \Cref{thm:arb_det} with $\alpha_{\max} = \Theta(\frac{\log m}{\varepsilon^2})$. First, we show that if $\alpha \in [2^{i-1},2^{i+2})$, then $H_i$ gives the correct answer. We do not prove anything about the output of the other graphs $H_j$, but simply disregard their output. 

    If $\alpha = O(\tfrac{\log m}{\eps^2})$, we just look at $G$ itself. So assume $\alpha = \Omega(\tfrac{\log m}{\eps^2})$. 
    Let $i$ such that $2^{i-1}\leq \alpha < 2^{i+2}$.
    We show correctness in three parts. 
    \begin{enumerate}
        \item  We show that for a set $S\subseteq V$ that satisfies w.h.p.\ $\alpha = \tfrac{|E(S)|}{|S|-1}$, that 
        $$\frac{1}{p_i}\frac{|E_{H_i}(S)|}{|S|-1} \geq (1-\eps)\alpha.$$
        \item We show that for \emph{any} $S\subseteq V$ we have w.h.p.\ that 
            $$\frac{1}{p_i}\frac{|E_{H_i}(S)|}{|S|-1} \leq (1+\eps)\alpha.$$
        \item In $H_i$, the fractional arboricity is at most $\alpha_{\max}= \Theta(\tfrac{\log m}{\eps^2})$. 
    \end{enumerate}
    \textbf{Part 1.} Since we sample each edge with probability $p_i$, we immediately have that $\E[ |E_{H_i}(S)|]=p_i |E(S)|$. Now by a Chernoff bound we obtain 
    \begin{align*}
        \mathbb P\left[ \frac{1}{p_i}\frac{|E_{H_i}(S)|}{|S|-1}< (1-\eps)\alpha\right] &= \mathbb P\left[ |E_{H_i}(S)|<(1-\eps)p_i\alpha(|S|-1)\right]\\
        &= \mathbb P\left[ |E_{H_i}(S)|<(1-\eps)p_i|E(S)|\right]\\
        &\leq e^{- p_i|E(S)|\eps^2/2}\\
        &\leq e^{- \frac{24c\log m}{2^i \eps^2}\alpha(|S|-1) \eps^2/2}\\
        &\leq m^{-c(|S|+2)},
    \end{align*}
    using that $\alpha \geq 2^{i-1}$. This allows us to union bound over all sets $S$ of size $|S|$, $O(n^{|S|})$ many, all sizes $|S|$, for $n$ sizes, and $m$ updates.

    \textbf{Part 2.}
    Also this follows by a Chernoff bound. Here we use the upper bound on the expectation. 
    \begin{align*}
        \mathbb P\left[ \frac{1}{p_i}\frac{|E_{H_i}(S)|}{|S|-1}> (1+\eps)\alpha\right] &= \mathbb P\left[ |E_{H_i}(S)|>(1+\eps)p_i\alpha(|S|-1)\right]\\
        % &= \mathbb P\left[ |E_{H_i}(S)|>(1+\eps)p_i|E(S)|\right]\\
        &\leq e^{- p_i\alpha(|S|-1)\eps^2/3}\\
        &\leq e^{- \frac{24c\log m}{2^i \eps^2}\alpha (|S|-1)\eps^2/3}\\
        &\leq m^{-c(|S|+2)},
    \end{align*}
    using that $\alpha \geq 2^{i-1}$. Again, this allows us to union bound over all sets $S$ of size $|S|$, $O(n^{|S|})$ many, all sizes $|S|$, for $n$ sizes, and $m$ updates.

    \textbf{Part 3.}
    Since $\alpha < 2^{i+2}$, we get by Part 2. that the fractional arboricity in $H_i$ is w.h.p.\ at most $p_i\cdot (1+\eps)\alpha \leq \frac{24c\log m}{2^i \eps^2}\cdot 2^{i+3}=O(\tfrac{\log m}{\eps^2})$.  
    
    Note that since we assume the adversary to be oblivious, the probabilistic guarantees from above hold for any graph, in particular for the graph after $t$ updates. 

    To see which $H_i$ to look at, we use the approximation from before the update:
    \begin{itemize}
        \item If the current estimate of $\alpha$ is at most $\Theta(\tfrac{\log m}{\eps^2})$ we consider the estimate on $G$.  
        \item If the current estimate of $\alpha$ lies in $[2^i,2^{i+1})$ we use the estimate from $H_i$ for the next update. 
        \item Using an efficient static algorithm, e.g., \cite{WorouG16}, we can compute an initial approximation to decide with which $H_i$ to start. 
    \end{itemize}
    Since the estimate of $H_i$ is correct up to a $(1+\eps)$ factor, and the arboricity can change by at most $1$, we have that if before the update our estimate $2^{i}\geq \tilde \alpha <2^{i+1}$, then after the update $\alpha \leq (1+\eps)\tilde \alpha+1<(1+\eps)2^{i+1}+1\leq 2^{i+2}$, and similarly $2^{i-1}\leq \alpha$. Hence the estimate from $H_i$ is a $(1+\eps)$-approximation.

    \textbf{Update time.}
    Whenever an edge gets deleted from $G$, we delete it from $H_i$, if it appears there. Whenever an edge gets inserted to $G$, we insert it in $H_i$ with probability $p_i$.
    
    By simply maintaining the data structures on each $H_i$, we obtain an algorithm that works against an oblivious adversary. This algorithm has amortized update time $O( \log^{7} m /\eps^{6})$ or worst-case update time $O(m^{o(1)}/\eps^{8})$ for each $H_j$. 

    Now we note that $p_{i+1}=p_i/2$, so the probability that an update needs to be processed in $H_{i+1}$ is half as big as the probability that it needs to processed in $H_i$. In the first (relevant) $H_i$, edges are sampled with probability $\leq 2^{-1}$. So we have $m/2$ updates to this $H_i$ in expectation, and $m$ w.h.p.\ by a Chernoff bound. Using the same argument for each subsequent $H_i$, we get $m+m/2+m/4+\dots=2m$ updates in total w.h.p. Hence running the algorithm for the $\log m$ copies has as many updates as for one copy, and we obtain the result. 
\end{proof}

\subsection{Downsampling Against an Adaptive Adversary}\label{sc:arb_multi_adap}
For our algorithm against an adaptive adversary, we use the same set-up as before: again, we have $\log m$ sampled graphs $H_i$ for different regimes of $\alpha$. However, we need to resample more often to fend of adversarial attacks. We first consider a naive way of doing this, and then describe a more involved process. A similar effort has been made in~\cite{Bhattacharya+24} for dynamic matching. However, their algorithm only works against an output-adaptive adversary, that allows updates to depend only on the algorithm's output. Furthermore, we believe the techniques of this section can be used to simplify the dynamic matching result of \cite{Wajc20}.
% \tijn{I've put David's references here now. Maybe we want to move them to the intro as well?}

\paragraph{Naive Resampling.}
An adaptive adversary can attack our sampling, for example by deleting our sampled edges. This forces us to introduce some form of resampling. The most straightforward way to do this, is that for any inserted or deleted edge $uv$. We resample all edges adjacent to either $u$ or $v$. This guarantees that we cannot over- or under-sample the edges of one vertex, and one can show that this is enough to preserve the fractional arboricity. The downside is that this approach is slow: each vertex can have up to degree $n$, even in a graph with bounded fractional arboricity. Instead, we will assign ownership of each edges to a vertex, such that we can sample more efficiently. 

\paragraph{Fancy Resampling.}
First we compute out-orientations such that each vertex has at most $\alpha(1+\eps)$ out-edges.  To each vertex, we assign its out edges. So every edge is assigned to a vertex, and each vertex has at most $\alpha(1+\eps)$ edges assigned to it. Now upon an update to $uv$, we recompute the out-orientation, and then resample all out-edges of each vertex for which its set of out-edges changed. 

There is one more complication that we address in the next paragraph where we provide the full description of the algorithm: we cannot afford to resample in each $H_i$. However, where resampling is too costly, it is also unnecessary. 

\paragraph{Algorithm Description.}
The main algorithm consists of the following steps. 
\begin{enumerate}
    \item Maintain out-orientation with maximum out-degree $(1+\eps)\alpha$. 
    \item Maintain the algorithm of \Cref{thm:arb_det} with $\alpha_{\max}=\Theta(\log m/\eps^4)$.
    \item Upon update to $uv$, i.e., an out-edge of $u$, do for each $i= 1, 2, \dots, \log m$:
    \begin{enumerate}
        \item Resample for $H_i$ all out-edges of $u$ with probability $p_i :=  8(c+3)\tfrac{\log m}{2^i \eps^4}$, using e.g.~\cite{BringmannP12}.
        \item If this leads to at most $\Theta(\tfrac{\log m}{\eps^4})$ changes to $H_i$, process them as updates in $H_i$'s fractional arboricity algorithm, \Cref{thm:arb_det} with $\alpha_{\max}=\Theta(\log m/\eps^4)$. If it leads to more changes, do nothing. 
    \end{enumerate} 
\end{enumerate}
We note that if resampling leads to more than $\Theta (\tfrac{\log m}{\eps^4})$ changes, that $p_i \alpha> p_i d^+_{H_i}(u)= \Omega (\tfrac{\log m}{\eps^4})$, so $\alpha = \Omega(2^i)$, hence $H_i$ is not the correct graph to look at. In particular, before we want to use $H_i$, we need $\alpha$ and hence $d^+_{H_i}(u)$ to decrease, which means we will resample its out-edges before we need the estimate from this graph. 

As before, in \Cref{thm:arb_obl}, we use the current output of the algorithm to see which $H_i$ we should use for the output after the next update.

\paragraph{Correctness.}
Next, we show that the given algorithm always maintains a correct estimate of the fractional arboricity. 

\begin{lemma}\label{lm:arb_aa_correct}
    At any moment, we maintain a $(1+\eps)$-approximation of the fractional arboricity. 
\end{lemma}
\begin{proof}
    Let $i$ such that $2^{i-1}\leq \alpha < 2^{i+2}$.
    Let $H_i$ be the corresponding sampled graph, where each edge is sampled with probability $p_i =  8(c+3)\tfrac{\log m}{2^i \eps^4}$. Since any action of the adversary leads to resampling all out-edges of all affected vertices, we can treat these resamplings as independent random events. We prove that with high probability our result holds. Hence even many attacks on the same vertex will not lead to breaking our guarantees. 

    For the proof, we consider $\tfrac{|E_{H_i}(S)|}{|S|-1}$, and aim to show this is roughly equal to $p_i\cdot \tfrac{|E(S)|}{|S|-1}$. To be precise, we need to show two parts:
    \begin{enumerate}
        \item  We show that for a set $S\subseteq V$ that satisfies $\alpha = \tfrac{|E(S)|}{|S|-1}$, that w.h.p.\
        $$\frac{1}{p_i}\frac{|E_{H_i}(S)|}{|S|-1} \geq (1-\eps)\alpha.$$ \label{part:1}
        \item We show that for any $S\subseteq V$ we have w.h.p.\ that 
    $$\frac{1}{p_i}\frac{|E_{H_i}(S)|}{|S|-1} \leq (1+\eps)\alpha.$$\label{part:2}
    \end{enumerate}

    Instead of working with the standard adaptive adversary, who can attack us by implicitly attacking our sampling, we show that the algorithm holds against a stronger adversary: at any moment in time the adversary is allowed to point at a vertex that needs to resample all its out-edges. In total the adversary can call for $O(m\eps^{-4}\log^3m)=O(m^2)$ resamples, since we have $m$ updates, each of which leads to $O(\eps^{-4}\log^3 m)$ recourse in the out-orientations.

    If the update sequence is longer than $m^2$, we can build new versions of the data structure in the background via period rebuilding. 
    
    For each part, we make a case distinction on the size of $S$ as follows
    \begin{enumerate}[a)]
        \item $|S| \leq 2/\eps$; or \label{part:a}
        \item $|S|>2/\eps$. \label{part:b}
    \end{enumerate}

    \textbf{Part~\ref{part:1}.}
    We note that at least $(1-\eps)|S|-1$ vertices $u\in S$ have at least $\eps\alpha$ out-neighbors in $S$. This follows from a simple pigeonhole argument: let $k$ denote the number of vertices that have at least $\eps\alpha$ out-neighbors. Then we see that 
    \begin{align*}
        &k(1+\eps)\alpha +(|S|-k)\eps\alpha \geq \alpha(|S|-1)\\
        \iff &k\geq |S|-1 -\eps|S|=(1-\eps)|S|-1.
    \end{align*}

    Next, consider a vertex $u\in S$ with at least $d^+(S,u)\geq \eps \alpha$ out-neighbors in $S$. In $H_i$ this vertex has in expectation $d^+(S,u)p_i$ out-neighbors, to show our result with high probability, we need to be more precise. 
    
    We say a vertex $u$ with $d^{+}(S,u) \geq \eps \alpha$ is \emph{good} for $S$, if the sampled degree for $v$ in $S$ satisfies 
    \begin{equation*}
        d_{H_i}^{+}(S,u) \geq (1-\eps) d^{+}(S,u)p_i.
    \end{equation*}
    If $u$ is not good, it is \emph{bad}. Let $B(S)$ be the set of bad vertices for $S$. Using a Chernoff bound, we see that the probability that a vertex $u$ is bad for $S$
    \begin{align*}
        \mathbb P[d_{H_i}^+(S,u)< (1-\eps)d^+(S,u)p_i] &\leq e^{-\eps^2d^+(S,u)p_i/2} \\
        &\leq e^{-\eps^2\eps \alpha8(c+3)\tfrac{\log m}{2^i \eps^4} /2}\\
        &\leq e^{-\eps^2\eps 2^{i-1}8(c+3)\tfrac{\log m}{2^i \eps^4}/2}\\
        &\leq m^{-(c+3)2/\eps}.
    \end{align*}
    
    \textbf{Part~\ref{part:1}\ref{part:a}.}
    % Now, we consider $|S|\leq 2/\eps$.
    We can use a union bound over all $u$ and $S$ of size at most $2/\eps$ to show that every $u$ is good for all such $S$.
    So w.h.p.\ we have that for $|S|\leq 2/\eps$ and for $u$ with $d^+(S,u)\geq \eps \alpha$ we have $d_{H_i}^+(S,u)\geq (1-\eps)d^+(S,u)p_i$. Further we have at most $\eps|S|+1$ vertices with out-degree at most $\eps\alpha$, in total covering at most $(\eps|S|+1)\eps\alpha$ edges in $S$. Combining these two facts gives us:

    \begin{align*}
        \frac{1}{p_i}\frac{|E_{H_i}(S)|}{|S|-1} &\geq \frac{(1-\eps)\alpha(|S|-1)-(\eps|S|+1)\eps \alpha}{|S|-1}\\
        &\geq (1-\eps)\alpha \left(1-(1-\eps)^{-1}\eps\frac{\eps|S|+1}{|S|-1}  \right)\\
        &\geq (1-\eps)(1-4\eps) \alpha. 
    \end{align*}
    Where we use that $|S|\geq 2$ and $\eps\leq \tfrac{1}{2}$.
    Now setting $\eps\leftarrow \eps/5$ gives the result.

    \textbf{Part~\ref{part:1}\ref{part:b}.}
    Next, we consider $|S|>2/\eps$. 
    If there are at most $|B(S)|\leq \eps (|S|-1)$ bad vertices $u\in S$, then $(1-2\eps)|S|$ good vertices have at least $\eps\alpha$ out-neighbors. Hence this guarantees
    \begin{align*}
        \frac{1}{p_i}\frac{|E_{H_i}(S)|}{|S|-1} &\geq \frac{(1-\eps)\alpha(|S|-1)-(\eps|S|+1)\eps \alpha-\eps(|S|-1)\alpha}{|S|-1}\\
        &\geq (1-\eps)\alpha \left(1-(1-\eps)^{-1}\eps\frac{\eps|S|+1}{|S|-1}  \right)-\eps\alpha\\
        &\geq (1-\eps)(1-4\eps) \alpha-\eps \alpha. 
    \end{align*}
    Setting $\eps\leftarrow\eps/6$ gives the result. 

    Now we compute the probability that there are more than $\eps(|S|-1)$ bad vertices.    
    Earlier we showed that every time a vertex is resampled it is bad for $S$ with probability at most $m^{-(c+3)2/\eps}$.
    Now, using this terminology, we want to show that with small enough probability the event $|B(S)| > \eps(|S|-1)$ occurs.
    We have assumed the adversary can do at most $m^{2}$ resampling attacks
    
    At least $\eps(|S|-1)$ bad resamplings for $S$ has to happen for $|B(S)| > \eps(|S|-1)$. 
    It follows that the event $|B(S)| > \eps(|S|-1)$ implies the existence of a subset $T \subset [m^2]$ with $|T| = \floor{\eps(|S|-1)+1}$ such that every resampling in $T$ is bad for $S$. 
    Hence, we find that:
    \begin{align*}
        \mathbb{P}(|B(S)| > \eps(|S|-1)) &\leq \sum \limits_{T \in [m^2]_{\floor{\eps(|S|-1)+1}}}\mathbb{P}(T \subset B(S)) \\
        &\leq \sum \limits_{T \in [m^2]_{\floor{\eps(|S|-1)+1}}} (m^{-(c+3)2/\eps})^{|T|} \\
        &\leq \binom{m^2}{\floor{\eps(|S|-1)+1}} (m^{-(c+3)2/\eps})^{|T|} \\
        &\leq m^{2\floor{\eps(|S|-1)+1}} (m^{-\tfrac{2(c+3)}{\eps}})^{\floor{\eps(|S|-1)+1}} \\
        &\leq m^{-\tfrac{2c}{\eps}\floor{\eps(|S|-1)+1}} \\
        &\leq m^{-\tfrac{2c}{\eps}\frac{\eps|S|}{2}} \\
        &\leq m^{-c|S|},
    \end{align*}
    since $|S| \geq 2$. 
    Here for a set $A$, we denoted by $A_{k} = \{A' \subset A: |A'| = k\}$ the set of all subsets of $A$ of size $k$. 

    We have at most $n^{t}$ choices for $S$ of size $t$, so union bounding over all $n$ choices of $t$ and all $n^{t}$ choices of $S$ for each $t$ gives that some $S$ is bad with probability at most: 
    \begin{equation*}
        \sum \limits_{t} n^{t} m^{-ct} \leq n \cdot{} m^{-(c-1)} \leq m^{-(c-2)}
    \end{equation*}
    By union bounding over all $m^2$ different choices of $G$ throughout the update sequence, we find that the bad event does not happen for any choice of $S$ or $G$ with probability at least $m^{-(c-4)}$, and so reassigning $c \leftarrow c+4$ gives that it does not happen with high probability.

    \textbf{Part~\ref{part:2}.}
    In this case, we say that a vertex is \emph{good} for $S$:
    \begin{itemize}
        \item If $d^+(S,u)\geq \eps\alpha$, if the sampled degree for $u$ in $S$ satisfies that 
    \begin{equation*}
        d_{H_i}^{+}(S,u) \leq (1+\eps) d^{+}(S,u)p_i.
    \end{equation*}
        \item If $d^+(S,u)< \eps\alpha$, if the sampled degree for $u$ in $S$ satisfies that 
    \begin{equation*}
        d_{H_i}^{+}(S,u) \leq 3\eps \alpha p_i.
    \end{equation*}
    \end{itemize}
    Again, if $u$ is not good, it is \emph{bad}. 
    
    If $d^+(S,u)\geq \eps\alpha$, we can show that $u$ is good for $S$ with probability at least $1-m^{-(c+3)2/\eps}$, again by applying a Chernoff bound, analogous to Part~\ref{part:1}. 
    
    Let $u\in S$ be a vertex with $d^+(S,u)\leq \eps\alpha$, we show that the probability that these vertices are bad is also small.
    Indeed, we have
    \begin{align*}
        \mathbb P\left[ d^+_{H_i}(S,u)> 3\eps \alpha p_i\right] &=  \mathbb P\left[ d^+_{H_i}(S,u)> \frac{3\eps \alpha}{d^+(S,u)}p_i d^+(S,u) \right]\\
        &\le  \mathbb P\left[ d^+_{H_i}(S,u)> \left(1+\frac{2\eps \alpha}{d^{+}(S,u)}\right)p_i d^+(S,u) \right]\\
        &\le \exp\left(- \frac{\left(   \frac{2\eps \alpha}{d^{+}(S,u)}\right)^2p_i d^+(S,u)}{2+\frac{2\eps \alpha}{d^{+}(S,u)}}\right)\\
        &\le \exp\left(- 2\eps\alpha p_i\frac{   \frac{2\eps \alpha}{d^{+}(S,u)}}{\frac{4\eps \alpha}{d^{+}(S,u)}}\right)\\
        &= \exp(-\eps\alpha p_i)\\
        &\le m^{-(c+3)2/\eps^3},
    \end{align*}
    since $\alpha p_i\geq 2^{i-1} 8(c+3)\tfrac{\log m}{2^{i}\eps^4}\geq2(c+3)\tfrac{\log m}{\eps^4}$.
    We conclude that in any case the probability that $u$ is bad for $S$ is at most $m^{-(c+3)2/\eps}$
    
    \textbf{Part~\ref{part:2}\ref{part:a}.}
     Again, we can union bound over all $u$ and all $S$ with $|S|\leq 2/\eps$, and all $O(m^2)$ updates to get that all vertices are good w.h.p.

    Now we see that if we let $k\geq 1$ denote the number of vertices that have $d^+(S,u)< \eps\alpha$ then
    \begin{align*}
        |E_{H_i}(S)| &= \sum_{\substack{u\in S \\ d^+(S,u)\geq \eps\alpha}} d^+_{H_i}(S,u) + \sum_{\substack{u\in S \\ d^+(S,u)< \eps\alpha}} d^+_{H_i}(S,u) \\
        &\leq (|S|-k)(1+\eps)\alpha p_i + k \cdot 3\eps\alpha p_i. 
    \end{align*}
    So we see that 
    \begin{align*}
        \frac{1}{p_i}\frac{|E_{H_i}(S)|}{|S|-1} &\leq \frac{1}{p_i}\frac{(|S|-k)(1+\eps)\alpha p_i + k \cdot 3\eps\alpha p_i}{|S|-1} \\
        &\leq \frac{(|S|-1)(1+\eps)  + k \cdot 3\eps }{|S|-1} \alpha\\
        &\le \left(1+\eps+3\frac{|S|}{|S|-1}\eps\right)\alpha\\
        &\le (1+7\eps)\alpha.
    \end{align*}
    Setting $\eps\leftarrow \eps/7$ gives the result.

    \textbf{Part~\ref{part:2}\ref{part:b}.}
    Finally, we consider $|S|>2/\eps$. 
    We first note
    \begin{equation}
        d_{H_i}^+(S,u)\leq d_{H_i}^+(u) \leq (1+\eps)p_i\alpha,\label{eq:degree}
    \end{equation}
    where the last inequality holds by a Chernoff bound independent from $S$ (so we only need to union bound over all $u\in V$). 
    
    If at most $\eps (|S|-1)$ vertices $u\in S$ are bad for $S$, then in the worst case they achieve equality in \Cref{eq:degree}. We now sum the bad and the good vertices, using the result on good vertices analogous to Part~\ref{part:2}\ref{part:a} to see
    \begin{equation*}
        \frac{1}{p_i} \frac{E_{H_i}(S)}{|S|-1} \leq \frac{1}{p_i} \frac{\eps(|S|-1)(1+\eps)p_i\alpha + (1+7\eps)p_i\alpha}{|S|-1} \leq (1+9\eps)\alpha.
    \end{equation*}
    Setting $\eps\leftarrow \eps/9$ gives the result.

    The fact that w.h.p.\ there are at most $\eps (|S|-1)$ bad vertices is analogous to Part~\ref{part:1}\ref{part:b}.

    \textbf{Bounding $\alpha_{\max}$.}
    Finally, we remark that the fractional arboricity in $H_i$ is at most $\alpha_{\max}=\Theta(\log m/\eps^4)$.
    Since $\alpha < 2^{i+2}$, we get by Part~\ref{part:2} that the fractional arboricity in $H_i$ is w.h.p.\ at most $p_i\cdot (1+\eps)\alpha \leq \frac{8(c+3)\log m}{2^i \eps^4}\cdot 2^{i+3}=O(\tfrac{\log m}{\eps^4})$. Hence by \Cref{thm:arb_det}, we maintain a $(1+\eps)$-approximation of the fractional arboricity in $H_i$, which scales to a $(1+\eps)$-approximation of the fractional arboricity in $G$ by Parts~\ref{part:1} and~\ref{part:2}. 
\end{proof}

\paragraph{Putting it all together.}
For the update time, we will need the following lemma. 
\begin{lemma}[\cite{chekuri2024adaptive}]\label{lm:out_orient}
     Given an unweighted, undirected (multi-)graph, we can maintain $(1+\eps)\alpha$ out-orientations with $O(\eps^{-6}\log^3 m \log \alpha)$ worst-case update time. The orientation maintained by the algorithm has $O(\eps^{-4}\log^2 m \log \alpha)$ recourse. 
\end{lemma}
Note that this algorithm works for multi-graphs, if one keeps all parallel edges in balanced binary search trees. When rounding the fractional orientation one has to take care of all $2$-cycles before further processing things, but these can be identified and handled using another balanced binary search tree sorted by fractional orientation from one vertex to another. 
This ensures that the refinement contains no parallel edges, and therefore only the orientation of one parallel copy is stored implicitly and can be looked up when necessary. 
Note that the above things can be implemented in $O(\log m)$ time, but this does not increase the running time as it is done in parallel to other more expensive steps. 

\arboricity*
\begin{proof}
    We use the algorithm as described above. 
    Correctness follows from \Cref{lm:arb_aa_correct}. To obtain the bounds on the update time we note the following.
    We maintain $\log m$ graphs $H_i$, on which we apply \Cref{thm:arb_det} with $\alpha_{\max}=\Theta(\tfrac{\log m}{\eps^4})$, which needs $O(\log^{7} m/\eps^{8})$ time per update. By the recourse of the out-orientation, \Cref{lm:out_orient}, we need to resample the out-edges of $O(\log^3 m/\eps^4)$ vertices for each $H_i$ per update to $G$. Every resample leads to $O(\log m/\eps^3)$ edge updates in $H_i$. So in total each $H_i$ has $O(\log^4 m/\eps^7)$ updates per update to $G$.
    Multiplying this with the aforementioned update time for $H_i$, we obtain 
    \begin{equation*}
        \log m \cdot O(\log^4 m/\eps^7) \cdot O(\log^{7} m/\eps^{8})= O(\log^{12}m/\eps^{15})
    \end{equation*}
    amortized update time. For the bound with worst-case update time, we have by \Cref{thm:arb_det} that each update takes $O(m^{o(1)}/\eps^{12})$ time, so the worst-case update time becomes $O(m^{o(1)}/\eps^{19})$.
\end{proof}

%% file: lowerbound.tex
\section{A Lower Bound for Greedy Tree-Packing}\label{sc:LB}

In this section, we will show the following theorem:
\TPlowerbound*
To do so, we construct a family of graphs, and give an execution of greedily packing trees on these graphs such that $|\ell^\T(e)-\ell^*(e)| > \eps/\lambda$ if $|\T|=o(\lambda/\eps^{3/2})$. 
We first do this for $\lambda=2$, then we note that we can obtain the result for any even $\lambda$ by essentially copying this construction $\lambda/2$ times. To get an intuition for the proof, we recommend the reader to look ahead to the figures. In the right part of \Cref{fig:Construction}, we depict the constructed graph with $\lambda=2$. This graph is a very uniform graph; every edge $e$ has $\ell^*(e)=1/2$ (see \Cref{lma:disjSPT}). The packing of trees is depicted in \Cref{fig:unmodified_pair} and \Cref{fig:modified_pair}, where we can see that certain edges are over-packed and others are under-packed. The over-packed edges will get a value $\ell^\T(e)$ well above $1/2=\ell^*(e)$, giving the result. 

The construction works for any tuple $(\lambda, k, n) \in (2\mathbb{Z}) \times \mathbb{Z}_{\geq 1} \times \mathbb{Z}_{\geq 10}$ with $k = \mathcal{O}(n^{1/3})$. 
Given $n$ and $k$ satisfying these requirements, we first show the construction for $\lambda = 2$. 
We extend the construction to any $\lambda \in 2\mathbb{Z}$ afterwards.
Before we present the family of graphs, we first introduce a simple operation which preserves a partition value of $2$. \\

\textbf{Operation:} Replace any vertex $v$ by two vertices $v'$ and $v''$ connected by two parallel edges. The edges around $v$ can be distributed to $v'$ and $v''$ in any arbitrary fashion. 
The following lemma is then straight-forward to check.
\begin{lemma} \label{lma:operation}
    Given a graph $G$ with $\Phi(G) = 2$ such that the trivial partition $\mathcal{P} = \{u\}_{u \in V(G)}$ is a minimum partition, then for all $v \in V(G)$, any graph $G_v$ obtained by performing a valid version of the operation on $v$ also has $\Phi(G_v) = 2$ and that the trivial partition $\mathcal{P} = \{w\}_{w \in V(G_v)}$ is a minimum partition.
\end{lemma}
\begin{proof}
Observe that any graph $H$ that satisfies the conditions $\Phi(H) = 2$ and that the trivial partition $\mathcal{P} = \{u\}_{u \in V(H)}$ is a minimum partition is exactly the union of two disjoint spanning trees. 

Indeed, suppose first that $H$ is a union of $2$ disjoint spanning trees. It then follows by~\cite{Nash61} that $\Phi(H)\geq 2$. Since the trivial partition induces partition value exactly $2$, this direction follows. 

To see the other direction, observe that $\Phi(H) = 2$ implies that one can pack two disjoint spanning trees of $H$ by~\cite{Nash61}. Since the trivial partition achieves this minimum, $H$ must in fact be the disjoint union of two spanning trees. 

Finally, observe that performing the operation and placing each new edge in a different spanning tree yields a new graph which is also exactly the union of two disjoint spanning trees. 
\end{proof}

Next, consider the following family of graphs indexed by $n$ and $k$: for given $n$ and $k$, we let $G_{n,k}$ be the following graph. 
Begin with the complete simple graph on $3$ vertices $K_3$. 
Pick an arbitrary edge and insert a parallel copy of it. Denote by $v^{1}_1$ the vertex not incident to any parallel edges. 
Denote by $v^{1}_2$ and $v^{1}_{3}$ the other two vertices (the choice can be made arbitrarily, but is fixed once made). Let $e^{1}_{1}$ be the edge $v^{1}_{1}v^{1}_{2}$ and $e^{1}_{2}$ the edge $v^{1}_{1}v^{1}_{3}$. See the left part of \Cref{fig:Construction} for an illustration. 

Next, we perform the operation on $v^{1}_{3}$ to get two new vertices $v'$ and $v''$. 
We distribute the edges incident to $v^{1}_{3}$ as follows: all edges incident to $v^{1}_{3}$ becomes incident to $v'$ except for one of the parallel edges which becomes incident to $v''$. 
We then (re-)assign $v^{1}_{3} = v'$, $v^{2}_{1} = v^{1}_{2}$, $v^{2}_{2} = v^{1}_{3}$, $v^{2}_{3} = v''$, $e^{2}_{1} = v^{2}_{1}v^{2}_{2}$, and $e^{2}_{2} = v^{2}_{1}v^{2}_{3}$.

Having constructed $v^{i}_{1}$, $v^{i}_{2}$, and $v^{i}_{3}$, we get $v^{i+1}_{1}$, , $v^{i+1}_{2}$, and $v^{i+1}_{3}$ similarly to above: 
we perform the operation on $v^{i}_{3}$ to get two new vertices $v'$ and $v''$. 
We distribute the edges incident to $v^{i}_{3}$ as follows: all edges incident to $v^{i}_{3}$ becomes incident to $v'$ except for one of the parallel edges which becomes incident to $v''$. 
We then (re-)assign $v^{i}_{3} = v'$, $v^{i+1}_{1} = v^{i}_{2}$, $v^{i+1}_{2} = v^{i}_{3}$, $v^{i+1}_{3} = v''$, $e^{i+1}_{1} = v^{i+1}_{1}v^{i+1}_{2}$, and $e^{i+1}_{2} = v^{i+1}_{1}v^{i+1}_{3}$. See the middle part of \Cref{fig:Construction} for an illustration.

\begin{figure}%
    \centering  
    \subfloat{{\includegraphics[width=4.3cm]{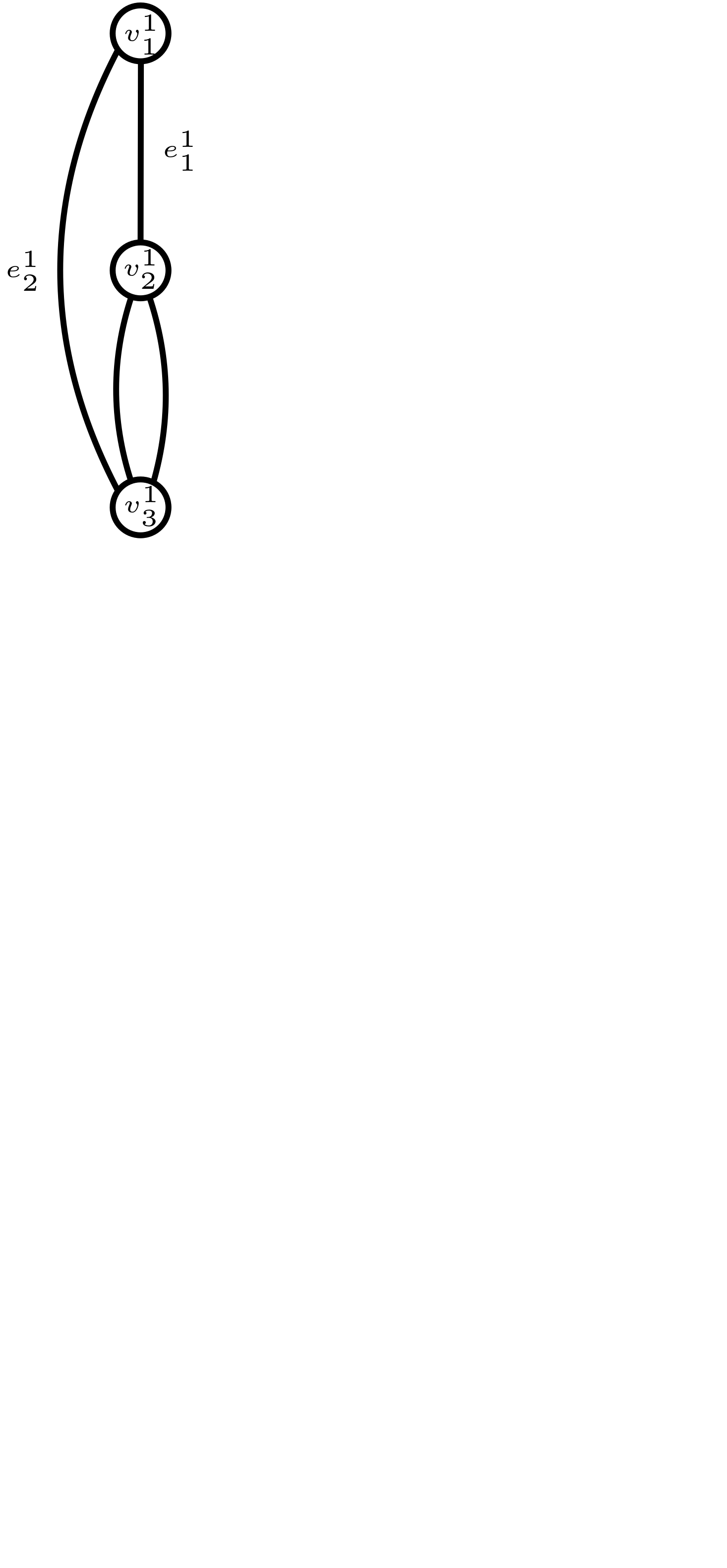} }}%
    \hspace{1mm}%
    \subfloat{{\includegraphics[width=5cm]{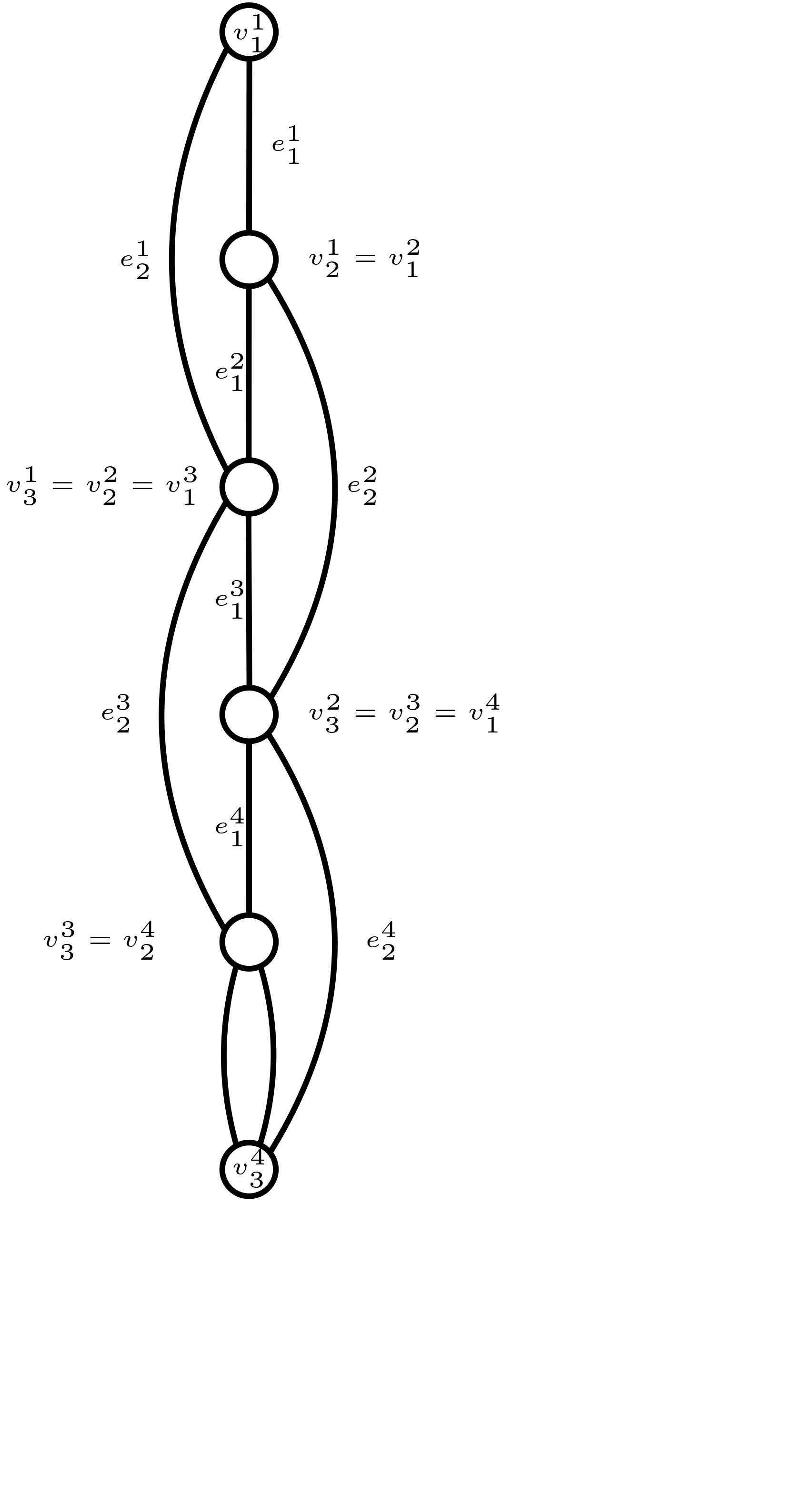} }}%
    \hspace{1mm}%
    \subfloat{{\includegraphics[width=5.2cm]{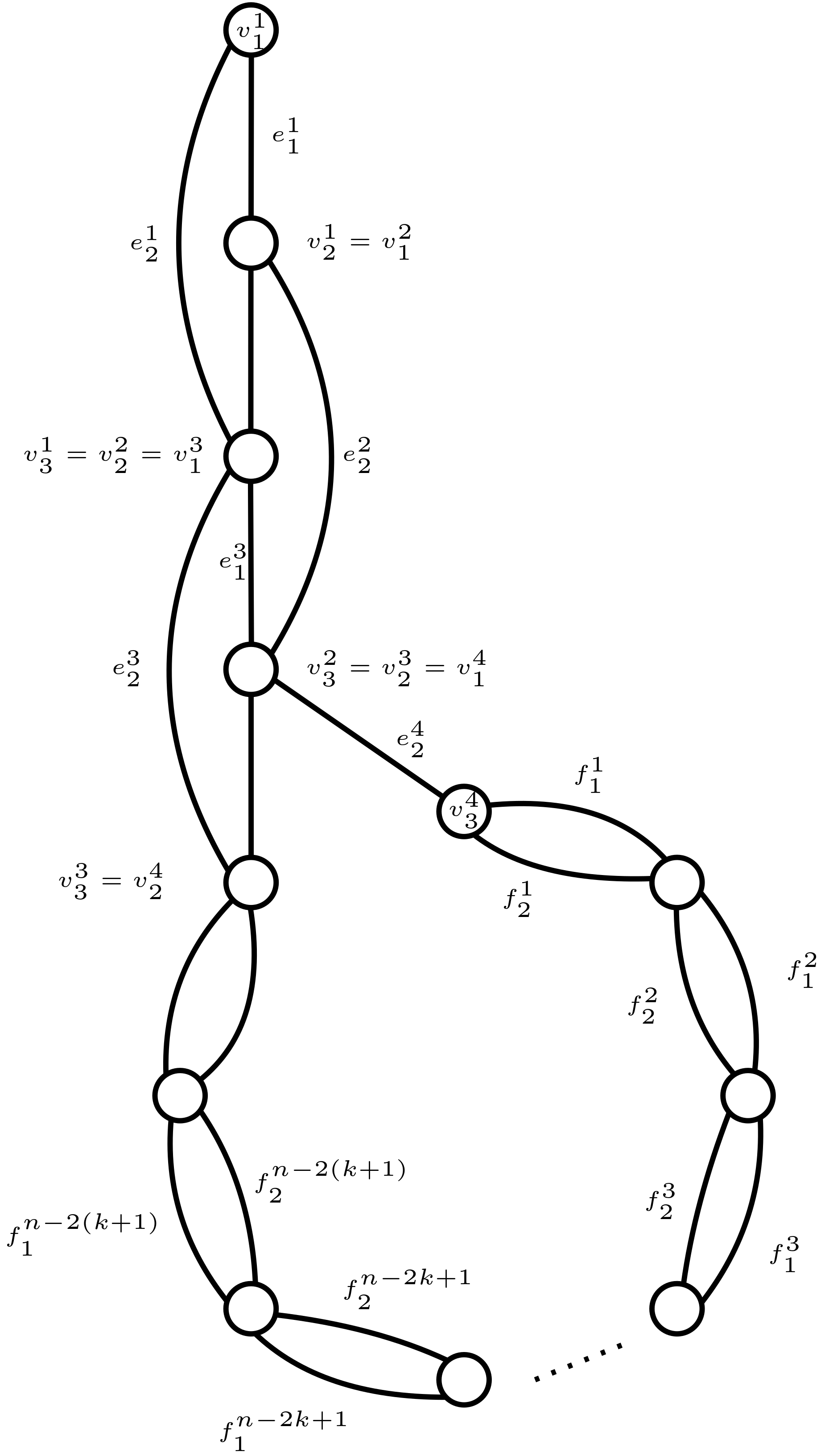} }}%
    \caption{Various stages of the construction of the graph $G_{n,k}$. The leftmost illustration shows $K_3$ with an additional parallel edge. The middle illustration shows the graph obtained by repeatedly applying the operation to the graph to the left. Finally, the rightmost illustration shows the final graph $G_{n,k}$.}
    \label{fig:Construction}%
\end{figure}

We perform the above step $2k-1$ times. Each step increases the number of vertices by $1$, and so the resulting graph has $3+(2k-1) = 2(k+1)$ vertices. 

We then perform the following step $n-2(k+1)$ times: 
perform the operation on $v^{2k}_{3}$ to get two new vertices $v'$ and $v''$. Let all edges incident to $v^{2k}_{3}$ become incident to $v'$ except for the parallel edges incident to $v^{2k}_{3}$ which become incident to $v''$. 
Then re-assign $v^{2k}_{3} = v'$ and denote the two new parallel edges by $f^{n-2(k+1)}_{1}$ and $f^{n-2(k+1)}_{2}$ (the choice can be made arbitrarily, but is fixed once made). 
Having constructed $f^{i}_1$ and $f^{i}_2$, we can construct $f^{i-1}_1$ and $f^{i-1}_{2}$ mutatis mutandis to above. See the right part of \Cref{fig:Construction} for an illustration.

Each time the second step is performed, the number of vertices is also increased by one, so in total the graph has $n$ vertices. 
By repeatedly applying \Cref{lma:operation}, we find the resulting graph is the disjoint union of two spanning trees. Hence, we have that:
\begin{lemma} \label{lma:disjSPT}
    For valid choices of $n,k$, we have that $G_{n,k}$ satisfies that $\Phi(G_{n,k}) = 2$ and that the trivial partition $\mathcal{P} = \{u\}_{u \in V(G_{n,k})}$ is a minimum partition.
\end{lemma}

Next, we will specify a packing of $G_{n,k}$ consistent with \Cref{thm:TP_lower_bound}. 
We consider the pairs of edges $Y_{i}$ consisting of $e^{i}_{1}$ and $e^{i}_{2}$. See \Cref{fig:level} for an illustration. 
After having packed the first $2j$ trees, we say that $Y_{i}$ is at level $2\alpha$ if all edges in $Y_{i}$ have been packed exactly $\alpha+j$ times.
We say that $Y_{i}$ is at level $2\alpha + 1$ if $e^{i}_{1}$ has been packed $\alpha + 1+j$ times and $e^{i}_{2}$ has been packed $\alpha+j$ times. 
If $Y_{i}$ is at level $\beta$ we write $\text{lev}(Y_i) = \beta$.

\begin{figure}%
    \centering
    \subfloat{{\includegraphics[width=5cm]{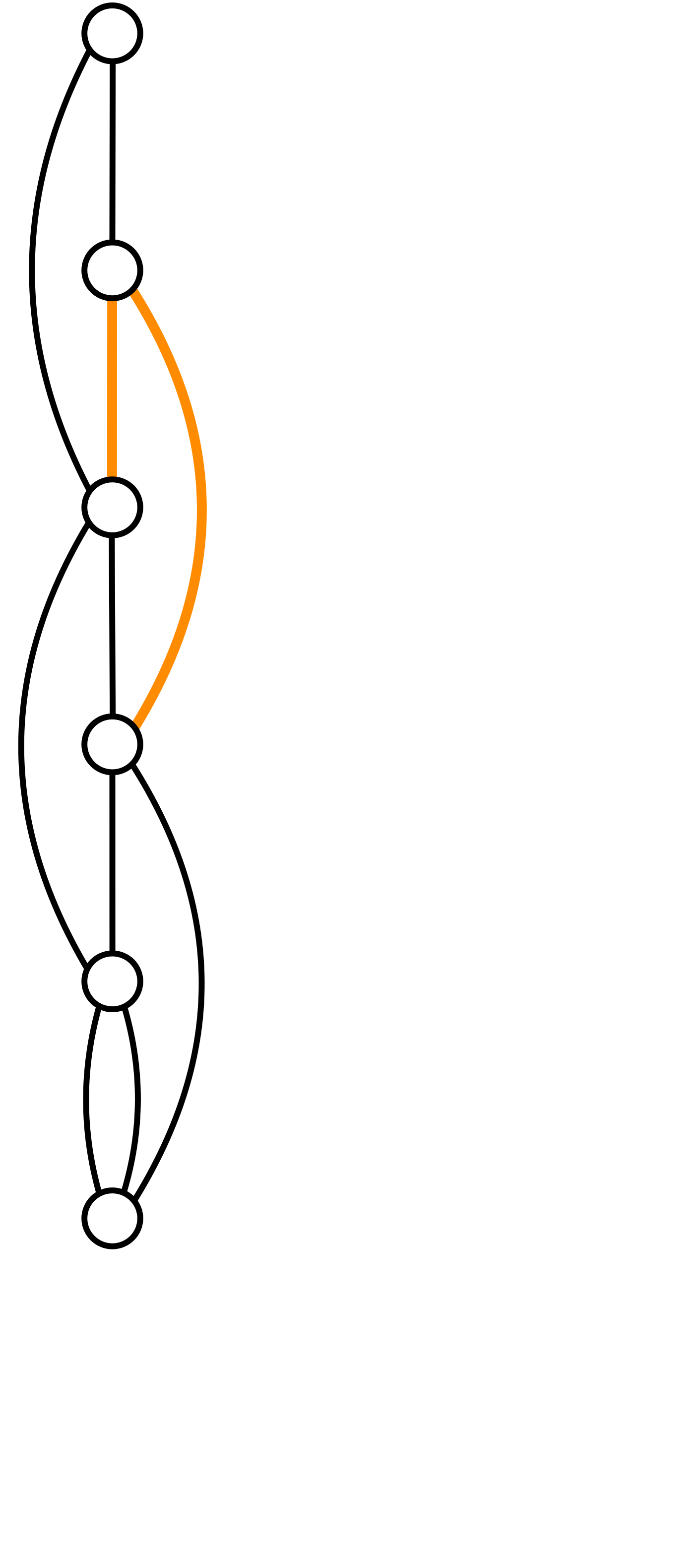} }}%
    \caption{The construction with the level $Y_2$ highlighted in orange. }
    \label{fig:level}%
\end{figure}

We use the following definition:
\begin{definition}
    We say that the tree-packing $\T$ on $G_{k,n}$ is in \emph{standard position} if the following holds:
    \begin{itemize}
        \item $|\T| = 2j$ for some $j \in \mathbb{Z}$.
        \item For all $i$: $Y_i$ is at level $\beta_{i}$ for some $\beta_{i}$. 
        \item For all $i<j$, we have that: $\beta_{i} \geq \beta_{j}$ and that $|\beta_{i}-\beta_{i-1}| \leq 1$.
        \item $\beta_{2k-1} \geq 0$.
        \item For all $i$: $f^{i}_{2}$ has been packed $j$ times. 
        \item There is some $\iota$ such that for all $i \geq \iota$: $f^{i}_{1}$ has been packed $j$ times, and for all $i < \iota$ $f^{i}_{1}$ has been packed $j-1$ times. 
    \end{itemize}
We let the vector $\beta = (\beta_1, \beta_2, \dots, \beta_{2k}) \in \mathbb{Z}^{2k}$ be the \emph{level profile} of $\T$.
\end{definition}

Next we will show that if $\T$ is in standard position with $\text{lev}(Y_i) = \text{lev}(Y_{i+1})$ for some $i$, then we can increase the level of some $Y_j$ without decreasing the level of any pair by adding only $O(k)$ trees to $\T$ while still ensuring that $\T$ ends up in standard position. 
The final tree-packing is then achieved by applying the above procedure in a systematic way $O(k^2)$ times. 
Note that the empty packing is in standard position.
Before showing this, we first show the following lemma.
\begin{lemma} \label{lma:OneAugment}
    Let $\T$ be a greedy tree-packing in standard position on $G_{n,k}$ with $|\T| = 2j$ and $\beta_{2k} \geq 0$. 
    Let $i$ and $s \geq 1$ be such that 1) either $\beta_i = \beta_{i+1} = \dots = \beta_{i+s} > \beta_{i+s+1}$ or $\beta_i = \beta_{i+1} = \dots = \beta_{i+s}$ and $i+s = 2j$, and 2) either $\beta_{i-1} > \beta_{i}$ or $i = 1$.
    Then there is a greedy tree-packing $\T'$ in standard position on $G_{n,k}$ with $2j + 2(\floor{\tfrac{s}{2}}+1)$ trees such that the level profile $\beta'$ of $\T'$ satisfies $\beta'_l = \beta_l + [i = l] - [i+s=l]$\footnote{Here $[P]$ denotes the Iverson bracket, which evaluates to $1$ if $P$ is true and $0$ otherwise.} for all $l$.
\end{lemma}
\begin{proof}
    We will show the lemma assuming that $\beta_i = \beta_{i+1} = \dots = \beta_{i+s} > \beta_{i+s+1}$ and that $\beta_{i-1} > \beta_{i}$. The other cases follow from analogous arguments. 

    We pack trees in pairs. Consider first the following trees $T$ and $T'$ that greedily extends $\T$. The edge $e^{l}_{1}$ is in $T$ if $\text{lev}(Y_l)$ is even, $e^{l}_{2}$ is in $T$ if $\text{lev}(Y_l)$ is odd, and all edges $f^{p}_{1}$ are in $T$.
    In order to verify that $T$ indeed extends $\T$ in a greedy manner, observe that since $\T$ is in standard position, it follows by downward induction on $l$ that $T$ extends $\T$ greedily. See \Cref{fig:unmodified_pair} for an illustration. 

    Similarly, we let $e^{j}_{2}$ in $T'$ if $\text{lev}(Y_j)$ is even, $e^{j}_{1}$ in $T'$ if $\text{lev}(Y_j)$ is odd, and all edges $f^{j}_{2}$ are in $T$.
    Exactly as above, it follows by induction that $T'$ greedily extends $\T \cup T$, and that $\T \cup \{T, T'\}$ is a greedy tree-packing in standard position with the same level profile as $\T$. 
    
    \begin{figure}%
        \centering
        \subfloat{{\includegraphics[width=4cm]{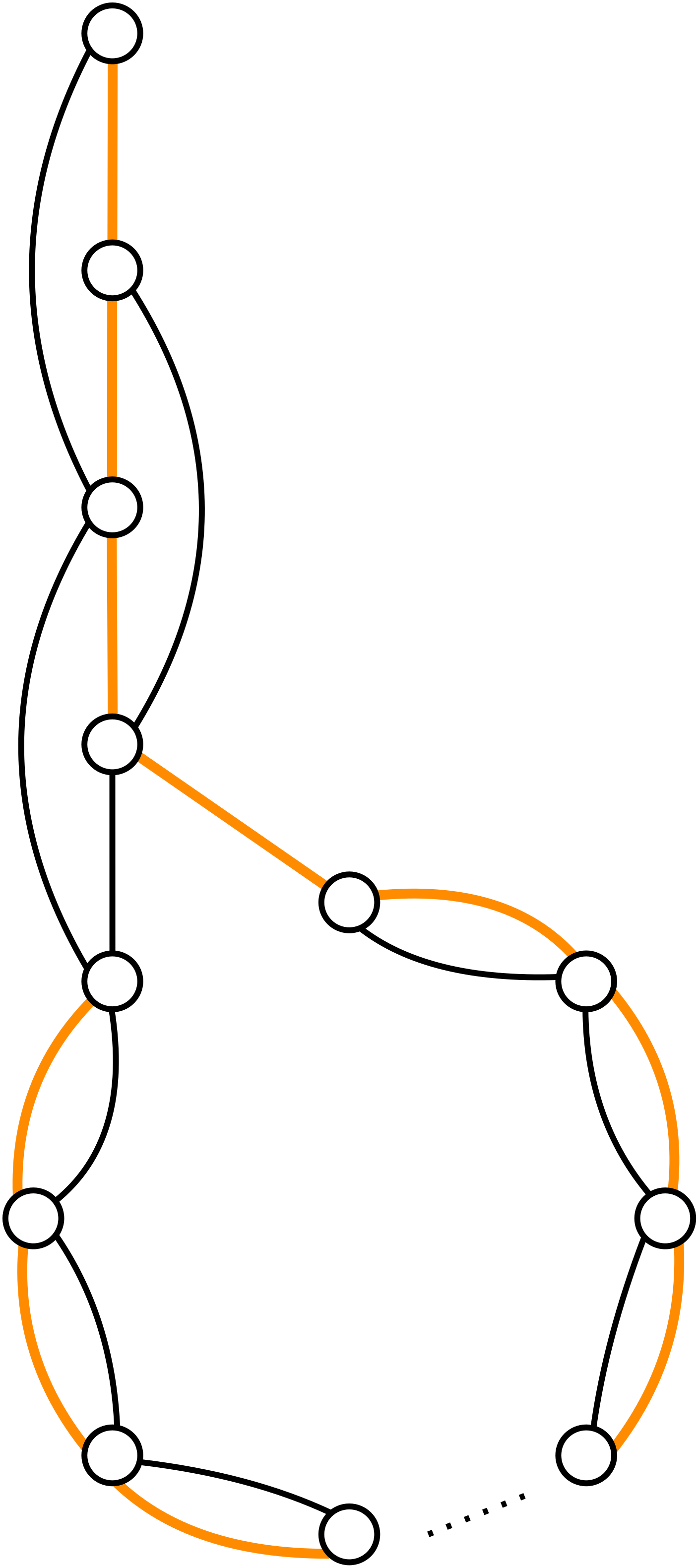} }}%
        \hspace{30mm}%
        \subfloat{{\includegraphics[width=4cm]{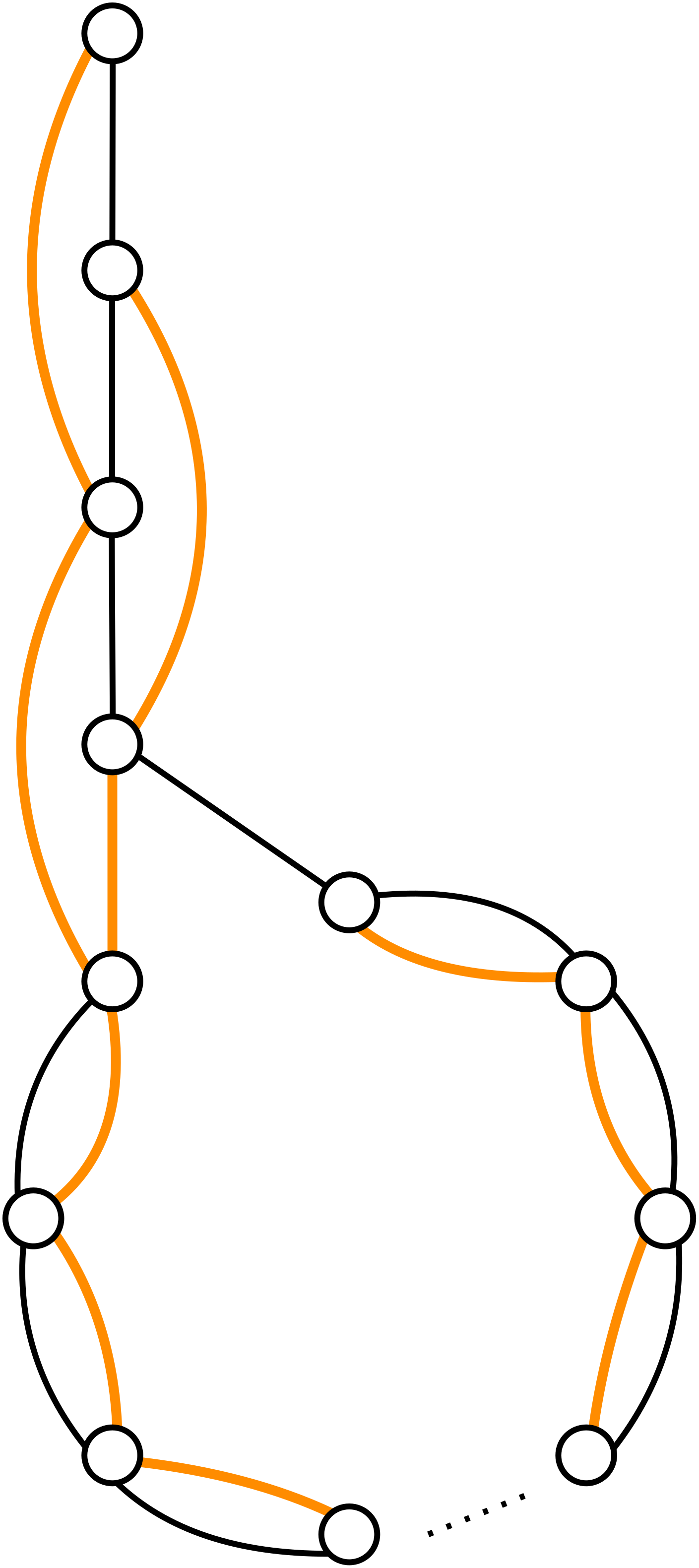} }}%
        \caption{The figure illustrates $T$ (on the left) and $T'$ (on the right) in the case where $i = 1$, $s = 3$, $\text{lev}(Y_1) = \text{lev}(Y_2) = \text{lev}(Y_3) = 2$, and $\text{lev}(Y_4) = 1$.}
        \label{fig:unmodified_pair}%
    \end{figure}

    We can extend $\T$ as above as many times as we would like to obtain a greedy tree-packing
    \[
    \T \cup \paren{ \bigcup_{i = 1}^{2(\floor{\tfrac{s}{2}}+1)} \{T_i, T'_{i}\}}
    \]
    in standard position with the same level profile as $\T$. 
    
    Consider first the case where $\text{lev}(Y_i) = \text{lev}(Y_{i+1})$ is even.
    Then, we can perform the swap of $e^{i}_{2}$ for $e^{i+1}_{1}$, while keeping $T_1$ a greedy extension of $\T$. 
    Similarly, we can now perform the swaps $e^{i}_{1}$ for $e^{i}_{2}$ and $e^{i+1}_{1}$ for $e^{i+2}_{2}$ in $T'_{1}$. 
    After these swaps, $T'_1$ is again a greedy extension of $\T \cup T_1$. See \Cref{fig:modified_pair} for an illustration.

    \begin{figure}%
        \centering
        \subfloat{{\includegraphics[width=4cm]{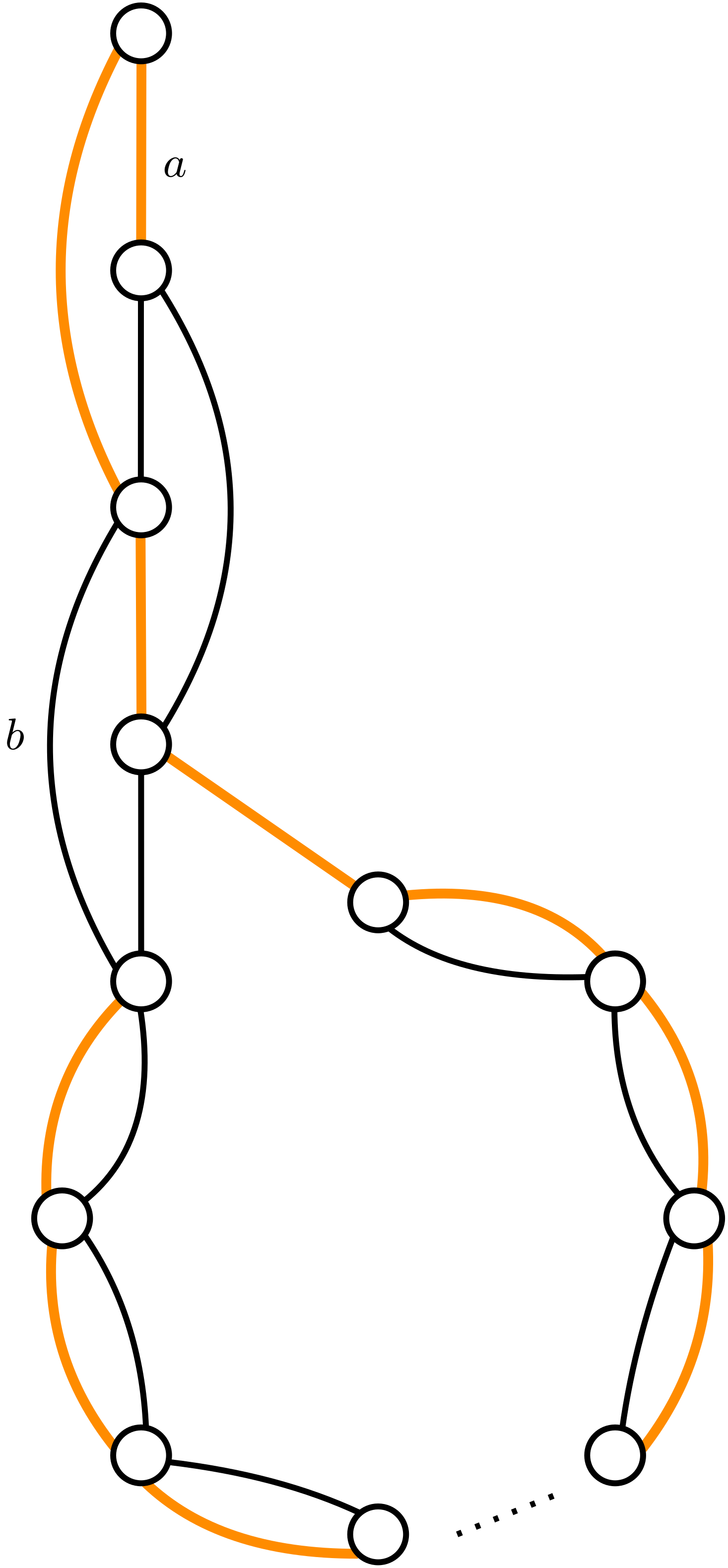} }}%
        \hspace{30mm}%
        \subfloat{{\includegraphics[width=4cm]{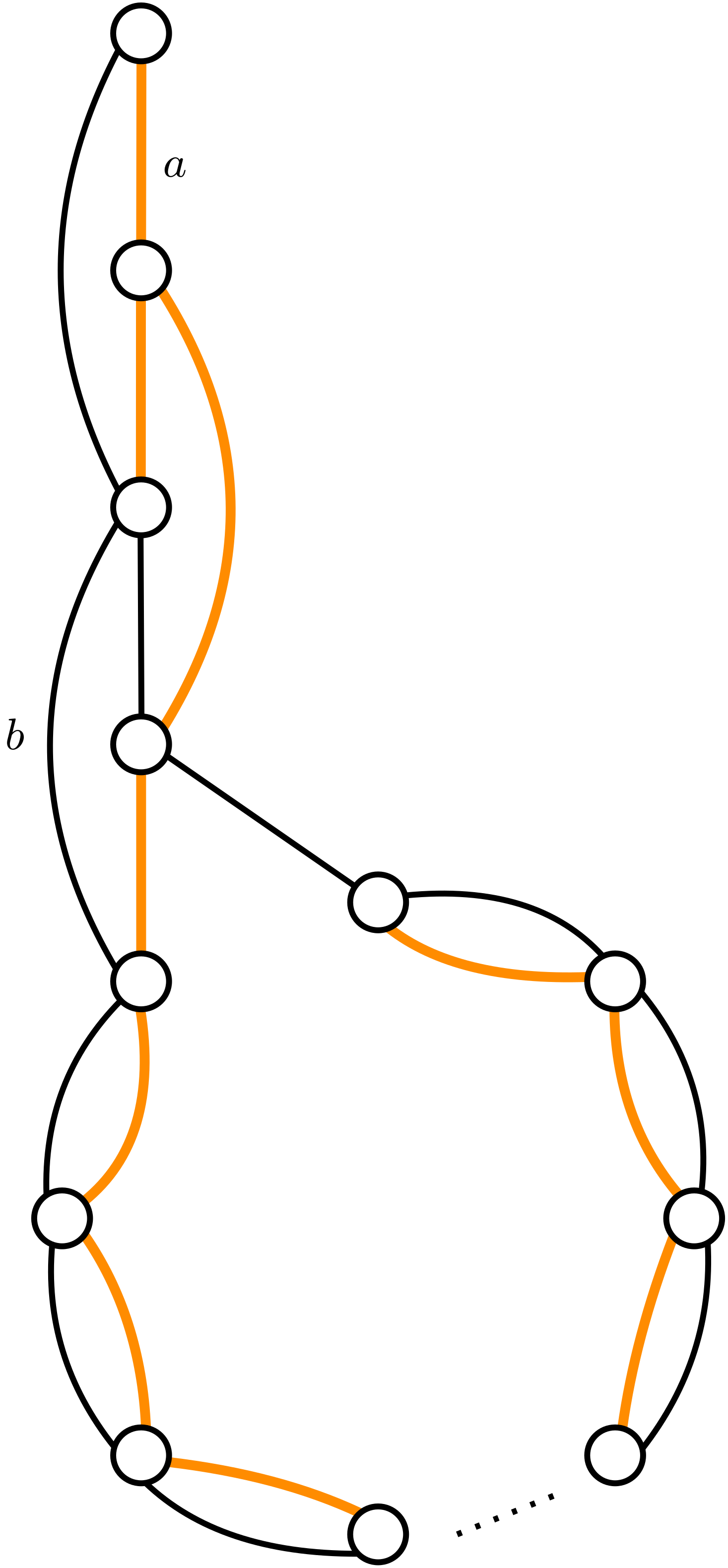} }}%
        \caption{The figure illustrates $T_1$ (on the left) and $T'_1$ (on the right) in the case where $i = 1$, $s = 3$, $\text{lev}(Y_1) = \text{lev}(Y_2) = \text{lev}(Y_3) = 2$, and $\text{lev}(Y_4) = 1$.
        Observe that all edges are packed once except for $a$, which is packed twice, and $b$ which is not packed. }
        \label{fig:modified_pair}%
    \end{figure}
    
    Next, we can change $T_2$ by swapping $e^{i}_{1}$ for $e^{i}_{2}$ and $e^{i+3}_{1}$ for $e^{i+2}_{2}$. 
    Then we change $T'_2$ by swapping $e^{i}_{2}$ for $e^{i}_{1}$, $e^{i+4}_{2}$ for $e^{i+3}_{1}$.
 
    We can continue this process mutatis mutandis, until at some point we pack neither $e^{i+s}_{1}$ nor $e^{i+s}_{2}$. 
    For even $s$, this happens in $T'_{\tfrac{s}{2}}$, and for odd $s$, this happens in $T_{\ceil{\tfrac{s}{2}}}$. 
    For even $s$, we make no more changes. For odd $s$, we need still need to change $T'_{\ceil{\tfrac{s}{2}}}$: we swap $e^{i}_{2}$ for $e^{i}_{1}$ and $e^{i+s}_{2}$ for $e^{i+s}_{1}$. 

    It follows by induction on $s$ that
    \[
    \T' = \T \cup \paren{ \bigcup_{i = 1}^{2(\floor{\tfrac{s}{2}}+1)} \{T_i, T'_{i}\}}
    \]
    is a greedy tree-packing in standard position with the claimed level profile. If needed, we can increase the size of the tree-packing to the required size by adding another pair of $T$ and $T'$ (based on the levels of $\T'$ and not of $\T$). 

    In the case where $\text{lev}(Y_i) = \text{lev}(Y_{i+1})$ is odd, we only alter $T'_1$. 
    Here, we exchange $e^{i+1}_{1}$ for $e^{i}_{2}$. 

    Next we perform the following swaps on $T_2$; we swap $e^{i}_2$ for $e^{i}_1$, and we swap $e^{i+2}_{2}$ for $e^{i+1}_{1}$. 
    We change $T'_2$ by swapping $e^{i}_1$ for $e^{i}_2$ and we swap $e^{i+3}_{1}$ for $e^{i+2}_{2}$. 
    
    We can continue this process mutatis mutandis, until at some point we pack neither $e^{i+s}_{1}$ nor $e^{i+s}_{2}$. 
    For even $s$, this happens in $T_{\tfrac{s}{2}+1}$, and for odd $s$, this happens in $T'_{\ceil{\tfrac{s}{2}}}$. 
    For odd $s$, we stop here. For even $s$, we need still need to pack $T'_{\tfrac{s}{2}+1}$ before stopping. Before doing so, we swap $e^{i}_{1}$ for $e^{i}_{2}$ and $e^{i+s}_{1}$ for $e^{i+s}_{2}$. 

    It follows by induction on $s$ that
    \[
    \T \cup \paren{ \bigcup_{i = 1}^{2(\floor{\tfrac{s}{2}}+1)} \{T_i, T'_{i}\}}
    \]
    is a greedy tree-packing in standard position with the claimed level profile.
\end{proof}
We can now use this lemma to extend the packing.
\begin{lemma} \label{lma:buildPacking}
    Let $\T$ be a greedy tree-packing in standard position on $G_{n,k}$ with $|\T| = 2j$ and $\beta_{2k} \in \{0,1\}$. 
    Let $i$ and $s >0$ be such that 1) either $\beta_i = \beta_{i+1} = \dots = \beta_{i+s} > \beta_{i+s+1}$ or $\beta_i = \beta_{i+1} = \dots = \beta_{i+s}$ and $i+s = 2j$, and 2) either $\beta_{i-1} > \beta_{i}$ or $i = 1$.
    Then there is a greedy tree-packing $\T'$ in standard position on $G_{n,k}$ with $2j + 2(2k+1)$ trees such that the level profile $\beta'$ of $\T'$ satisfies $\beta'_l = \beta_l + [i = l]$ for all $l$.
\end{lemma}
\begin{proof} 
    We will show the lemma assuming that $\beta_i = \beta_{i+1} = \dots = \beta_{i+s} > \beta_{i+s+1}$ and that $\beta_{i-1} > \beta_{i}$. The other cases follow from identical arguments. 

    We begin by letting $t = i$, $s' = s$. Then we apply \Cref{lma:OneAugment} with $t$ and $s'$ as arguments to get a new packing with only $2j + 2(\floor{\tfrac{s'}{2}}+1)$ trees. We then let $ t = t+s'$ and $s'$ be the smallest non-negative integer such that $\text{lev}(Y_{t})>\text{lev}(Y_{t+s'})$. Again we apply \Cref{lma:OneAugment} with $t$ and $s'$ as arguments. We do this recursively, until $t+s' = 2k$.

    At this point, we have a tree-packing $\T'$ of size at most $2j+4k$. Indeed, in the worst-case $s' = 0$ in every iteration, leaving us with at most $2k$ recursive calls. 

    If the initial level of $Y_{2k}$ was $1$, we pack $T$ with $e^{2k}_2$ and $f^{\iota}_{1}$ swapped. Then we pack $T'$, but with $e^{2k}_{1}$ and $e^{2k}_2$ swapped. 
    Note here that $T$ and $T'$ should be constructed with respect to the levels of $\T'$ and not $\T$.

    If the initial level of $Y_{2k}$ was $0$, we pack $T$ with $e^{2k}_1$ and $f^{\iota}_{1}$ swapped. Then we pack $T'$, but with $e^{2k}_{1}$ and $e^{2k}_2$ swapped. Again $T$ and $T'$ should be constructed with respect to the levels of $\T'$ and not $\T$.
    
    Finally, we can, if necessary, pad with un-altered copies of $T$ and $T'$ (based on the levels of the final tree-packing above) to achieve a greedy tree-packing of the form
    \[
        \T \cup \paren{\bigcup_{l = 1}^{2(2k+1)} \{T_{l}, T'_{l}\}} 
    \]
   in standard position with $\iota$ one larger than before.  
   Observe that the ultimate tree-packing has level profile 
   \[
   \hat{\beta}_l = \beta_{l} + [l = 2k] + \sum_{(t,s')} [l = t] - [l = t+s'] = \beta_{l} + [l = i] + \sum_{t} -[l = t] + [l = t] = \beta_{l} + [l = i],
   \]
   as claimed. 
\end{proof}

To obtain the final tree-packing of $G_{n,k}$ we do as follows. 
Beginning from the empty packing, which is in standard position, we repeatedly apply \Cref{lma:buildPacking}. The goal is to achieve a level profile of the form $(2k-1, 2k-2, \dots, 0)$. 
To do so, assume that we have constructed a level profile of the form $(i,i, \dots, i, i-1, i-2, \dots, 0)$. 
We then apply \Cref{lma:buildPacking} on $j$ beginning with $j=1$, then $j=2$ and so on up to and including $j=2k-i$ to increase the $j^{th}$ coordinate to $i+1$. 

Since we apply \Cref{lma:buildPacking} $\sum \limits_{j = 1}^{2k} j = k(2k+1)$ times, and each application extends the tree-packing with $2(2k+1)$ trees, the resulting tree-packing $\T$ contains  $2k(2k+1)^2$ trees. 
Since the tree-packing is in standard position, we observe that some edge $e$ has been packed $\tfrac{|\T|}{2}+k$ times. 
By \Cref{lma:disjSPT} we have that $\ell^{*}(e) = \tfrac{1}{2}$, so if 
\[
\frac{k}{|\T|} = \frac{k}{2k(2k+1)^2} = |\ell^{\T}(e) - \ell^{*}(e)| \leq \frac{\eps}{2}
\]
we have that 
\[
    \eps^{-1} \leq (2k+1)^2
\]
i.e., $k = \Omega(\eps^{-\tfrac{1}{2}})$. 
In particular, the lemma now follows for $\lambda = 2$. Indeed, the above construction with specific $k = O(\eps^{-\tfrac{1}{2}})$ yields a tree-packing with $\Omega(\lambda \cdot{} \eps^{-3/2})$ trees that does not achieve the required concentration. 
Note that this only holds for large enough $\eps$ with $\eps^{-1} \in O(n^{1/3})$, since otherwise $\iota$ might become too big for the argument in \Cref{lma:buildPacking} to go through. 

To generalize the statement to any $\lambda \in 2\mathbb{Z}$, let $\lambda = 2s$. 
Then we get $G_{s,n,k}$ by duplicating every edge of $G_{n,k}$ $s$ times. 
We get a greedy tree-packing $\T_s$ by replacing each tree $T$ in the tree-packing $\T$ on $G_{n,k}$ by $s$ parallel and isomorphic copies of $T$, each using their own set of edges. 

In total $|\T_s| = 2k(2k+1)^2\cdot{}s$, and so the calculation now becomes:
\[
    \frac{k}{|\T|} = \frac{k}{2k(2k+1)^2\cdot{}s} = |\ell^{\T}(e) - \ell^{*}(e)| \leq \frac{\eps}{2s}
\]
and again we conclude that 
\[
    \eps^{-1} \leq (2k+1)^2 
\]
 and obtain the theorem exactly as before. 

%% file: existence.tex
\section{Existence of Small Tree-Packing}\label{sc:existence}
The goal of this section is to show that for all graphs $G$ there exists a tree-packing that approximates the ideal load decomposition well. 
Formally, we will prove the following. 

\TPexistence*
We restate the theorem slightly. We recall that $\lambda/2 < \Phi \leq \lambda$ by \Cref{lm:phi_lambda}. Hence the statement is equivalent to $|\T|=\Theta(\Phi/\eps)$ trees guaranteeing 
\begin{equation*}
        |\ell^\T(e)-\ell^*(e)| \leq \eps/\Phi,
\end{equation*}
for all $e\in E$.

To show this, we first inspect the easier case, where the trivial partition $\mathcal{P} = \{v_{i}\}_{i = 1}^{n}$ achieves the minimum partition value $\Phi_{G}$ of $G$. 

To do so, we will generalize Kaiser's simple proof of the tree-packing theorem~\cite{Kaiser12}, to packing trees plus one forest. The proof is very similar to the proof in~\cite{Kaiser12}, we include it for completeness. We start by introducing some notation. 

Let $k\geq 1$. A \emph{$k$-decomposition} $\T$ of a graph\footnote{By abuse of notation, we use $\T$ both for a $k$-decomposition and a tree-packing. In the end, this will correspond to essentially the same packing.}  is a $k$-tuple of spanning subgraphs such that $\{E(T_i) : 1\leq i\leq k\}$ is a partition of $E(G)$. 

We define the sequence $(\P_0,\P_1, \dots, \P_\infty)$ of partitions of $V(G)$ associated with $\T$ as follows. First $\P_0=\{V(G)\}$. For $i\geq 0$, if there exists $c\in \{1, \dots, k\}$ such that the induced subgraph $T_c[X]$ is disconnected for some $X\in \P_i$, then let $c_i$ be the least such $c$, and let $\P_{i+1}$ consist of the vertex sets of all components of $T_{c_i}[X]$, where $X$ ranges over all the classes of $\P_i$. Otherwise, the process ends by setting $\P_\infty=\P_i$, and we set $c_j=k+1$ and $\P_j=\P_i$ for all $j\ge i$. 

The \emph{level}, $\lev(e)$, of an edge $e\in E(G)$ (w.r.t.\ $\T$) is defined as the largest $i$ (possibly $\infty$) such that both endpoints of $e$ are contained in one class in $\P_i$. 

When $\P$ and $\Q$ are partitions of $V(G)$, we say that $\P$ \emph{refines} $\Q$, denoted by $\P \leq Q$, if every class of $\P$ is a subset of a class of $\Q$. 

Finally, we define a strict partial order on $k$-decompositions. Given two $k$-decompositions $\T$ and $\T'$, we set $\T\prec \T'$ if there is some finite $j\geq 0$ such that both of the following hold
\begin{itemize}
    \item for $0\le i< j$, $\P_i=\P_i'$ and $c_i=c_i'$\footnote{Here we use $\P_i'$ and $c_i'$ to denote the partitions/values corresponding to $\T'$.},
    \item either $\P_j < \P_j'$ or $\P_j=\P_j'$ and $c_j< c_j'$.
\end{itemize}

\begin{lemma}\label{lm:Kaiser_adjusted}
    Let $G$ be a graph on vertex set $V(G) = \{v_1, v_2, \dots, v_n\}$ such that the trivial partition $\mathcal{P} = \{v_{i}\}_{i = 1}^{n}$ achieves the minimum partition value $\Phi_{G}$ of $G$. Then, there exists a disjoint packing with $\lfloor \Phi_{G}\rfloor$ spanning trees and one forest $F$ on exactly $(\Phi_{G}-\lfloor \Phi_{G}\rfloor)(n-1)$ edges. 
\end{lemma}
\begin{proof}
    The idea is the following: pick a $k$-decomposition that first of all contains $\lfloor \Phi_{G}\rfloor$ disjoint spanning trees and (at most) one disjoint spanning subgraph $F$ on $(\Phi_{G}-\lfloor \Phi_{G}\rfloor)(n-1)$ edges, that subject to these constraints maximize the partial order $\prec$. Note that this is a $k$-decomposition for $k=\ceil{\Phi_G}$, and that there is only a forest $F$ iff $\Phi_G$ is non-integer. If $\Phi_G$ is integer the statement is exactly the tree-packing theorem, so follows by e.g.~\cite{Kaiser12}. So we assume it is not. In that case $F=T_k$.
    
    The claim now is that $\P = \P_{\infty}$ has at least $(\Phi_{G}-\lfloor \Phi_{G}\rfloor)(n-1)$ non-parallel, inter-partition edges. If $F$ is a forest, we are done.

    We first argue that if $F$ is not a forest, then $F$ must contain a cycle in $G/\P$. To see this, consider $G/\P$. Since $T_i$ and $F$ all induce trees inside each partition of $\P$, we conclude that $F$ must have at least $(\Phi_{G}-\lfloor \Phi_{G}\rfloor) (|\P|-1)$ edges in $G/\P$. $F$ contains a tree on each partition $P\in \P$ so we have at least $n-|\P|$ edges inside the partitions. 
    Now we see that
    \begin{equation*}
        (\Phi_G-\lfloor\Phi_G \rfloor)(n-1) = |E(F)| \geq (\Phi_{G}-\lfloor \Phi_{G}\rfloor) (|\P|-1)+ (n-|\P|).
    \end{equation*}
    Rearranging gives 
    \begin{equation*}
        (\Phi_G-\lfloor\Phi_G \rfloor)(n-|\P|) \geq n-|\P|,
    \end{equation*}
    and thus for $|\P|\neq n$, $\Phi_G-\lfloor\Phi_G \rfloor\geq 1$, a contradiction. We conclude that $|\P|=n$, and thus any cycle in $F$ is a cycle in $F/\P$. 
    
    Let $e$ be an edge in a cycle of $F=T_k$ of minimum level, and set $m=\lev(e)$. Let $P$ be the class of $\P_m$ containing both endpoints of $e$. Since $e$ joins different components of $T_{c_m}[P]$, we have $c_m\neq k$, and the unique cycle $C$ in $T_{c_m}+ e$ contains an edge with only one endpoint in $P$. Thus for some edge $e'\in C$ we have $\lev(e')< m$. Let $e'$ be such an edge of lowest level. Let $Q$ be the class of $\P_{\lev(e')}$ containing both endpoints of $e'$. Observe that $V(C)\subseteq Q$. We now create a new $k$-decomposition with $e$ and $e'$ swapped: let $\T'$ be the $k$-decomposition obtained from $\T$ by replacing $T_{c_m}$ with $T_{c_m}+e-e'$ and $T_k$ with $T_k-e+e'$. 

    Next, it is easy to check that $\T \prec \T'$, contradicting that $\T$ was maximal. First, we show that for $i\leq m$ we have $\P_i'=\P_i$ and $c_i'=c_i$. We do this by induction. For $i=0$ we have $\P_0'=\{V(G)\}=\P_0$ and $c_0'=k=c_0$, so the base case is clear. Now we assume that the statement holds for $0\leq i<m$, and we prove it for $i+1$. 

    Let $S$ be an arbitrary class of $\P_{i+1}$, by definition, $T_{c_i}[S]$ is connected. We want to show that $T'_{c_i'}[S]$ is also connected. By the induction hypothesis, $c_i'=c_i$. If $c_i\notin \{c_m,k\}$, then $T'_{c_i'}[S]=T_{c_i'}[S]=T_{c_i}[S]$, so $T'_{c_i'}[S]$ is connected. We have $i<m$, so we are only left with the case $c_i=k$. We have that $E(T_k)-E(T_k')=e$, so if not both endpoints of $e$ are in $S$, then we have $T_k'[S]$ is connected as well. If $S$ does contain both endpoints of $e$, then $P\subseteq S$, because every class of $\P_{i+1}$ containing both endpoints of $e$ contains $P$. Hence $T_k'[S]$ is connected. We see that $\P_{i+1} \leq \P_{i+1}'$. By maximality of $\T$, we conclude that $\P_{i+1} = \P_{i+1}'$.

    Next, we show that $c_{i+1}'=c_{i+1}$. Let $R\in \P_{i+1}$ and $c<c_{i+1}$. Since $\P_{i+1} = \P_{i+1}'$, we have $R\in \P_{i+1}'$. By definition of $c_{i+1}$, $T_c[R]$ is connected. Similar as before, we can argue that $T_c'[R]$ is connected. Hence, $c_{i+1}'\geq c_{i+1}$. Again by maximality of $\T$, we get that we must have $c_{i+1}'=c_{i+1}$.

    Now we look at the next step. We have from the above that $\P_m' = \P_m$ and $c_m'=c_m$, so the classes of $\P_{m+1}'$ are the vertex set of components of $T_{c_m}'[U]$ for $U\in \P_m$. For $U\in \P_m-\{P\}$, we have $T_{c_m}'[U]=T_{c_m}[U]$, so their components coincide. The graph $T_{c_m}'[P]$ equals $T_{c_m}[P]$ with the extra edge $e$ connecting two components of $T_{c_m}[P]$. Hence $\P_{m+1}< \P_{m+1}'$, so also $\T\prec \T'$, contradicting $\T$ being maximal. 
\end{proof}

Now the result follows easily on graphs with uniform $\Phi_G$.
\begin{lemma}
     Let $G$ be an unweighted, undirected (multi-)graph, where the trivial partition $\mathcal{P} = \{v_{i}\}_{i = 1}^{n}$ achieves the minimum partition value $\Phi_{G}$ of $G$.
    There exists a tree-packing $\T$ that needs $|\T|=\Theta(\Phi_G/\eps)$ trees to satisfy
    \begin{equation*}
        |\ell^\T(e)-\ell^*(e)| \leq \eps/\Phi_G,
    \end{equation*}
    for all $e\in E$.
\end{lemma}
\begin{proof}
    W.l.o.g., assume that $1/\eps$ is an integer. We replace every edge in $G$ by $1/\eps$ copies to obtain $G'$. Note that $\ell^*_{G'}(e)=\eps\ell^*_G(e)$ and $\Phi_{G'}=\Phi_G/\eps$. By \Cref{lm:Kaiser_adjusted}, we see that we can pack $\lfloor \Phi_{G'}\rfloor$ spanning trees and one forest on exactly $(\Phi_{G'}-\lfloor \Phi_{G'}\rfloor)(n-1)$ edges. 
    Now let $\T$ consists of these $\lfloor \Phi_{G'}\rfloor$ spanning trees, together with one tree that is the forest extended to a tree in an arbitrary way. Now we have that $L^\T(e)=1/\eps$ or $L^T(e)=1/\eps+1$. We see that in the first case that
    \begin{align*}
        |\ell^\T(e)-\ell^*(e)| &= \left| \frac{1/\eps}{\lceil \Phi_G/\eps \rceil}-\frac{1}{\Phi_G} \right|= \frac{1}{\Phi_G} -\frac{1/\eps}{\lceil \Phi_G/\eps \rceil}= \frac{\lceil \Phi_G/\eps \rceil-\Phi_G/\eps}{\lceil \Phi_G/\eps \rceil\Phi_G}\\
        &\leq \frac{1}{\lceil \Phi_G/\eps \rceil\Phi_G} \leq \eps/\Phi_G.
    \end{align*}
    And in the second case we have 
    \begin{align*}
        |\ell^\T(e)-\ell^*(e)| &= \left| \frac{1/\eps+1}{\lceil \Phi_G/\eps \rceil}-\frac{1}{\Phi_G} \right|=  \left|\frac{\Phi_G/\eps+\Phi_G-\lceil \Phi_G/\eps \rceil}{\lceil \Phi_G/\eps \rceil\Phi_G}\right|\\
        &\leq \left|\frac{\Phi_G/\eps-\lceil \Phi_G/\eps \rceil}{\lceil \Phi_G/\eps \rceil\Phi_G}\right|+\left|\frac{\Phi_G}{\lceil \Phi_G/\eps \rceil\Phi_G}\right|    
        \leq 2\eps/\Phi_G,
    \end{align*}
    where for the first inequality we use the triangle inequality, and for the second inequality we use the first case. Now by setting $\eps\leftarrow\eps/2$, we obtain the result. 
\end{proof}

At last, we are ready to show the general case.
\begin{proof}[Proof of \Cref{thm:existence}]
We recall the definition of the ideal relative loads: $\ell^*(e)$ is defined recursively as follows.  
\begin{enumerate}
    \item Let $\P^*$ be a packing with $\rm{pack\_val}(\P^*)=\Phi$.
    \item For all $e\in E(G/\P)$, set $\ell^*(e):= 1/\Phi$.
    \item For each $S\in \P^*$, recurse on the subgraph $G[S]$. 
\end{enumerate}

We prove the lemma by induction on the depth of this recursive definition. If the depth is $0$ then the minimum partition value $\Phi_{G}$ is exactly achieved by the trivial partition $\mathcal{P} = \{v_{i}\}_{i = 1}^{n}$, so the base case is immediate. Suppose it holds up to depth $i$. Then we have for each component $X\in \P^*$ that there exists a tree-packing $\T_X$ such that $|\T_X|=\Theta(\Phi/\eps)$ and 
\begin{equation*}
        |\ell^{\T_X}(e)-\ell^*(e)| \leq \eps/\Phi,
\end{equation*}
for all $e\in E(G[X])$. Note that this follows from the induction hypothesis since $\Phi_{G[X]}\geq \Phi$, so we can set $\eps\leftarrow \tfrac{\Phi_{G[X]}}{\Phi}\eps$. We also have a tree-packing $\T_{\P*}$ on $G/\P^*$ such that $|\T_{\P^*}|=\Theta(\Phi/\eps)$ and 
\begin{equation*}
        |\ell^{\T_{\P^*}}(e)-\ell^*(e)| \leq \eps/\Phi,
\end{equation*}
for all $e\in E(G/\P^*)$.
Clearly $\T=\T_{\P^*}\cup \left(\bigcup_{X\in \P^*} \T_X\right)$, where each tree is a union of the respective trees, is a tree-packing of $G$ that satisfies the bounds on $\ell^\T$. 
\end{proof}